\documentclass{amsart}
\usepackage[utf8]{inputenc}
\usepackage[english]{babel}
\numberwithin{equation}{section}

\usepackage{amsrefs}

\usepackage{amssymb,amsthm,amsfonts,amstext}
\usepackage{amsmath}
\usepackage{bbold}
\usepackage{enumerate}

\usepackage{mathtools}

\mathtoolsset{showonlyrefs}

\newtheorem{thm}{Theorem}[section]
\newtheorem{lem}[thm]{Lemma}
\newtheorem{cor}[thm]{Corollary}

\newtheorem{prop}[thm]{Proposition}

\newtheorem{definition}[thm]{Definition}

\newtheorem{rem}[thm]{Remark}

\newcommand{\mc}[1]{{\mathcal #1}}
\newcommand{\mf}[1]{{\mathfrak #1}}

\newcommand{\bb}[1]{{\mathbb #1}}

\newcommand\EE{{\mathbb E}}
\newcommand\PP{{\mathbb P}}
\newcommand\NN{{\mathbb N}}
\newcommand\RR{{\mathbb R}}

\newcommand\ZZ{{\mathbb Z}}

\newcommand{\llangle}{\langle\langle}
\newcommand{\rrangle}{\rangle\rangle}

\usepackage{tikz}
\usepackage{tikz-cd}
\usetikzlibrary{backgrounds}
\usetikzlibrary{patterns,fadings}
\usetikzlibrary{arrows,decorations.pathmorphing}
\usetikzlibrary{calc}
\definecolor{light-gray}{gray}{0.95}
\usepackage{graphicx}
\usepackage{epsfig}
\usetikzlibrary{shapes}
\usetikzlibrary{decorations}

\usetikzlibrary{decorations.pathmorphing}
\usetikzlibrary{decorations.pathreplacing}
\usetikzlibrary{decorations.shapes}
\usetikzlibrary{decorations.text}
\usetikzlibrary{decorations.markings}
\usetikzlibrary{decorations.fractals}
\usetikzlibrary{decorations.footprints}

\definecolor{brightcerulean}{rgb}{0.11, 0.67, 0.84}
\definecolor{cerulean}{rgb}{0.0, 0.48, 0.65}
\definecolor{Gray}{rgb}{0.5, 0.5, 0.5}

\definecolor{darkkgreen}{rgb}{0.0, 0.5, 0.0}
\definecolor{columbiablue2}{rgb}{0.41, 0.37, 2.0}

\begin{document}

\title[Equilibrium fluctuations for long jumps exclusion]{Equilibrium fluctuations for diffusive symmetric exclusion with long jumps and infinitely extended reservoirs}


\author{C.Bernardin}
\address{C\'edric Bernardin, Universit\'e C\^ote d'Azur, CNRS, LJAD\\
	Parc Valrose\\
	06108 NICE Cedex 02, France}
\email{{\tt cbernard@unice.fr}}

\author{P.  Gon\c calves}
\address{Patr\'icia Gon\c calves, Center for Mathematical Analysis,  Geometry and Dynamical Systems,
	Instituto Superior T\'ecnico, Universidade de Lisboa,
	Av. Rovisco Pais, 1049-001 Lisboa, Portugal}
\email{{\tt pgoncalves@tecnico.ulisboa.pt}}

\author{M. Jara}
\address{Milton Jara,  IMPA, Estrada Dona Castorina 110, CEP 22460 Rio de
	Janeiro, Brasil
}
\email{{\tt mjara@impa.br}}

\author{S. Scotta}
\address{Stefano Scotta, Center for Mathematical Analysis,  Geometry and Dynamical Systems,
	Instituto Superior T\'ecnico, Universidade de Lisboa,
	Av. Rovisco Pais, 1049-001 Lisboa, Portugal}
\email{{\tt stefano.scotta@tecnico.ulisboa.pt}}

\begin{abstract}
	We provide a complete description of the equilibrium fluctuations for diffusive symmetric exclusion processes with long jumps in contact with infinitely extended reservoirs and prove that they behave as generalized Ornstein-Uhlenbeck processes with various boundary conditions, depending mainly on the strength of the reservoirs. On the way, we also give a general statement about uniqueness of the Ornstein-Uhlenbeck process originated by the microscopic dynamics of the underlying interacting particle systems and adapt it to our study.
\end{abstract}

\medskip

\renewcommand{\subjclassname}{%
  \textup{2010} Mathematics Subject Classification}

\maketitle

\tableofcontents

\section{Introduction}

%
%

The aim of this work is to study the equilibrium fluctuations of a symmetric exclusion process with long jumps and infinitely extended reservoirs. The latter is a well known stochastic interacting particle system whose Markovian dynamics can be summarized as follows. Each point of a one dimensional sub lattice $\Lambda_N=\{1,\ldots, N-1\}$ of $\mathbb Z$ of size $N$, called ``bulk'', can be occupied by at most one particle (exclusion rule). Any particle evolves like a continuous time symmetric random walk with a translation invariant transition probability $p(\cdot)$ given by
\begin{equation}
\label{eq:choice_p}
p(x,y):=p(x-y)=\mathbb{1}_{\{x\neq y\}}\frac{c_{\gamma}}{|x-y|^{\gamma+1}},
\end{equation}
apart from the fact that jumps violating the exclusion rule are suppressed. Above $c_\gamma$ is a normalizing constant. The previous bulk dynamics conserves the number of particles. The boundary of $\Lambda_N$  acts as infinitely extended particles reservoirs whose role is to maintain a certain density $\alpha$ on the left and $\beta$ on the right, destroying the previous conservation law and creating a net density current in the bulk if $\alpha \ne \beta$. The intensity $\kappa/N^\theta$ of the reservoirs is regulated by the two parameters $\kappa>0$ and $\theta \in \RR$. In particular, since $N$ is large, the greater the value of $\theta$, the weaker the action of the reservoirs. Therefore, for $\theta\geq 0$ we will use the terminology \textit{slow reservoirs} while for $\theta<0$ we will use \textit{fast reservoirs}. The main interest in this model is to understand the typical macroscopic space-time evolution of the empirical density (hydrodynamic limits) as well as its fluctuations around this typical behavior (see \cite{KL}, \cite{Spohn}). This behavior strongly depends on the fact that $p(\cdot)$ has finite variance, in which case a diffusive behavior is observed, or infinite variance, in which case a superdiffusive fractional behavior appears. The interesting feature of the presence of reservoirs, even in the diffusive case, is that they introduce a bunch of boundary conditions to the hydrodynamic equation. We refer for example to \cite{OLV, Derrida, LMO, Kipnis, changlandimolla, patricianote, FGN_Robin, PTA} for various research papers on the subject.\\    
%

The hydrodynamic limit of the long jumps exclusion process with an infinite variance $p(\cdot)$ and without reservoirs has been studied in \cite{MJ}. In \cite{BJGO2, BJ, MY} infinitely extended reservoirs are added to the system and both the hydrodynamic limit and hydrostatics are obtained. The hydrodynamic equations are given, in the case of infinite variance, by a  collection of fractional reaction-diffusion equations with Dirichlet, fractional Robin or fractional Neumann boundary conditions. In \cite{BGJO}, the case of finite variance is considered, extending the work of \cite{Adriana}, which was limited to a nearest-neighbor transition probability, non-extended reservoirs and only related to the slow ($\theta\ge 0$) regime. The hydrodynamic equations are given by a  collection of reaction-diffusion equations with Dirichlet, Robin or Neumann boundary conditions. The equilibrium fluctuations, i.e. when the densities of the two reservoirs are the same $\alpha=\beta=\rho$ and starting the system from the stationary measure, for these models have been studied \cite{PTA} but only in the context of \cite{Adriana}. In particular, also here the interesting fast regime ($\theta<0$) which gives rise to singular reaction terms at the hydrodynamic level is not considered.\\

In this work, we derive the equilibrium fluctuations for the model presented in \cite{BGJO}, i.e. in the case where the transition  probability $p(\cdot)$ has infinite  support but finite variance. In \cite{PTA}  the  authors obtain that the equilibrium fluctuations are given by Ornstein-Uhlenbeck processes with either Dirichlet, Robin or Neumann boundary conditions, that are solutions of linear stochastic partial differential equations (SPDEs) of the form 
\begin{equation}
\label{eq:OU-th}
\partial_t \mathcal{Y}_t= \mathcal{A} \mathcal{Y}_t +\mathcal{B} {\dot {\mc W}}_t
\end{equation}
where $\mathcal A$ and $\mathcal{B}$ are two unbounded operators depending on the underlying hydrodynamic equations and ${\dot {\mc W}}_t$ is a space-time white noise. The mathematical study of these equations was pioneered  in \cite{holley1978generalized} and their solutions are called generalized Ornstein-Uhlenbeck processes. In our case, the process describing the fluctuations around equilibrium is also a generalized Ornstein-Uhlenbeck process. In the regime $\theta \ge 1$ the results and proofs obtained here are very similar to those of \cite{PTA}. The main novelty is in the regime $\theta<1$, which corresponds to the case where hydrodynamic equations are given by a diffusion ($2-\gamma<\theta<1$), a singular reaction-diffusion ($\theta=2-\gamma$) or a singular reaction equation ($\theta<2-\gamma$) with Dirichlet boundary conditions. In the case $2-\gamma<\theta<1$ the proof of the uniqueness of the solution to the SPDE is non-standard and requires new technical tools. The difficulty comes from the fact that the martingale problem defining the corresponding Ornstein-Uhlenbeck process is provided for test functions which vanish together with all their derivatives at the boundary, while for proving uniqueness, we need to extract information about the boundary which is not seen by functions of this type. In the case $\theta=2-\gamma$, the proof of uniqueness follows the standard approach but it requires a very precise study of a singular Sturm-Liouville problem, which has, in fact, its own interest, independently from the content of this paper.\\ 

In this article, we close the scenario of the equilibrium fluctuations, initiated in \cite{PTA} for the case of nearest-neighbor jumps, and we let open the case where the variance of $p(\cdot)$ is infinite. We think however that the study initiated here gives some insight to attack the difficult problem of the equilibrium fluctuations when the variance of $p(\cdot)$ is infinite. We also restricted ourselves to the case when the system is at equilibrium and a very challenging problem is to analyze the case where the system is initially outside of equilibrium. We remark that in the case of nearest-neighbor jumps, the non-equilibrium fluctuations were derived in \cite{FGN_Robin} in the case of Robin boundary conditions (corresponding to $\theta=1$) and in \cite{GJMN} for the case of Dirichlet (in the case of slow boundary, that is $0\leq \theta<1$) and Neumann boundary conditions (that is $\theta>1$).\\

We present now the structure of the paper. In Section \ref{s2} we describe in details the model that we study. In the same section (Subsection \ref{s3}) we also recall the results on the hydrodynamic limit obtained in \cite{BGJO}. Then, in Subsection \ref{s4} we start the analysis of the equilibrium fluctuations for this model, introducing the fluctuation field, the space of test functions, and then we state our main results. In Section \ref{sec_car} we give a characterization of the limit points of the sequence of fluctuation fields. In Section \ref{s6} we give a proof of tightness of the sequence of probability measures associated with the fluctuation fields in a suitably chosen topology. Finally, in Section \ref{s7} we prove that the limit of this probability measures is the unique martingale solution of an SPDE. In Section \ref{s8} we prove some technical lemmas that we used throughout the paper. Section \ref{Subsec:UOU} provides a general proof of uniqueness of the Ornstein-Uhlenbeck process solution of \eqref{eq:OU-th}. The paper is concluded by some appendices.

\section{The model and statement of results}
\label{s2}

\subsection{The model}
In this subsection we recall the microscopical model introduced in \cite{BGJO}, which consists in a symmetric exclusion process with infinitely extended reservoirs. 

Let $N>1$ be an integer. We call ``bulk'' the set $\Lambda_N=\{1,\cdots, N-1\}$. The exclusion process that we consider is a Markov process   $\{\eta_t\}_{t\geq 0}$ with state space $\Omega_N=\{0,1\}^{\Lambda_N}$. As usual, for a configuration $\eta\in \Omega_N$ we say that the site $x \in \Lambda_N$ is occupied if $\eta(x)=1$ and it is empty if $\eta(x)=0$. Note that it is not possible to have more than one particle per site (exclusion rule). Let us consider the symmetric transition probability $p:\mathbb{Z}\times \bb Z\rightarrow{[0,1]}$ defined in \eqref{eq:choice_p}. This probability characterizes the rate at which a particle jumps from $x$ to $y$ or from $y$ to $x$. We assume that $p(\cdot)$ has a finite variance denoted  by $\sigma^{2}:=\sum_{z\in \ZZ}z^{2}p(z)<\infty$, which, for this choice of $p(\cdot)$, means that  $\gamma>2$. Let us also introduce the quantity 
$m=\sum_{z\ge 1} z p(z)$,
which will be useful through this work. The action of the stochastic reservoirs,  placed at the left and  at the right of the bulk, depends on the four fixed parameters $\alpha, \beta\in[0,1]$, $\kappa>0$ and   $\theta \in \mathbb R$.  With this notation we can define  the  dynamics of the process. It is easier to explain separately the dynamics involving only the bulk and the ones involving the reservoirs.
\begin{itemize}
	\item \textbf{Bulk dynamics:} each couple of sites in the bulk $(x,y) \in \Lambda_N\times \Lambda_N$ carries an independent Poisson process of intensity $1$. If for the configuration $\eta$, there is an occurrence in the Poisson process associated to the couple  $(x,y)$, then we exchange the value of the occupation variables $\eta(x)$ and $\eta(y)$ with probability $p(y-x) /2$.
	\item \textbf{Reservoirs dynamics:} each couple of sites $(z,x)$ with  $x\in\Lambda_N$ and $z \leq 0$ carries 	a Poisson process of intensity $\kappa N^{-\theta}p(x-z) [(1-\alpha) \eta(x)  + \alpha (1-\eta(x)) ]$, all of them being independent. The value of $\alpha \in [0,1]$ represents the particle density of the reservoir located in $z$. If for the configuration $\eta$, there is an occurrence in the Poisson process associated to the couple $(z,x)$ , then we change the value $\eta(x)$ into $1-\eta(x)$ with probability $1$. At the right boundaries ($z \geq N$) the dynamics is similar but instead of $\alpha$ the density of particles of the reservoirs is given by $\beta$.
\end{itemize}
See Figure \ref{fig1} for an illustration of the dynamics described above.
\begin{figure}[htbp]
	\begin{center}
		\begin{tikzpicture}[thick]

		\draw[shift={(-5.01,-0.15)}, color=black] (0pt,0pt) -- (0pt,0pt) node[below]{\dots};
		\draw[shift={(5.01,-0.15)}, color=black] (0pt,0pt) -- (0pt,0pt) node[below]{\dots};
		
		\fill [color=blue!45] (-5.3,-0.6) rectangle (-3.3,0.6);
		\fill [color=blue!45] (5.3,-0.6) rectangle (3.3,0.6);
		
		\draw[-latex] (-5.3,0) -- (5.3,0) ;
		\draw[latex-] (-5.3,0) -- (5.3,0) ;
		\foreach \x in {-4.5,-4,-3.5,...,4.5}
		\pgfmathsetmacro\result{\x*2+7}
		\draw[shift={(\x,0)},color=black] (0pt,0pt) -- (0pt,-2pt) node[below]{\scriptsize \pgfmathprintnumber{\result}};
		
		\node[ball color=black, shape=circle, minimum size=0.3cm] (B) at (-1.5,0.15) {};
		
		\node[ball color=black, shape=circle, minimum size=0.3cm] (C) at (1.5,0.15) {};
		
		\node[ball color=black, shape=circle, minimum size=0.3cm] (D) at (2.5,0.15) {};
		
		\node[ball color=black, shape=circle, minimum size=0.3cm] (E) at (-2,0.15) {};
		
		\node[draw=none] (S) at (3.5,0.15) {};
		\node[draw=none] (R) at (0,0.15) {};
		\node[draw=none] (L) at (-4,0.15) {};
		\node[draw=none] (M) at (-1,0.15) {};

		\path [<-] (S) edge[bend right =70, color=blue]node[above] {\footnotesize $(1-\beta)\dfrac{\kappa}{N^{\theta}}p(2)$}(D);
		\path [->] (C) edge[bend right =70, color=blue]node[above] {\footnotesize $\dfrac{p(3)}{2}$}(R);			
		\path [<-] (M) edge[bend right =70, color=blue]node[above] {\footnotesize $\alpha\dfrac{\kappa}{N^{\theta}}p(6)$}(L);

		\end{tikzpicture}
		\caption{Example of the dynamics of the model with $N=14$, the sites colored in blue act as reservoirs.}	
		\label{fig1}
	\end{center}	
\end{figure}
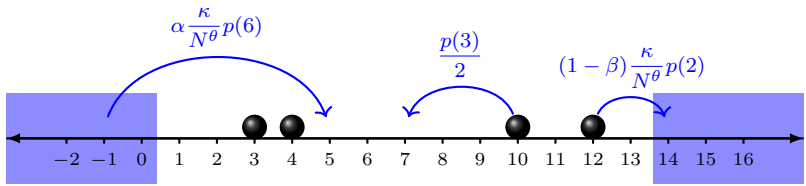

The process is completely characterized by its infinitesimal generator which is given by  
\begin{equation}
\label{Generator}
L_{N} = L_{N}^0+ L_{N}^{l} + L_N^r,
\end{equation}
where  $L_{N}^0$, $L_{N}^l$ and $L_N^r$ act on functions  $f:\Omega_N \to \RR$  as
\begin{equation}
\label{Generators}
\begin{split}
(L_N^0f)(\eta)&=\frac{1}{2}\sum_{x,y \in \Lambda_N}p(x-y)(f(\sigma^{x,y}\eta)-f(\eta)),\\
(L_N^lf)(\eta)&=\frac{\kappa}{N^{\theta}}\sum_{\substack{x \in \Lambda_N \\ y\leq 0}}c_x(\eta, \alpha)p(x-y)f(\sigma^{x}\eta)-f(\eta))\\
(L_N^rf)(\eta)&=\frac{\kappa}{N^{\theta}}\sum_{\substack{x \in \Lambda_N \\ y\geq N}}c_x(\eta, \beta)p(x-y)(f(\sigma^{x}\eta)-f(\eta))
\end{split}
\end{equation}
where, for any $\delta \in [0,1]$ and $x, y\in \Lambda_N$, we define $c_x(\eta,\delta):=\delta (1-\eta(x))+(1-\delta)\eta(x)$ and we use the notation $\sigma^{x,y}\eta$ and $\sigma^x\eta$ for the configurations
\begin{equation}\label{tranformations}
(\sigma^{x,y}\eta)(z) := 
\begin{cases}
\eta(z), & \text{ if } \; z \ne x,y,\\
\eta(y), & \text{ if }\; z=x,\\
\eta(x), & \text{ if }\; z=y
\end{cases}
, \quad (\sigma^x\eta)(z):= 
\begin{cases}
\eta(z), & \text{ if }\; z \ne x,\\
1-\eta(x), & \text{ if }\; z=x.
\end{cases}
\end{equation}

We speed up  the Markov process in the  time scale $t\Theta(N)$ and we denote it by $\eta_t^N=\eta_{t\Theta(N)}$. The time scale $\Theta(N)$ will be chosen later (see \eqref{time_scales}) and will depend on the value of $\theta$. Note that $\{ \eta_{t \Theta(N)} \,; \, t\ge 0\} $ has  infinitesimal generator given by $\Theta(N)L_{N}$.  

We note that if $\alpha=\beta=\rho \in [0,1]$, the Bernoulli product measures $\nu_\rho=\otimes_{x \in \Lambda_N} {\mc B} (\rho)$  are reversible and hence invariant for the dynamics (\textcolor{black}{see Lemma 2.4.1 of \cite{patricianote} for a proof).}

\subsection{Hydrodynamic limits}
\label{s3}
In this subsection we recall the results about hydrodynamic limits obtained in \cite{BGJO} for this model in the diffusive regime ($\gamma>2$). \\

Let us first introduce some notation. For any $T>0$ and $m,n \in \NN$ we will denote by $C^{m,n}([0,T]\times [0,1])$ the space of functions defined on $[0,T]\times [0,1]$ which are $m$ times continuously differentiable in the first variable and $n$ times on the second. We will consider also the space $C_c^{m,n}([0,T]\times [0,1])$ which is the space of the functions $H \in C^{m,n}([0,T]\times [0,1])$ such that, for  any fixed time $s \in [0,T]$, the function $H(s,\cdot)$ has compact support in $(0,1)$.  For some interval $I \subset \mathbb{R}$ and any Polish space $E$, we will consider the Skorokhod space {${D}(I,E)$ which is the space of functions from $I$ to $E$  that are right continuous and admit left limits.} Analogously we denote by $C(I,E)$ the space of continuous function from $I$ to any space $E$. Moreover, \textcolor{black}{for any $w:I\rightarrow [0,+\infty)$}, we will denote by $L^2_w (I)$ the space of real functions $H$ defined on $I$ for which the norm $$||H||^2_{L^2_w(I)}:=\int_I|H(u)|^2 w(u)du$$ is finite. This norm is \textcolor{black}{associated} to the inner product \textcolor{black}{$\langle \cdot, \cdot \rangle_{L^2_w (I)}$}, defined on real functions $H,G$ as 
$$\langle H, G\rangle_{L^2_w (I)}=\int_I H(u) G(u) w(u)du.$$ 
When $w$ is constantly equal to $1$ we will omit the index {$w$ in } $L^2_w (I)$. In addition, in order to have a lighter notation, if $I=[0,1]$ we will write only $L^2_w$ to denote this space.

For $k\ge 1$ we will denote by $\mathcal{H}^k(I)$ (resp. $\mathcal{H}_0^k (I)$) the classical Sobolev space $\mathcal{W}^{k,2}(I)$ defined as the closure of $C^k(I)$ (resp. $C_c^k(I)$) with respect to the norm defined by 
$$||H||^2_{\mathcal{H}^k(I)}:= \sum_{j=0}^{k} ||H^{(j)}||^2_{L^2(I)}.$$ 
{Above and below, we denote by $H^{(i)}$  the $i$--th derivative of $H$, while for $i=1$ (resp. $i=2$) we simply denote it by $\nabla H$, $H'$ or $\partial_u H$  (resp. $\Delta H$, $\partial_u^2 H$ or $H''$). }
If $I=[0,1]$ we will omit often the interval in the notation, i.e. $\mathcal{H}^k([0,1])=\mathcal{H}^k$ (resp. $\mathcal{H}_0^k= \mathcal{H}_0^k([0,1])$). We will also denote $\Vert H' \Vert_{L^2}$ by $\Vert H\Vert_1$.
Moreover, when we write that a function $H$ defined on $[0,T]\times I$ is in $L^2([0,T],\mathcal{H}^k(I))$, we mean that $$\int_0^T||H(s,\cdot)||_{\mathcal H^k(I)}^2ds<\infty.$$ 

\textcolor{black}{In order to have a lighter notation, for any two real functions $f,g$, we write $f(u) \lesssim g(u)$ if there exists a constant $C$, that does not depend of $u$, such that $f(u) \le C g(u)$ for every $u$.
}

\textcolor{black}{In order to define the hydrodynamic equations related to this model and the corresponding notion of weak solutions,} we have to introduce the following functions: for any $u \in (0,1)$,
\begin{equation}\label{V}
V_0(u):=\alpha r^-(u)+\beta r^+(u) \quad \text{ and } \quad \textcolor{black}{V_1(u):=r^-(u)+r^+(u),}
\end{equation}
\textcolor{black}{where $r^+(u):=c_{\gamma}\gamma^{-1}(1-u)^{-\gamma}$ and $r^-(u)=c_{\gamma}\gamma^{-1}u^{-\gamma}.$}\\

Now, we can write the family of hydrodynamic equations associated to this model, which depend on some parameters $\hat \sigma \geq 0, \hat \kappa\geq 0, \mathfrak a, \mathfrak b \geq 0$. Let $g:[0,1]\rightarrow [0,1]$ be a measurable function. We consider the following linear reaction-diffusion parabolic PDE with non-homogeneous boundary conditions:
\begin{equation}
\label{PDE}
\begin{cases}
&\partial_{t}\rho_{t}(q)=[\hat{\sigma}\Delta  - \hat \kappa  V_1]\,{\rho}_{t}(q)+ \hat \kappa V_0(q), \quad (t,q) \in [0,T]\times(0,1),\\
&\mathfrak a \partial_{q}\rho _{t}(0)=\mathfrak b  (\rho_{t}(0) -\alpha),\quad \mathfrak a \partial_{q} \rho_{t}(1)=\mathfrak b(\beta -\rho_{t}(1)),\quad t \in [0,T] \\
&{ \rho_0}(\cdot)= g(\cdot).
\end{cases}
\end{equation}

In \cite{BGJO} it is defined a notion of weak solution for this class of PDEs and it is shown there existence and uniqueness of such weak solutions.

\begin{rem}
	\label{neumann_cond_rem}
	Observe that in the previous definition, if $\hat \sigma >0$, $\hat \kappa=0$,  $\mf a=0$ and $\mf b\neq0$, then we recover the heat equation with Dirichlet boundary conditions. While, if $\mathfrak a,\mf b\neq 0$ then \eqref{PDE} is the heat equation with linear Robin  boundary conditions.  
	If $\mathfrak a\neq 0$ and $\mathfrak b =0$ then  \eqref{PDE} is the heat equation with Neumann boundary conditions.  If $\hat \sigma =0$ and $\hat \kappa\neq 0$, then \eqref{PDE} does not have the diffusion term. 
\end{rem}

%
%
%
%

The following theorem, proved in \cite{BGJO}, states the hydrodynamic limit (summarized also in Figure \ref{fig:hlf_lj}) relative to this model for the whole set of parameters, in the case of finite variance, that is $\gamma >2$. In order to state the theorem we need to introduce the following definition.

\begin{definition}
	Let $\rho_0: [0,1]\rightarrow[0,1]$ be a measurable function. We say that a sequence of probability measures $\lbrace\mu_{N}\rbrace_{N> 1 }$ on $\Omega_{N}$  is associated {with} the profile $\rho_{0}(\cdot)$ if for any continuous function $H:[0,1]\rightarrow \mathbb{R}$  and every $\delta > 0$ 
	\begin{equation*}
	\lim _{N\to\infty } \mu _{N}\bigg( \eta \in \Omega_{N} : \Big\vert \dfrac{1}{N}\sum_{x \in \Lambda_{N} }H\left(\tfrac{x}{N} \right)\eta(x) - \int_{0}^1H(q)\rho_{0}(q)dq \Big\vert    > \delta \bigg)= 0.
	\end{equation*}
\end{definition}

\begin{thm}
	\label{th:hl900}
	Let $g:[0,1]\rightarrow[0,1]$ be a measurable function and let $\lbrace\mu _{N}\rbrace_{N> 1}$ be a sequence of probability measures in $\Omega_{N}$ associated {with} $g(\cdot)$. Then, for any $0\leq t \leq T$ and $G \in C^1([0,1])$
	\begin{equation*}
	\label{limHidreform}
	\begin{split}
	&\lim _{N\to\infty } \PP_{\mu _{N}}\bigg( \eta_{\cdot}^{N} \in  D([0,T], {\Omega_{N}}) : \Big| \tfrac{1}{N-1}\sum_{x \in \Lambda_{N} }G\left(\tfrac{x}{N} \right)\eta^N_{t}(x) - \int_{0}^1G(q)\rho_{t}(q)dq \Big|   > \delta \bigg)\\
	&= 0,
	\end{split}
	\end{equation*}
	where $\mathbb{P}_{\mu_N}$ is the probability measure associated to the process $\{ \eta_t^N\, ;\, t\in [0,T] \} := \{\eta_{t\Theta(N)}\, ;\, t \in [0,T] \}$ with starting distribution $\mu_N$, the time scale is given by
	\begin{equation}\label{time_scales}
	\Theta(N)= \begin{cases}
	N^2, &\quad\textrm{if} \,\,\,  \theta\geq 2-\gamma,\\
	N^{\gamma+\theta}, &\quad   \textrm{if} 
	\,\,\, \theta<2-\gamma,\\
	\end{cases}
	\end{equation}
	and $\{\rho_{t}(\cdot)\, ; \, t\in [0,T]\}$ is the unique weak solution of \eqref{PDE}: 
	\begin{itemize}
		\item[$\bullet$] with $\hat \sigma=0$, {$\mathfrak a =0$} and $\hat \kappa=\kappa$, if $\theta<2-\gamma$;
		\item [$\bullet$] with $\hat \sigma=\tfrac{\sigma^2}{2}$,  {$\mathfrak a =0$} and $\hat \kappa=\kappa $, if $\theta =2-\gamma$; 
		\item[$\bullet$] with $\hat \sigma=\tfrac{\sigma^2}{2}$, {$\mathfrak a =0$} and $\hat \kappa=0$, if $\theta \in (2-\gamma,1)$;
		\item[$\bullet$] with $\hat \sigma=\tfrac{\sigma^2}{2}$, {$\mathfrak a =\frac{\sigma^2}{2}$}, $\hat \kappa =0$ and {$\mathfrak b=m\kappa $}, if $\theta =1$;
		\item[$\bullet$]  with $\hat \sigma=\tfrac{\sigma^2}{2}$, {$\mathfrak a =1$},  $\hat \kappa =0$ and {$\mathfrak b =0 $}, if $\theta >1$.
	\end{itemize}
\end{thm}


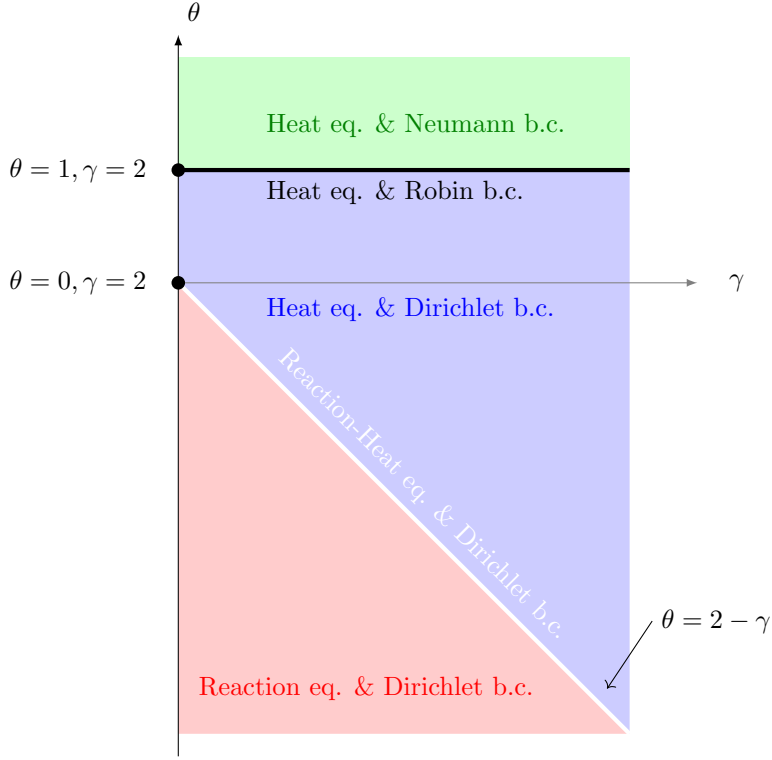
\begin{figure}[http]
	\begin{center}
		\begin{tikzpicture}[scale=0.30]
		\fill[green!20] (0,5) -- (0,10) --(20,10)--(20,5)-- cycle;
		\fill[blue!20] (0,5) -- (20,5) --(20,-20)--(0,0)-- cycle;
		\fill[red!20] (0,0) -- (20,-20) --(0,-20)-- cycle;
		\draw (0,12) node[right]{$\theta$};
		\draw (24,0) node[right]{$\gamma$};
		\draw (-1,0) node[left]{$\theta=0, \gamma=2$};
		\draw (-1,5) node[left]{$\theta=1, \gamma=2$};
		\draw[->,>=latex, gray] (0,0) -- (23,0);
		\draw[->,>=latex] (0,-21) -- (0,11);
		\draw[-,=latex,white,ultra thick] (0,0) -- (20, -20) node[midway, above, sloped] {{ Reaction-Heat eq.  \& Dirichlet b.c. }};
		\draw[-,=latex,black,ultra thick] (0,5) -- (20, 5);
		\node[right] at (3.5, 4){Heat eq. \& Robin b.c.};
		\node[right, darkkgreen] at (3.5,7) {Heat eq. \& Neumann b.c.} ;
		\node[right, blue] at (3.5,-1.2) {Heat eq. \&  Dirichlet b.c.} ;
		\node[right, red] at (0.5,-18) {Reaction eq. \& Dirichlet b.c.} ;
		\fill[black] (0,0) circle (0.3cm);
		\fill[black] (0,5) circle (0.3cm);
		\draw[<-,black] (19,-18) -- (21, -15) node[right] {$\theta=2-\gamma$};
		\end{tikzpicture}
	\end{center}
	\caption{The five different hydrodynamic regimes obtained in Theorem \ref{th:hl900} as function of $\gamma$ and $\theta$.}
	\label{fig:hlf_lj}
\end{figure}

\subsection{Equilibrium fluctuations} \label{sec:equi_fluctuations}
\label{s4}

\subsubsection{Holley-Strook Theory}
\label{ssec:test_functions}

Before stating our main result we review the notion of solutions to the SPDEs that we will derive from our underlying dynamics.  We follow the theory of martingale problems initiated in \cite{holley1978generalized} which, in this setting, is explained in \cite{KL} and generalized in the presence of boundary conditions in \cite{PTA}. Here we provide a general theoretical framework encapsulating all the previous studies. \\

\begin{definition}
	\label{def:OUP}
	Fix $T>0$ some  time horizon. Let ${\mc C}$ be a topological vector space, ${\mc A}: {\mc C} \to {\mc C}$ be an operator letting ${\mc C}$ invariant and $c:{\mc C} \to [0, \infty)$ be a continuous functional \textcolor{black}{satisfying  \begin{equation}\label{hypothesis_c}c(\lambda H)=|\lambda|c(H)\end{equation}for all $\lambda\in\mathbb R$ and $H\in\mathcal C$.} Let ${\mc C}'$ be the topological dual of ${\mc C}$ equipped with the weak $\star$ topology. We say that the process $\{\mc Y_t\; ; \; t \in [0,T]\} \in C ([0,T], {\mathcal C}^\prime)$ is a solution of the Ornstein-Uhlenbeck martingale problem 
	$$OU({\mc C}, {\mc A}, c)$$
	on the time interval $[0,T]$ with initial (random) condition ${y}_0 \in {\mc C}^\prime$ if:
	\begin{itemize}
		\item[(i)] 
		for any $H \in {\mc C}$ the two real-valued processes $M_\cdot (H)$ and $N_\cdot (H)$ defined by
		\begin{equation}
		M_t(H)=\mathcal{Y}_t(H)-\mathcal{Y}_0(H)-\int_0^t\mathcal{Y}_s ( \mathcal{A} H )ds;
		\label{limmart}
		\end{equation}
		\begin{equation}
		\label{limmart_1}
		N_t(H)=M_t(H)^2-tc^2 (H),
		\end{equation}
		are martingales with respect to the natural filtration of the process $\mc Y_\cdot$, that is, $\{\mc F_t \; ; \; t \in [0,T]\}:=\{\sigma(\mc Y_s(H)\; ;\;  s\leq t,  H \in \mathcal{C})\; ; \; t\in [0,T]\}$.	
		\item[(ii)] $\mc Y_0=y_0$ in law.
	\end{itemize}
\end{definition}

Holley-Strook approach \cite{holley1978generalized} can be extended in a general setting and gives the following proposition. 

\begin{prop}
	\label{prop:unique}
	Assume there exists a semigroup $\{P_t\}_{{t\ge 0}}$ on ${\mathcal C}$ associated to the operator ${\mc A}: {\mc C} \to {\mc C}$ in the sense that 
	\begin{equation}
	\label{eq:exp-OU}
	P_{t+\epsilon} H - P_t H =\epsilon {\mc A} P_t H \; +\; o (\epsilon, t)
	\end{equation}
	where $\epsilon^{-1} o(\epsilon, t)$ goes to $0$, as $\epsilon$ goes to $0$, in ${\mc C}$ uniformly on compact time intervals. Then there exists a unique solution to $OU({\mc C}, {\mc A}, c)$.
\end{prop}

\begin{proof}
	The proof is given in Section \ref{Subsec:UOU}.
\end{proof}

Let us observe that \eqref{eq:exp-OU} implies that for any $H \in {\mc C}$, the application $t\in [0, \infty) \to P_t H \in {\mc C}$ is ${C}^1$ and satisfies 
\begin{equation}\label{eq:semigroup_properties}
\partial_t P_t H = {\mc A} P_t H, \quad P_0 H = H.
\end{equation}

\textcolor{black}{Denoting by $\mathbb{E}$ the expectation with respect to the law of $\{\mc Y_t\; ; \; t \in [0,T]\} \in C ([0,T], {\mathcal C}^\prime)$ we can give the following definition of stationary solution of the martingale problem $OU(\mc C, {\mc A}, c)$.}

\begin{definition}
	\label{def:OUP-ss}
	Assume that the conditions of Proposition \ref{prop:unique} are satisfied. We say that the process $\{\mc Y_t\; ; \; t \in [0,T]\} \in  C ([0,T], {\mathcal C}^\prime)$ is a stationary solution of the martingale problem $OU(\mc C, {\mc A}, c)$ on the time interval $[0,T]$ if $\mc Y_0$ is a Gaussian field of mean zero and covariance function ${\mf C}: \mc C \times \mc C \to \RR$ given on $H,G \in\mathcal{ C}$ by
	\begin{equation}
	\label{eq:covar1}
	\mathbb{E} \big[ \mc Y_0(H) \mc Y_0(G)\big] \;:=\; {\mf C} (H, G).
	\end{equation}
	From this, it follows that the process $\{\mc Y_t\; ; \; t\in [0,T]\}$ is Gaussian and that for any $0\le s\le t \le T$, any $H,G \in \mathcal{C}$, its covariance  is given by
	\begin{equation}
	\label{eq:covpat}
	\mathbb{E}[\mathcal{Y}_t(H)\mathcal{Y}_s(G)]= {\mf C} ( P_{t-s} H, G ),
	\end{equation}
	where $\{P_t\; ; \; t\ge 0\}$ is the semigroup generated by the operator  $\mathcal{A}$ in the sense of \eqref{eq:semigroup_properties}.
\end{definition}

\begin{rem}
	The fact that the solution of the martingale problem $OU(\mc C, {\mc A}, c)$, whose initial condition is a centered Gaussian field, is a Gaussian process with the covariance function prescribed in the definition above, comes directly from the proof of  Proposition \ref{prop:unique}. See Lemma \ref{lem:Guss} for details.
\end{rem}   

It is important to observe that in the definition above, the choice of the space ${\mathcal C}$ plays a capital role  and, in the present study, encapsulates the boundary conditions that we want to observe.

Let $C^\infty ([a,b])$ ($a < b$ both in $\RR$) be the space of real-valued smooth {\footnote{We mean that the function $H$ is smooth on $(a,b)$ and that all the derivatives of $H$ can be continuously extended to $[a,b]$.}} functions $H : [a,b] \to \bb R$. We equip this space with the topology defined by the family of seminorms
\begin{equation}
\label{seminorms}
\left\{\sup_{u \in [0,1]} |{H^{(i)}(u)|} \right\}_{i \in \mathbb{N}\cup\{0\}}.
\end{equation}

In the sequel, we will consider in particular the following topological vector spaces
\begin{align}
\mc S :=& \Big\{ H \in C^\infty ([0, 1])\; ; \;  H^{(i)}(0)=H^{(i)}(1)=0, \; \forall i \in \mathbb{N}\Big\} ,\label{eq:space_test_func}\\
\mc S_{Dir} :=& \Big\{ H \in C^\infty ([0, 1])\; ; \;  H^{(2i)}(0)=H^{(2i)}(1)=0, \; \forall i \in \mathbb{N}\Big\}, \label{eq:space_test_funcDir}\\
\mc S_{Rob} :=&\Big\{ H \in C^\infty ([0,1]) \; ; \; H^{(2i+1)}(0)=-\tfrac{2m\kappa}{\sigma^2}H^{(2i)}(0) \Big.\\
& \Big. \quad \quad \text{ and } H^{(2i+1)}(1)=\tfrac{2m\kappa}{\sigma^2}H^{(2i)}(1), \; \forall i \in \mathbb{N}\Big\}, \label{eq:space_test_funcRob}
\\
\mc S_{ Neu} :=& \Big\{H \in C^\infty([0,1]) \; ; \; H^{(2i+1)}(0)=H^{(2i+1)}(1)=0, \; \forall i \in \mathbb{N}\Big\}.\label{eq:space_test_funcNeu}
\end{align}

\begin{rem}\label{Frec} Note that all the spaces  introduced above are Fr\'echet spaces. We recall that a Fr\'echet space is a complete Hausdorff space whose topology is induced by a countable family of semi-norms. Indeed, since the space $C^\infty([0, 1])$ endowed with the semi-norms defined in \eqref{seminorms} is a Fr\'echet space, and any closed subspace of a Fr\'echet space is also a Fr\'echet space, to conclude, it is enough to observe that uniform convergence implies point-wise convergence (see \cite{RS} for more details on Fr\'echet spaces). Working with Fr\'echet spaces will be fundamental in Section \ref{s6} for the proof of tightness. 
\end{rem}

\begin{definition}
	\label{norm}
	\hspace{10cm}
	\begin{itemize}
		\item Let $\mathcal{S}_{\theta}$ be the set $\mathcal{S}$ if $\theta<1$, the set $\mathcal{S}_{Rob}$ if $\theta=1$ and $\mathcal{S}_{Neu}$ if $\theta>1$. 
		\item We define the inner product $\llangle H, G \rrangle_\theta$ between two functions $H,G : [0,1]\rightarrow \RR$ with regularity at least $C^1([0,1])$ by
		\begin{equation}
		\label{inprod}
		\begin{split}
		\llangle H,G\rrangle_{\theta}:=2\chi(\rho)\sigma^2 &\bigg[ \mathbb{1}_{\{\theta\geq 2-\gamma\}} \langle\nabla H,\nabla G\rangle\\&+\mathbb{1}_{\{\theta \leq 2-\gamma\}}\frac{\kappa}{\sigma^2}\textcolor{black}{\int_0^1V_1(u)H(u)G(u)du}\\
		&+\mathbb{1}_{\{\theta=1\}}\frac{\sigma^2}{4\kappa m}(\nabla H (0)\nabla G (0)+\nabla H(1)\nabla G(1))\bigg],
		\end{split}
		\end{equation}
		where $\chi(\rho):=\rho(1-\rho)$ and $\langle \cdot,\cdot \rangle$ is the standard $L^2$ inner product. 
		Then, the Hilbert space ${\bb H}_\theta$ is obtained as the completion of the space of functions $H \in C^1 ([0,1])$ for the norm induced by the previous inner product:
		\begin{equation}
	\textcolor{black}{	\Vert H \Vert^2_{\theta}=\llangle H,H\rrangle_{\theta}.}
		\label{inn}
		\end{equation}
	\end{itemize}
\end{definition}

\begin{lem}
	For any $\theta$, the application $H \in {\mc S}_\theta \to \Vert H \Vert_\theta \in \RR$ is a continuous functional.
\end{lem}

\begin{proof}
	In order to prove the continuity of the functional, fix some $G\in\mathcal S_\theta$ and some  $\epsilon>0$.  We have to prove that there exists some open neighborhood  of $G$ such that, if $H$ is in this neighborhood, then $\big| \Vert H \Vert^2_{\theta}- \Vert G \Vert^2_{\theta}\big|<\epsilon$. We take this neighborhood in the form 
	\begin{equation*}
	V^{k,\delta}_G:=\left\{H \in {\mc S}_\theta \; ; \; \forall j\le k, \; \sup_{u \in [0,1]} \vert (H-G)^{(j)} (u) \vert \le \delta \right\}
	\end{equation*}
	for suitable $k\ge 0$ and $\delta>0$ fixed later. In the sequel $H\in V^{k,\delta}_G$.\\

	Let us start with the case $\theta\geq {2-\gamma} $. Note that,
	\begin{equation*}
	\begin{split} 
	\Big|\int_0^1(\nabla H(u))^2-(\nabla G(u))^2du\Big|&\leq \int_0^1|\nabla H(u)-\nabla G(u)||\nabla H(u)+\nabla G(u)|du\\
	&\le \delta \, \Big [\delta + 2 \sup_{u \in [0,1]} |(\nabla G )(u)| \Big] \le \epsilon ,
	\end{split}
	\end{equation*} 
	if $\delta$ is taken sufficiently small with respect to $\epsilon$. 
	
	The same argument  as used above, allows us to deal with the extra term  obtained in the case $\theta=1$.
	
	To conclude it remains to analyse the reaction term. To that end, observe that in the regime $\theta\leq 2-\gamma$, ${\mc S}_\theta={\mc S}$, and we can restrict our study to  the integral with $r^- $, since for the one with $r^+$ the argument is exactly the same.  Since $H-G\in\mathcal S$, by applying  Taylor expansion to the function $H-G$ around zero at order $k$ we get that:
	\begin{equation*}
	\forall u \in [0,1], \quad \vert (H-G)(u) \vert \lesssim \left[ \sup_{v \in [0,1]} \vert (H-G)^{(k)} (v) \vert \right] \; u^k \lesssim \delta \; u^k.   
	\end{equation*}
	On the other hand we have that
	\begin{equation*}
	\sup_{u \in [0,1]} \vert (H+G) (u) \vert \le \delta + 2\sup_{u \in [0,1]} \vert G (u) \vert.
	\end{equation*}
	It follows that
	
	\begin{equation*}
	\begin{split}
	\Big|\int_0^1r^-(u)(H^2(u)-G^2(u))du\Big|&\lesssim \delta \left[ \delta + 2\sup_{u \in [0,1]} \vert G (u) \vert\right] \Big|\int_0^1u^{(k-\gamma)}du\Big|.
	\end{split}
	\end{equation*}
	By choosing $k$ sufficiently big so that last integral is finite, last display becomes bounded by $\epsilon$ if $\delta$ is chosen sufficiently small. This ends the proof. 
\end{proof}

We define	the unbounded operators ${\mc A}_\theta$ on ${\mc S}_\theta$ by  
\begin{equation}
\label{eq:operator_A_theta}
\mathcal{A}_{\theta}:=\tfrac{\sigma^2}{2}\, {\mathbb 1}_{\theta \geq 2-\gamma}\, \Delta \, - \, \kappa\,  \mathbb{1}_{\theta\leq 2-\gamma} V_1.
\end{equation}

In \cite{PTA, FGN_Robin, GJMN} it is proved the following theorem 
\begin{thm}[\cite{PTA, FGN_Robin, GJMN}]
	\label{thm:ptg}
	It holds that:
	\begin{itemize}
		\item If $\theta \ge 1$, then there exists a unique  (in law)  solution of the martingale problem $OU({\mc S}_\theta, {\mc A}_\theta, \Vert\cdot\Vert_\theta)$.
		\item If $2-\gamma < \theta <1$, then there exists a unique  (in law)  solution of the martingale problem $OU({\mc S}_{Dir}, {\mc A}_\theta, \Vert\cdot\Vert_\theta)$.
	\end{itemize}
\end{thm}

\subsection{Density fluctuation field}\label{ssec:density} 
We are interested in studying the density fluctuations at equilibrium. The system is in equilibrium if, for a fixed parameter $\rho \in (0,1)$, we put $\alpha=\beta=\rho$ and we choose the initial measure on $\Omega_N$ as the Bernoulli product measure $\nu_{\rho}$. We denote by $\mathbb{P}_{\nu_{\rho}}$ the law of the process $\{ \eta_t^N \, ; \, t \in[0,T] \}$ starting from $\nu_\rho$ and the corresponding expectation is denoted by $\EE_{\nu_\rho}.$ Note that this expectation is not the same as ${E}_{\nu_{\rho}}$ which is with respect to the measure $\nu_{\rho}$ defined on $\Omega_N$. Since $\nu_{\rho}$ is a stationary measure we have that $\mathbb E_{\nu_{\rho}}[\eta^N_{t}(x)]$ is constant in $t $ and $x$ and it is equal to $\rho$.
We are therefore looking at the fluctuations of the configurations around their mean with respect to $\nu_{\rho}$, namely at the following quantity.

\begin{definition} For any $t>0$,  the \emph{density fluctuation
		field} $\mathcal{Y}^N_t \in {\mathcal S}_\theta^\prime$ is defined as the random distribution acting on test functions $H \in \mathcal{S}_{\theta}$  as
	\begin{equation*}
	\mathcal Y^N_t (H) := \frac 1{\sqrt{ N-1}} \sum_{x=1}^{N-1} H\left(\tfrac x N\right)
	\overline{\eta}_t^N(x),
	\end{equation*}
	where 
	\begin{equation}\label{eq:centered_eta}
	\overline{\eta}_t^N(x):=\eta^N_{t}(x) - \rho.
	\end{equation}
\end{definition}
We  denote by $Q^{\theta,N}$ the probability measure on ${D}([0,T], \mathcal{S}_{\theta}')$ associated to the density fluctuation field $\mathcal{Y}_{\cdot}^N$. We  denote by $Q$ the limit point of $\{Q^{\theta,N}\}_{N> 1}$ (we will prove that it exists, is unique and concentrated on ${C}([0,T], \mathcal{S}_{\theta}')$). So, the expectation denoted by $\mathbb{E}_Q$ is with respect to the limit measure $Q$ which will be defined on the space ${C}([0,T], \mathcal{S}_{\theta}')$.

Our first result extends Theorem \ref{thm:ptg} to the regime $\theta \le 2-\gamma$ and proposes a different notion of solution to the martingale problem in the case $\theta \in (2-\gamma,1)$.

\begin{thm}
	\label{prop:uniqueness}
	The following holds:
	\begin{enumerate}[1.]
		\item Assume $\theta \le 2-\gamma$. There exists a unique solution of the martingale problem $OU({\mc S}_\theta, {\mc A}_\theta, \Vert \cdot \Vert_{\theta})$.
		\item Assume $2-\gamma < \theta <1$. There exists a unique (in law) random  element $\{\mc Y_t\; ; \; t \in [0,T]\} \in  C ([0,T], {\mathcal S}^\prime_\theta) $ which is a stationary solution of the martingale problem $OU({\mc S}_\theta, {\mc A}_\theta, \Vert \cdot \Vert_\theta)$ and which satisfies the two extra conditions:
		\begin{itemize}
			\item[(i)] \textit{regularity condition:} $\mathbb{E} [(\mathcal{Y}_t(H))^2]\lesssim ||H||_{L^2}^2$ for any $H \in \mc S_\theta$;
			\item[(ii)] \textit{boundary condition:}  Let $\iota^0_{\epsilon}:=\epsilon^{-1} \mathbb{1}_{(0,\epsilon]}$ and $\iota^1_{\epsilon}:=\epsilon^{-1} \mathbb{1}_{[1-\epsilon,1)}$. For any $s \in [0,T]$ and $j=0,1$, consider ${\mc Y}_s (\iota_\epsilon^{j})$,  defined using Lemma \ref{lem:ex1-ss}\footnote{$\iota_{\epsilon}^1$ and $\iota_{\epsilon}^0$ are not continuous functions and so, in particular, are not elements of any of the test function spaces considered. Making sense of this is exactly the \textcolor{black}{content} of Lemma \ref{lem:ex1-ss}.},  we have that
			\begin{equation*}
			\lim_{\epsilon \rightarrow 0}\mathbb{E} \left[\sup_{0\le t\le T} \left(\int_0^t \mathcal{Y}_s (\iota^0_{\epsilon}) ds \right)^2\right]=\lim_{\epsilon \rightarrow 0}\mathbb{E} \left[\sup_{0\le t\le T} \left(\int_0^t \mathcal{Y}_s (\iota^1_{\epsilon}) ds\right)^2\right]=0.
			\end{equation*} 
			Moreover it coincides with the unique stationary solution of the martingale problem $OU({\mc S}_{Dir}, {\mc A}_\theta, \Vert \cdot \Vert_\theta)$
		\end{itemize}
	\end{enumerate}
\end{thm}

\begin{rem}
	The first part of the previous theorem follows closely the Holley-Strook approach, the difficulty here is  that this approach requires to well understand the properties of the singular Sturm-Liouville problem associated to $\mc A_\theta$ for $\theta =2-\gamma$. The analysis of this problem is non-standard and it  is the content of Subsection \ref{susubsec:slpbm}. In the second part of the theorem we observe that if ${\mc S}_\theta$ is replaced by ${\mc S}_{Dir}$ then, the uniqueness of  $OU({\mc S}_{Dir}, {\mc A}_\theta, \Vert \cdot \Vert_\theta)$ holds without the extra conditions (i) and (ii). Nevertheless, the space ${\mc S}_\theta={\mc S}$ is not sufficiently large to encapsulate Dirichlet boundary conditions and to provide uniqueness of the solution of the martingale problem $OU({\mc S}, {\mc A}_\theta, \Vert \cdot \Vert_\theta)$. As explained below, since we do not know how to show that the limiting points of our sequence of fluctuations fields is a stationary solution of $OU({\mc S}_{Dir}, {\mc A}_\theta, \Vert \cdot \Vert_\theta)$ but only a stationary solution of $OU({\mc S}, {\mc A}_\theta, \Vert \cdot \Vert_\theta)$, we have to show that the limiting points  satisfy the extra conditions (i) and (ii), which are sufficient to restore uniqueness. 
\end{rem}

The second main result of this paper is the following theorem.

\begin{thm} \label{maintheorem}
	The sequence of probability measures $\big\{Q^{\theta,N}\big\}_{N> 1}$ associated to the sequence of fluctuation fields $\{\mathcal{Y}_{\cdot}^N\}_{N > 1}$ converges, as $N$ goes to infinity, to a probability measure $Q$, concentrated on the unique stationary solution of $OU({\mc S}_\theta, {\mc A}_\theta, \Vert \cdot \Vert_{\theta})$ if $\theta \notin (2-\gamma, 1)$ and to the unique stationary solution of $OU({\mc S}_{Dir}, {\mc A}_\theta, \Vert \cdot \Vert_{\theta})$ if $\theta \in (2-\gamma, 1)$.
\end{thm}

The proof of last theorem is a consequence of two facts: the sequence of density fluctuation fields is tight and there exists a unique limit point. In the next section we analyze the convergence of discrete martingales associated to the density fluctuation field $\mathcal Y_t^N$ to the martingales given in \eqref{limmart} and ~\eqref{limmart_1}. In Section \ref{s6}, we prove the existence of the limit point. Section \ref{s7} is devoted to the proof of uniqueness of the limit point, i.e. to the proof of Theorem \ref{prop:uniqueness}. We finish by presenting in Section \ref{s8} some technical lemmas that were used along the article. In Section \ref{Subsec:UOU} we prove Proposition \ref{prop:unique}, that is the uniqueness of the Ornstein-Uhlenbeck process and in the appendix we present an approximation lemma which is needed along the proofs.

\section{Characterization of limit points}
\label{sec_car}

In order to characterize the limit points, denoted by $Q$, of the sequence $\{Q^{N,\theta}\}_{N> 1}$, we start by using  Dynkyin's formula (see, for example, Lemma A.5 of \cite{KL}), which permits to conclude that, for any test function $H\in C^\infty([0,1])$,
\begin{equation}
M_t^N(H)=\mathcal{Y}_t^N(H)-\mathcal{Y}_0^N(H)-\int_0^t\Theta(N)L_N\mathcal{Y}_s^N(H)ds
\label{mart1}
\end{equation}
is a martingale with respect to the natural filtration $\{\mathcal{F}^N_t\; ; \; t\in [0,T]\} =\{ \sigma(\eta^N_s\; ;\;  s \leq t)\; ; \; t \in [0,T]\}$. Up to here we did not impose $H$ to belong to any of the space of test functions introduced above. Below, we will split the argument into several regimes of $\theta$ and then we will precise where the test function $H$ will live.
In order to  characterize  the limit points of the  sequence $\{Q^{N,\theta}\}_{N > 1}$, we analyze the convergence of $\{M_t^N(H)\}_{N > 1}$, $\{\mathcal{Y}_0^N(H)\}_{N > 1}$ and $\big\{\int_0^t\Theta(N)L_N\mathcal{Y}_s^N(H)ds\big\}_{N > 1}$ separately.  We start by exploring the convergence of the integral term. 

\subsection{Convergence of the integral term}

We compute the integral term in \eqref{mart1}. From (\ref{Generators}), we can easily see that
\begin{equation}
\label{integral1}
\begin{split}
\Theta(N)L_N\mathcal{Y}_s^N(H)=& \frac{\Theta(N)}{\sqrt{N-1}}\sum_{x \in \Lambda_N}\mathcal{L}_NH\left(\tfrac{x}{N}\right)\overline{\eta}_s^N(x)\\
-&\frac{\kappa\Theta(N)}{N^\theta\sqrt{N-1}}\sum_{x \in \Lambda_N}H\left(\tfrac{x}{N}\right)\left[r^-_N\left(\tfrac{x}{N}\right)+r^+_N\left(\tfrac{x}{N}\right)\right]\overline{\eta}_s^N(x),
\end{split}
\end{equation}
where, for any $x\in \Lambda_N$, we define
\begin{equation}
\mathcal{L}_NH\left(\tfrac{x}{N}\right)=\sum_{y\in\Lambda_N}p(y-x)\left[H\left(\tfrac{y}{N}\right)-H\left(\tfrac{x}{N}\right)\right]
\end{equation}
and
\begin{equation}\label{def:rN}
r_N^-\left(\tfrac{x}{N}\right)=\sum_{y \geq x}p(y), \qquad r_N^+\left(\tfrac{x}{N}\right)=\sum_{y \leq x-N}p(y).
\end{equation}

We will analyze (\ref{integral1}) extending the test function \textcolor{black}{$H \in \mathcal{S}_{\theta}$ (recall Definition \ref{norm}) to a bounded function $\widehat H$ defined on the whole set $\RR$.} Extending the sum in the definition of $\mc L_N$ to the whole $\bb Z$ we get
\begin{equation}
\begin{split}
&\frac{\Theta(N)}{\sqrt{N-1}}\sum_{x \in \Lambda_N}\mathcal{L}_NH\left(\tfrac{x}{N}\right)\overline{\eta}_s^N(x)\\ &=\frac{\Theta(N)}{\sqrt{N-1}}\sum_{x \in \Lambda_N}K_N\widehat H\left(\tfrac{x}{N}\right)\overline{\eta}_s^N(x)\\
&-\frac{\Theta(N)}{\sqrt{N-1}}\sum_{x \in \Lambda_N}\sum_{y \leq 0} p(y-x)\left[\widehat H\left(\tfrac{y}{N}\right)-\widehat H\left(\tfrac{x}{N}\right)\right]\overline{\eta}_s^N(x)\\
&-\frac{\Theta(N)}{\sqrt{N-1}}\sum_{x \in \Lambda_N}\sum_{y \geq N} p(y-x)\left[ \widehat H\left(\tfrac{y}{N}\right)-\widehat H\left(\tfrac{x}{N}\right)\right]\overline{\eta}_s^N(x)
\end{split}
\label{integral2}
\end{equation}
where
\begin{equation}\label{KN}
K_N {\widehat H} \left(\tfrac{x}{N}\right):=\sum_{y\in \ZZ}p(y-x)\left[\widehat H\left(\tfrac{y}{N}\right)-\widehat H\left(\tfrac{x}{N}\right)\right],
\end{equation}
for any $x \in \Lambda_N$.
Below we analyze the contribution of each term in the decomposition above for each regime of $\theta.$

\subsubsection{Case $\theta <2-\gamma$}\label{Secrev}
In this case recall that $\Theta(N)=N^{\gamma+\theta}$ and $\mathcal{S}_{\theta}=\mathcal{S}$. We take $H \in {\mathcal S}$ and extend it by $0$ outside of $[0,1]$ producing a smooth function $\widehat H$.

Then, we can rewrite the first term on the \textcolor{black}{right-hand side} of (\ref{integral2}) as
\begin{equation}
\label{K}
\frac{N^{\gamma+\theta-2}}{\sqrt{N-1}}\sum_{x \in \Lambda_N}N^2K_N\widehat H\left(\tfrac{x}{N}\right)\overline{\eta}_s^N(x).
\end{equation}
Since $\nu_\rho$ \textcolor{black}{is stationary and it is a product measure},  by Lemma \ref{lemmasigma}   the \textcolor{black}{square of the $L^2(\nu_\rho)$} norm  of last term is bounded from above by
\begin{equation}
\frac{\sigma^4 \chi(\rho)}{4}N^{\gamma+\theta-2}\Vert \Delta \widehat H \Vert^2_{L^2} \lesssim N^{\gamma+\theta-2},
\end{equation}
plus some terms of lower order in $N$. Therefore we conclude that \eqref{K} goes to $0$ in $L^2 (\nu_\rho)$, as $N$ goes to infinity.

Let us now treat the last two terms on the \textcolor{black}{right-hand side} of (\ref{integral2}). They can be treated exactly in the same way, so we explain how to analyze only the first one. Since, by the observation given above, in this regime we are using a function $\widehat H$ vanishing outside $[0,1]$, the second term on the \textcolor{black}{right-hand side} of  (\ref{integral2}) is equal to
\begin{equation}
\frac{N^{\gamma + \theta}}{\sqrt{N-1}}\sum_{x \in \Lambda_N} {H} \left(\tfrac{x}{N}\right)\overline{\eta}_s^N(x)r^-_N\left(\tfrac{x}{N}\right).
\label{dd}
\end{equation}
Now, we use the fact that functions in $\mathcal{S}$ have all derivatives equal to zero at $0$. From  Taylor's expansion  of $H$ around the point $0$ up to order $d\geq 1$, 
plus the fact that $r^-_N(x/N)$ is of order $x^{-\gamma}$ (\textcolor{black}{see Section 3 of \cite{BGJO} for details on this estimate}), the $L^2 (\nu_\rho)$ norm of \eqref{dd} is bounded from above, using Cauchy-Schwarz inequality, by a constant times
\begin{equation*}
N^{2( \theta + \gamma-d) -1}\sum_{x \in \Lambda_N} x^{-2\gamma+2d},
\end{equation*}
plus lower order terms with respect to $N$.
Then, choosing $d$ such that $2\gamma-2d<1$, the previous display is bounded from above by  $N^{2\theta}$ and, since $\theta<0$, it vanishes, as $N$ goes to infinity.

Now, we still have to analyze the last term in \eqref{integral1}. We explain in details just how to do it for the part involving $r^{-}_N$ since the other part can be treated equivalently. We can rewrite it as
\begin{equation*}
\frac{\kappa}{\sqrt{N-1}}\sum_{x \in \Lambda_N}H\left(\tfrac{x}{N}\right)N^{\gamma}r^-_N\left(\tfrac{x}{N}\right)\overline{\eta}_s^N(x).
\end{equation*}
By summing and subtracting $r^-(\tfrac xN)$ inside the sum above and using 
Lemma \ref{reaction}, we conclude that this term can be replaced by 
\begin{equation*}
\frac{\kappa}{\sqrt{N-1}}\sum_{x \in \Lambda_N}H\left(\tfrac{x}{N}\right)r^-(\tfrac{x}{N})\overline{\eta}_s^N(x)=\kappa\mathcal{Y}_s^N(Hr^-)
\end{equation*}
plus some term that vanishes in $L^2 (\nu_\rho)$, as $N$ goes to infinity.

From the previous computations, we conclude that  in the regime $\theta < 2-\gamma$ and for any test function $H\in \mc S$
\begin{equation}
\textcolor{black}{M^N_t(H)=\mathcal{Y}^N_t(H)-\mathcal{Y}^N_0(H)+\kappa\int_0^t\mathcal{Y}^N_s(HV_1)ds}
\label{<2gamma}
\end{equation}
plus terms that vanish  in $L^2(\nu_\rho)$, as $N$ goes to infinity.

\subsubsection{Case $\theta =2-\gamma$}\label{secrev2}
In this case the only difference  with respect to the previous case is the fact that $\Theta(N)=N^2$. We also assume that  $H\in\mathcal{S}$ and extend $H$ by $0$ outside of $[0,1]$ producing a smooth function $\widehat H$.

Therefore the first term of (\ref{integral2}) does not vanish, but instead, as a consequence of  Lemma \ref{lemmasigma}, it can be replaced by $\frac{\sigma^2}{2}\mathcal{Y}_s^N(\Delta {\widehat H})= \frac{\sigma^2}{2}\mathcal{Y}_s^N(\Delta {H})$ plus some terms that vanish in $L^2(\nu_\rho)$, as $N$ goes to infinity.  We also observe that the other two terms in (\ref{integral2}) and the last two terms in (\ref{integral1}) are treated in an analogous way to what we did in the case $\theta < 2-\gamma$. We leave this details to the reader.

Summarizing, from the previous computations we conclude that  in the regime $\theta =2-\gamma$ and for test functions $H\in \mc S$
\begin{equation}
M^N_t(H)=\mathcal{Y}^N_t(H)-\mathcal{Y}^N_0(H)-\int_0^t\left\{\frac{\sigma^2}{2}\mathcal{Y}^N_s(\Delta H)-\kappa\mathcal{Y}^N_s(H(r^-+r^+)) \right\}ds
\label{=2gamma}
\end{equation}
plus terms that vanish  in $L^2(\nu_\rho)$, as $N$ goes to infinity.

\subsubsection{Case $\theta \in ( 2-\gamma,1)$}
\label{2-1}

To characterize the limit points of the sequence $\{\mc Y_{\cdot}^N\}_{N> 1}$ in this regime we would like to take a test function $H\in {\mathcal S}_{Dir}$ (introduced in Section \ref{ssec:test_functions}) since the uniqueness of the stationary solution of $OU({\mc S}_{Dir}, {\mc A}_\theta, \Vert \cdot\Vert_{\theta} )$ requires this set of test functions. But to control the additive functionals that appear in Dynkin's formula, involving the boundary dynamics, we have to consider $H$ in the smaller set  $\mathcal{S}$. After, we will show that if the martingale problem holds for test functions in $\mathcal{S}$, then, thanks to some extra condition we will prove later, it also holds  for test functions in $\mathcal{S}_{Dir}$ and within this set {uniqueness holds for } the solution of our martingale problem. 

Hence, we consider here $H \in {\mathcal S}$ and we extend it by $0$ outside of $[0,1]$ producing a smooth function $\widehat H$.

First, let us treat the boundary terms on the \textcolor{black}{right-hand side} of (\ref{integral1}). We will focus on the term involving $r^-_N$, the other can be treated analogously. Using Lemma \ref{lemma1} we can obtain the following bound
\begin{equation}
\mathbb{E}_{\nu_{\rho}}\Bigg[\bigg(\int_{0}^t\tfrac{\kappa N^2}{N^\theta\sqrt{N-1}}\sum_{x \in \Lambda_N}H\left(\tfrac{x}{N}\right)r^-_N\left(\tfrac{x}{N}\right)\overline{\eta}_s^N(x)\bigg)^2\Bigg] \lesssim \kappa N^{1-\theta}\sum_{x \in \Lambda_N } H^2\left(\tfrac{x}{N}\right)r^-_N\left(\tfrac{x}{N}\right).
\end{equation}
Then, using the fact that the test function $H $ belongs to  $\mc S$, applying Taylor expansion of $H$ around the point $0$ up to order $d\geq 1$,  we get that the term on the \textcolor{black}{right-hand side} of the previous display is of order
\begin{equation}
\kappa N^{1-\theta-2d}\sum_{x \in \Lambda_N}\frac{1}{x^{\gamma-2d}}.
\end{equation}
Taking a value of $d$ such that the sum diverges, we have that the order of the previous term is 
$N^{2-\gamma-\theta}$ and, since $\theta >2-\gamma$, it goes to $0$, as $N$ goes to infinity.

The analysis of all the terms in \eqref{integral2} is analogous to the case $\theta=2-\gamma$.  Indeed, these terms do not have any dependence on $\theta$ and since $H \in \mc S$, we can repeat the argument.  So, we have that
\begin{equation}
M^N_t(H)=\mathcal{Y}^N_t(H)-\mathcal{Y}^N_0(H)-\int_0^t\frac{\sigma^2}{2}\mathcal{Y}^N_s(\Delta H)ds,
\label{>2gamma}
\end{equation}
plus terms that vanish  in $L^2( \nu_\rho)$, as $N$ goes to infinity.

\subsubsection{Case $\theta =1$}
\label{1}
Now we  assume that  $H\in\mathcal{S}_{Rob}$ and we consider $\widehat H$ a smooth and bounded extension of it  defined on $\RR$.

In this case, the first term on the \textcolor{black}{right-hand side} of (\ref{integral2}) is treated simply using Lemma \ref{lemmasigma}, so it can be written as $\frac{\sigma^2}{2}\mathcal{Y}_s^N(\Delta H)$ plus a term that vanishes in $L^2 (\nu_\rho)$, as $N$ goes to infinity.
Let us now consider the other two terms  on the \textcolor{black}{right-hand side} of (\ref{integral2}). We focus on the first one, since the analysis of the second one is analogous. Performing a Taylor expansion  on $H$ around the point $\tfrac xN$, we can rewrite it as
\begin{equation}
\begin{split}
\tfrac{1}{\sqrt{N-1}}\sum_{x \in \Lambda_N}\Big\{N\Theta^-_x&H'\left(\tfrac{x}{N}\right)\overline{\eta}_s^N(x)+\tfrac{1}{2}\sum_{y \leq 0} p(y-x)(y-x)^2H''\left(\tfrac{x}{N}\right)\overline{\eta}_s^N(x)\Big\}
\end{split}
\label{taylortheta}
\end{equation}
plus lower order terms with respect to $N$.  Above $\Theta^-_x:=\sum_{y \leq0}(y-x)p(y-x)$. First, let us consider the term on the \textcolor{black}{right-hand side} of the previous display. It is easy to see that its  $L^2 (\nu_\rho)$ norm  is bounded from above, using Cauchy-Schwarz inequality, by 
\begin{equation}
\frac{\chi(\rho)\|H''\|^2_{\infty}}{N} \sum_{x \in \Lambda_N}\Big(\sum_{y \leq 0} p(y-x)(y-x)^2\Big)^2\lesssim \frac{1}{N}\sum_{x\in \Lambda_N}\frac{1}{x^{2\gamma-4}}.
\label{prev}
\end{equation}
Now, if $\gamma>5/2$, the sum converges and  last term goes to $0$, as $N$ goes to infinity. If $\gamma=5/2$ the sum is of order $\log(N)$ and the last term still vanishes, as $N$ goes to infinity. Otherwise,  the sum is of order $N^{5-2\gamma}$ and the order of the whole expression is $N^{4-2\gamma}$, which means that it goes to $0$, as $N$ goes to infinity.

It remains to analyze the term on the \textcolor{black}{left-hand side} of \eqref{taylortheta}. It can be written, using Taylor expansion, as 
\begin{equation}
\frac{N}{\sqrt{N-1}} H'(0)\sum_{x \in \Lambda_N}\Theta_x^-\overline{\eta}_s^N(x)\\
+\frac{1}{\sqrt{N-1}}H''(0)\sum_{x \in \Lambda_N}x\Theta_x^-\overline{\eta}_s^N(x)
\label{taylortheta2}
\end{equation}
plus lower order terms with respect to $N$.  A simple computation shows that $\Theta^-_x$ is of order $x^{1-\gamma}$ (\textcolor{black}{see (3.10) of \cite{BGJO}})  and so the $L^2 (\nu_\rho)$ norm of the term on the \textcolor{black}{right-hand side} of the previous expression is bounded, using Cauchy-Schwarz inequality, by 
\begin{equation*}
\frac{1}{N-1}\left(H''(0)\right)^2\sum_{x \in \Lambda_{N} }x^2 (\Theta_x^-)^2\chi(\rho)\lesssim \frac{1}{N}\sum_{x\in \Lambda_N}\frac{1}{x^{2\gamma-4}}.
\end{equation*}
So, reasoning as in the analysis of \eqref{prev}, we can conclude that it vanishes, as $N$ goes to infinity.  Let us now keep the remaining  term  of (\ref{taylortheta2}). 

We still have to analyze the second term on the \textcolor{black}{right-hand side} of  (\ref{integral1}). So, let us focus on the term with $r^-_N$ which, by  a Taylor expansion of $H$ around the point $0 $, can be written as 
\begin{equation}\label{taylor3}
\frac{\kappa N}{\sqrt{N-1}}\sum_{x \in \Lambda_N}H\left(0\right)r^-_N\left(\tfrac{x}{N}\right)\overline{\eta}_s^N(x)+ \frac{\kappa}{\sqrt{N-1}}\sum_{x \in \Lambda_N}xH'(0)r^-_N\left(\tfrac{x}{N}\right)\overline{\eta}_s^N(x)
\end{equation}
plus lower order terms with respect to $N$. The $L^2(\nu_\rho)$ norm of the  term on the \textcolor{black}{right-hand side} of last display is bounded, using Cauchy-Schwarz inequality, by
\begin{equation*}
\frac{\kappa^2}{N-1}\sum_{x \in \Lambda_N}x^2\left(H'(0)\right)^2\left(r^-_N\left(\tfrac{x}{N}\right)\right)^2\chi(\rho)\lesssim \frac{1}{N}\sum_{x \in \Lambda_{N} }\frac{1}{x^{2\gamma-2}},
\end{equation*}
and since  $2\gamma-2>1$ the sum is convergent and the whole term goes to $0$, as $N$ goes to infinity. Now, let us treat the term on the \textcolor{black}{left-hand side of \eqref{taylor3}}. Putting it together with the term on the \textcolor{black}{left-hand side} of \eqref{taylortheta2}  we get \textcolor{black}{
\begin{equation}
\frac{N}{\sqrt{N-1}}\sum_{x \in \Lambda_N} \Big(H'(0)\Theta_x^-  +\kappa H\left(0\right)r^-_N\left(\tfrac{x}{N}\right)\Big)\overline{\eta}_s^N(x).
\end{equation}}
Thanks to Lemma \ref{lemmaPat1} we can exchange, in the previous expression, $\overline{\eta}_s^N(x)$ by $\overline{\eta}_s^N(1)$ paying an error which vanishes in $L^2(\nu_\rho)$, as $N$ goes to infinity.
Then, by adding and subtracting $m$ and $\sigma^2/2$, the last term is equal to
\begin{equation}
\begin{split}
&\frac{N}{\sqrt{N-1}}\Big\{ H'(0)\Big(\sum_{x \in \Lambda_N}\Theta_x^--\tfrac {\sigma^2}{2}\Big)  +\kappa H\left(0\right)\Big(\sum_{x \in \Lambda_N}r^-_N\left(\tfrac{x}{N}\right)-m\Big)\Big\}\overline{\eta}_s^N(1)\\+&
\frac{N}{\sqrt{N-1}} \Big(H'(0)\tfrac {\sigma^2}{ 2}  +\kappa H\left(0\right)m\Big)\sum_{x \in \Lambda_N}\overline{\eta}_s^N(1).
\label{robin}
\end{split}
\end{equation}
The last term in  the previous display vanishes due to the fact that \textcolor{black}{ $H \in \mc S_{{Rob}}$}. Now, we have to estimate the $L^2(\nu_\rho)$ norm of the remaining term in \eqref{robin} which thanks to Lemma \ref{left1} and the fact that $\theta=1$, it is of order 
\begin{equation*}
\Big(\sum_{x \in \Lambda_N}\Theta_x^--\tfrac {\sigma^2}{2}\Big)^2+\Big(\sum_{x \in \Lambda_N}r^-_N\left(\tfrac{x}{N}\right)-m\Big)^2.
\end{equation*}
Now, using Fubini's theorem, the previous expression is equal to
\begin{equation*}
\Big(\sum_{x \geq N}x^2p(x)\Big)^2+\Big(\sum_{x \geq N}xp(x)\Big)^2,
\label{orderrobin}
\end{equation*}
and since both sums are convergent, the limit as $N$ goes to infinity of the last display is $0$.

Summarizing, we have proved that, for any $H \in \mathcal{S}_{Rob}$,
\begin{equation}
M^N_t(H)=\mathcal{Y}^N_t(H)-\mathcal{Y}^N_0(H)-\int_0^t\frac{\sigma^2}{2}\mathcal{Y}^N_s(\Delta H)ds,
\label{=1}
\end{equation}
plus terms that vanish  in $L^2(\nu_\rho)$, as $N$ goes to infinity.

\subsubsection{Case $\theta >1$}
\label{>11}
In this last case we consider a test function \textcolor{black}{$H \in \mc S_{Neu}$} and we extend it to a smooth function $\widehat H$ defined on $\RR$ which remains bounded.

First of all let us consider the boundary terms in (\ref{integral1}), in particular focus on the one with $r^-_N$, the other one is treated analogously. We can use Lemma \ref{lemma1} to see that its $L^2$ norm is bounded from above by 
\begin{equation}
\begin{split}
&\mathbb{E}_{\nu_{\rho}} \left[\left(\int_0^t \frac{\kappa N^2}{N^\theta\sqrt{N-1}}\sum_{x \in \Lambda_N}H(\tfrac{x}{N})r^-_N\left(\tfrac{x}{N}\right)(\rho-\eta^N_s(x))ds\right)^2\right] \\
&\lesssim \textcolor{black}{\frac{\kappa^2 N}{N^{\theta}}}\sum_{x\in\Lambda_N}r^{-}_N(\tfrac{x}{N})H(\tfrac{x}{N})^2
\lesssim \Vert H \Vert^2_{\infty}N^{1-\theta}\sum_{x\in\Lambda_N}r^{-}_N(\tfrac{x}{N})
\end{split}
\end{equation}
and so it clearly goes to $0$ because the sum converges and $\theta>1$.

Now, let us pass to the term (\ref{integral2}). The first term of the \textcolor{black}{right-hand side} is treated as in the previous case, using Lemma \ref{lemmasigma}, and so it can be replaced by $\frac{\sigma^2}{2}\mathcal{Y}_s^N(\Delta H)$ plus terms that  vanish in $L^2 (\nu_\rho)$, as $N$ goes to infinity. The other two terms at the \textcolor{black}{right-hand side} of (\ref{integral2}) are treated exactly as in the case $\theta=1$ but in an easier way. Indeed, when we arrive at the point (\ref{taylortheta2}) we are done since the term with $H'(0)$ is equal $0$, by the condition imposed on the test function, and the other term goes to $0$ exactly as we showed in the previous case.

Summarizing, we have proved that, for any $H\in \mathcal{S}_{Neu}$,
\begin{equation}
M^N_t(H)=\mathcal{Y}^N_t(H)-\mathcal{Y}^N_0(H)-\int_0^t\frac{\sigma^2}{2}\mathcal{Y}^N_s(\Delta H)ds,
\label{>1}
\end{equation}
plus terms that vanish  in $L^2 (\nu_\rho)$, as $N$ goes to infinity.


\subsection{Convergence at initial time}
The next proposition states that the fluctuation field at time $0$ converges, as $N$ goes to infinity.
\begin{prop}
	The sequence $\{\mathcal{Y}_0^N\}_{N \in \NN}$ converges in distribution to a Gaussian field $\mathcal{Y}_0$ of mean $0$ and covariance given on $H,G \in \mc S_{\theta}$ by
	\begin{equation}
	\mathbb{E}_Q [\mathcal{Y}_0(H)\mathcal{Y}_0(G)]=2 \chi(\rho)\langle H,G\rangle.
	\end{equation}
	\label{initialtime}
\end{prop}
\begin{proof}
	This result can be proved exactly as in \textcolor{black}{Proposition 3 of \cite{PTA}} using characteristic functions, therefore we omit its proof.
\end{proof}

\subsection{Convergence of the martingale} \label{convmart}
Now, we show that  the sequence of martingales $\{M_t^N(H) \; ; \;  t \in [0,T]\}_{N \in\NN}$ converges. This is a consequence of Theorem VIII, 3.12 in \cite{lib} which, in our case, can be written in the following way.
\begin{thm}
	\label{Shir}
	Let $\{M_t^N(H): t \in [0,T]\}_{N> 1}$ be a sequence of martingales living in the space{\footnote{We equip this space with the uniform topology.}} $D ([0,T]; \RR)$ and denote by $\langle M^N(H)\rangle_t$ the quadratic variation of $M_t^N(H)$, for any $N> 1$ and  any $t\in [0,T]$. If we assume that
	\begin{enumerate}
		\item for any $N> 1$, the quadratic variation process $\langle M^N(H)\rangle_t$ has continuous trajectories almost surely;
		\item the following limit holds
		\begin{equation}
		\lim_{N \rightarrow \infty}\mathbb{E}_{\nu_{\rho}}\Big[\sup_{0\leq s \leq T}|M_s^N(H)-M_{s^-}^N(H)|\Big]=0;
		\label{maxjump}
		\end{equation}
		\item for any $t \in [0,T]$ the sequence of random variables $\{\langle M^N(H)\rangle_t\}_{N> 1}$ converges in probability to $t \Vert H \Vert^2_{\theta}$;
	\end{enumerate}
	then the sequence $\{M_t^N(H);  t \in [0,T]\}_{N> 1}$ converges in law in $D([0,T]; \RR)$, as $N$ goes to infinity,  to a mean zero Gaussian process $\{M_t(H)\; ; \;  t \in [0,T]\}$ which is a martingale with continuous trajectories and whose quadratic variation is given by $t \Vert H \Vert_{\theta}^2$, for any $t \in [0,T]$.
	
\end{thm}

So, in order to prove the convergence of the sequence of martingales we just have to prove that the three hypothesis of the previous theorem hold.

\subsubsection{Proof of Hypothesis (2) of Theorem \ref{Shir}}
Observe that using \eqref{mart1} we can write down the following relation
\begin{equation*}
|M_s^N(H)-M_{s^-}^N(H)|=|\mathcal{Y}^N_s(H)-\mathcal{Y}^N_{s^-}(H)|
\end{equation*}
since the function inside the integral in \eqref{mart1} is integrable. Let us now evaluate the expectation \eqref{maxjump} which is then equal to
\begin{equation}
\mathbb{E}_{\nu_{\rho}}\Big[\sup_{0\leq s \leq T}|\mathcal{Y}^N_s(H)-\mathcal{Y}^N_{s^-}(H)|\Big]=\tfrac{1}{\sqrt{N-1}}\mathbb{E}_{\nu_{\rho}}\Big[\sup_{0\leq s \leq T}\big|\sum_{x \in \Lambda_N} H(\tfrac{x}{N})(\eta_s^N(x)-\eta^N_{s^-}(x))\big|\Big].
\label{cont}
\end{equation}
{Note, now that in an infinitesimal time only  one jump occurs and it changes the value of the configuration $\eta$ at two sites: the site from where the particle leaves and the one where the particles arrives. This fact permits to write the sum in previous  supremum as}
\begin{equation*}
H(\tfrac{x}{N})[0-1]+H(\tfrac{y}{N})[1-0]=H(\tfrac{y}{N})-H(\tfrac{x}{N})=H'(\tfrac{z}{N})\tfrac{|x-y|}{N}
\end{equation*} 
for some $z \in (x,y)$ by the mean value theorem. Therefore, since $|x-y|\leq N-1$ and $||H'||_{\infty}$ is finite,   \eqref{cont} is bounded from above by a term of order $1/\sqrt{N}$.
This implies that condition (2) of Theorem \ref{Shir} is satisfied.

\subsubsection{Proof of Hypothesis (1) and Hypothesis (3) of Theorem \ref{Shir}}
First, let us show that hypothesis (1) holds, which is almost trivial from the explicit form of the quadratic variation. Then, we compute the expectation of quadratic variation of the martingales $M_t^N(H)$ for the different regimes of $\theta$. Finally we show that for any $H \in {\mc S}_{\theta}$ the quadratic variation $\langle M^N(H)\rangle_t$ converges in $L^2$ to the value $t||H||^2_{\theta}$, as $N$ goes to infinity, which implies  condition (3). Remember that as we explained in the beginning of Section \ref{sec_car}, through this part of the proof we take $\mathcal{S}_{\theta}=\mathcal{S}$ also in the case $2-\gamma<\theta<1$.

The explicit form of the quadratic variation of the martingale is given by (\textcolor{black}{see Section 5 of \cite{KL}})
\begin{equation}
\langle M^N(H)\rangle_t=\int_0^t \Theta(N)L_N(\mathcal{Y}_s^N(H))^2-2\mathcal{Y}_s^N(H)\Theta(N)L_N(\mathcal{Y}_s^N(H))ds,
\label{qv1}
\end{equation}
and this directly implies that it has continuous trajectories in time, so hypothesis (1) of Theorem \ref{Shir} is satisfied. Now, thanks to the following two lemmas, we can prove hypothesis (3) of the theorem.

\begin{lem}
	For any $\theta \in \RR$ and any $H \in {\mc S}_{\theta}$ we have
	\begin{equation}
	\mathbb{E}_{Q}[M_t(H)^2]:=\lim_{N \rightarrow \infty}\mathbb{E}_{\nu_{\rho}}[M_t^N(H)^2]=t||H||^2_{\theta},
	\label{quadvar}
	\end{equation}
	where $||H||_{\theta} $ was introduced in \eqref{inprod}.
	\label{mean}
\end{lem}
\begin{proof}
	After some computations it is possible to show that  \eqref{qv1} is equal to
	\begin{equation}
	\begin{split}
	\langle M_t^N(H) \rangle=&\int_0^t\dfrac{\Theta(N)}{N} \sum_{x,y \in \Lambda_N}\big(H(\tfrac{x}{N})-H(\tfrac{y}{N})\big)^2p(y-x)(\eta_s^N(x)-\eta_s^N(y))^2\, ds\\
	+&\int_0^t\dfrac{\Theta(N)\kappa}{N^{1+\theta}}\sum_{x \in \Lambda_N}\Big(r^-_N(\tfrac{x}{N})+r^+_N(\tfrac{x}{N})\Big)H(\tfrac{x}{N})^2\Big(\rho+(1-2\rho)\eta_s^N(x)\Big)
	ds.
	\end{split}
	\label{qv}
	\end{equation}
	Now, we use Fubini's theorem in order to pass the expectation in the statement of the lemma inside the integral. In this way, the expectation of the first term at the \textcolor{black}{right-hand side} of \eqref{qv} is given by
	\begin{equation}
	\begin{split}
	\int_0^t \frac{\Theta(N)}{N}\sum_{x,y \in \Lambda_N}\big(H(\tfrac{x}{N})-H(\tfrac{y}{N}))\big)^2p(y-x)\mathbb{E}_{\nu_{\rho}}[(\eta_s^N(x)-\eta_s^N(y))^2]ds\\
	=2 \chi(\rho) t \frac{\Theta(N)}{N}\sum_{x<y \in \Lambda_N}\big(H(\tfrac{x}{N})-H(\tfrac{y}{N})\big)^2p(y-x).
	\end{split}
	\end{equation}
	Now, performing a Taylor expansion on the test function $H$ around $y/N$, we can rewrite the previous expression as \textcolor{black}{
	\begin{equation}
	\label{est_jap}
	\begin{split}
	& \chi(\rho) t \frac{2\Theta(N)}{N^3}\sum_{y \in \Lambda_N}H'(\tfrac{y}{N})^2\sum_{x=1}^{y-1}(y-x)^2p(y-x)\\
	=&2 \chi(\rho) t \frac{\Theta(N)}{N^3}\sum_{y \in \Lambda_N}H'(\tfrac{y}{N})^2\sum_{z=-N+2}^{N-2}z^2p(z),
	\end{split}
	\end{equation}
	plus a term which is absolutely bounded from above by a constant times \begin{equation}\label{est_ll}
	\frac{\Theta(N)}{N^{4}}\sum_{z=-N+2}^{N-2}z^4p(z)\lesssim \frac{\Theta(N)}{N^{\gamma}},
	\end{equation}
	so, it vanishes, as $N$ goes to infinity, in any regime that we consider.
}
	
	We have now two different cases to distinguish \textcolor{black}{in the analysis of  \eqref{est_jap}}.
	\begin{itemize}
	\item When $\theta<2-\gamma$,  $\Theta(N)=N^{\gamma+\theta}$, 
	the order of that term is
	\begin{equation*}
	2 \chi(\rho) t \frac{N^{\gamma+\theta}}{N^2}\frac{1}{N}\sum_{y \in \Lambda_N}H'(\tfrac{y}{N})^2\sum_{z=-N+2}^{N-2}z^2p(z) \lesssim \frac{N^{\gamma+\theta}}{N^2}
	\end{equation*}
	where the inequality follows from the fact that $||H'||_{\infty}$ is bounded and that the sum in $z$ is convergent. So, in this case the term in the last display vanishes, as $N$ goes to infinity, since $\theta < 2-\gamma$.
	\item When $\theta \geq 2-\gamma$, $\Theta(N)=N^2$, we can repeat the same computation above, but at the end, when we pass to the limit, we get that this term converges to $2\chi(\rho)t \sigma^2||H'||^2_{L^2}$ as stated in the lemma.
	\end{itemize}
	
	Now, we are going to analyze the second term in \eqref{qv}. We have to divide the proof in several cases according to the value of $\theta$.
	
	If $\theta \leq 2-\gamma$, using Fubini's theorem and the fact that $\mathbb{E}_{\nu_{\rho}}[\eta_s^N(x)+\rho-2\eta_s^N(x)\rho]=2\chi(\rho)$,  the expectation of the second term in \eqref{qv} is given by
	\begin{equation}
	2\chi(\rho)\kappa t\frac{N^{\gamma}}{N}\sum_{x \in \Lambda_N}(r^-_N(\tfrac{x}{N})+r^+_N(\tfrac{x}{N}))H^2(\tfrac{x}{N}).
	\end{equation}
	Thanks to Lemma \ref{reaction},  the limit, as $N$ goes to infinity, of the previous expression converges to
	\begin{equation}
	2\chi(\rho) \kappa t \int_0^1 (r^-(u)+r^+(u))H(u)^2du.
	\end{equation}
	This ends the proof in the case $\theta\leq 2-\gamma$.
	
	In the case $\theta \in (2-\gamma,1)$, since $\Theta(N)=N^2$,  the expectation of the second term in \eqref{qv} is given by
	\begin{equation}
	2\chi(\rho)\kappa t\frac{N^{2}}{N^{\theta+1}}\sum_{x \in \Lambda_N}(r^-_N(\tfrac{x}{N})+r^+_N(\tfrac{x}{N}))H^2(\tfrac{x}{N})\lesssim N^{-\theta+1-2d}\sum_{x \in \Lambda_{N} }x^{2d-\gamma}\lesssim N^{2-\gamma-\theta},
	\end{equation}
	where the first estimate follows from a Taylor expansion around $0$ of the test function $H \in {\mc S}$ up to the order $d\geq 1$ \textcolor{black}{(recall that, since $H \in \mathcal{S}$, we can always write $|H(\tfrac{x}{N})|=|H(\tfrac{x}{N})-H(0)|\lesssim \big|\tfrac{x}{N}\big|^d$)} and the second follows from a choice of $d$ such that the sum diverges and so it is bounded by a term  of order $N^{2d-\gamma+1}$. Then, since $\theta>2-\gamma$, we can conclude that the term in the previous display goes to $0$, as $N$ goes to infinity.
	
	Let us now treat the case $\theta=1$. The expectation of the second term in \eqref{qv} is given by
	\begin{equation*}
	\begin{split}
	&2\chi(\rho)\kappa t\sum_{x \in \Lambda_N}(r^-_N(\tfrac{x}{N})+r^+_N(\tfrac{x}{N}))H^2(\tfrac{x}{N})\\
	&=2\chi(\rho)\kappa t\Big[H^2(0)\sum_{x \in \Lambda_N}r^-_N(\tfrac{x}{N})+H^2(1)\sum_{x \in \Lambda_N}r^+_N(\tfrac{x}{N})\Big]
	\end{split}
	\end{equation*}
	plus lower order terms with respect to $N$. Then, since the sums on the \textcolor{black}{right-hand side} of the previous expression converge to $m$,  as  $N$ goes to infinity, we can rewrite it, in the limit $N \rightarrow \infty$, as
	\begin{equation}
	2\chi(\rho)\kappa tm[H^2(0)+H^2(1)]=t\frac{\chi(\rho) \sigma^4}{2m\kappa}[H'(0)^2+H'(1)^2]
	\end{equation}
	where the equality follows from the properties of the functions in $\mathcal{S}_{Rob}$ (see \eqref{eq:space_test_funcRob}). This concludes the proof in the case $\theta=1$.
	
	We still have to treat the case $\theta>1$. The expectation of the second term in \eqref{qv} is equal to
	\begin{equation}
	t\chi(\rho)\kappa N^{1-\theta}\sum_{x \in \Lambda_N}(r^-_N(\tfrac{x}{N})+r^+_N(\tfrac{x}{N}))H^2(\tfrac{x}{N})\lesssim t\chi(\rho)\kappa N^{1-\theta}\sum_{x \in \Lambda_N}(r^-_N(\tfrac{x}{N})+r^+_N(\tfrac{x}{N})).
	\end{equation}
	Then, since the sum is convergent and $\theta>1$, the whole term goes to $0$, as $N$ goes to infinity. It ends the proof.
\end{proof}

The next step is to prove that the quadratic variation converges in $L^2$ to its mean.

\begin{lem}
	For any $\theta \in \RR$ and any test function $H \in \mathcal{S}_{\theta}$, the following convergence holds
	\begin{equation}\label{statconv}
	\lim_{N\rightarrow \infty}\mathbb{E}_{\nu_{\rho}}\bigg[\Big(\langle M^N(H)\rangle_t-\mathbb{E}_{\nu_{\rho}}[\langle M^N(H)\rangle_t ]\Big)^2\bigg]=0,
	\end{equation}
	where $\langle M^N(H)\rangle_t$ was defined in \eqref{qv}.
	\label{conv}
\end{lem}
\begin{proof}
	First, using \eqref{qv} and Fubini's theorem, it is easy to evaluate
	\begin{equation}
	\begin{split}
	\mathbb{E}_{\nu_{\rho}}[\langle M^N(H)\rangle_t]&=2\int_0^t\dfrac{\Theta(N)}{N} \sum_{x,y \in \Lambda_N}\big(H(\tfrac{x}{N})-H(\tfrac{y}{N})\big)^2p(y-x)\chi(\rho) ds\\
	&+2\int_0^t\dfrac{\Theta(N)\kappa}{N^{1+\theta}}\sum_{x \in \Lambda_N}\Big(r^-_N(\tfrac{x}{N})+r^+_N(\tfrac{x}{N})\Big)H(\tfrac{x}{N})^2\chi(\rho)
	ds.
	\end{split}
	\end{equation}
	Now, using $(x+y)^2\lesssim x^2+y^2$ and \eqref{qv}, we can bound from above the expectation in \eqref{statconv} by a constant times the sum of 
	\begin{equation}
	\mathbb{E}_{\nu_{\rho}}\Bigg[\bigg(\int_0^t\frac{\Theta(N)}{N}\sum_{x\ne y \in \Lambda_N}\big(H(\tfrac{x}{N})-H(\tfrac{y}{N})\Big)^2p(y-x)
	Z_{x,y} (\eta^N_s)
	ds\bigg)^2\Bigg]
	\label{first}
	\end{equation}
	and
	\begin{equation}
	\mathbb{E}_{\nu_{\rho}}\Bigg[\bigg(\int_0^t \kappa \frac{\Theta(N)}{N^{1+\theta}}\sum_{x \in \Lambda_N}(r^-_N+r^+_N)(\tfrac{x}{N})H(\tfrac{x}{N})^2W_{x} (\eta^N_s) ds\bigg)^2\Bigg],
	\label{second}
	\end{equation}
	where we introduced the random variables $Z_{x,y} (\eta)=(\eta (x)-\eta (y))^2-2\chi(\rho)$ and $W_x (\eta)=\eta (x)+\rho-2\eta (x)\rho-2\chi(\rho)$. Observe that, under $\nu_\rho$, the random variables $(W_x)_x$ are centred and independent and that $(Z_{x,y})_{x\ne y}$ are also centred and uncorrelated for disjoint pairs of indices.   
	
	Let us first analyze \eqref{first}. \textcolor{black}{In order to do that we follow Appendix D.2 of \cite{pat_milr}}. \textcolor{black}{From Cauchy-Schwarz inequality and since the expectation is with respect to the invariant Bernoulli product measure, we  see that  \eqref{first} is bounded from above by a constant times
	\begin{equation}
	\label{newlines}
	\begin{split}
&\frac{\Theta(N)^2}{N^2}\sum_{x\ne y\in \Lambda_N}\sum_{w\in \Lambda_N\backslash\{x\}}\big(H(\tfrac{x}{N})-H(\tfrac{y}{N})\Big)^2\big(H(\tfrac{w}{N})-H(\tfrac{x}{N})\Big)^2p(y-x)p(x-w)\\
& \hspace{8cm}  \times {E}_{\nu_{\rho}}[Z_{x,y} (\eta) Z_{w,x} (\eta)]\\
&+\frac{\Theta(N)^2}{N^2}\sum_{x\ne y\in \Lambda_N}\sum_{t \in \Lambda_N\backslash\{ y\} }\big(H(\tfrac{x}{N})-H(\tfrac{y}{N})\Big)^2\big(H(\tfrac{y}{N})-H(\tfrac{t}{N})\Big)^2p(y-x)p(t-y) \\
&\hspace{8cm} \times {E}_{\nu_{\rho}}[Z_{x,y} (\eta) Z_{y,t} (\eta)].
	\end{split}
	\end{equation}
By using that ${E}_{\nu_{\rho}}[Z_{x,y} (\eta) Z_{w,x} (\eta)]$ is bounded last display can be bounded from above by a constant times
	\begin{equation}
	 \frac{\Theta(N)^2}{N^2}\sum_{x\in \Lambda_N}\Big\{\sum_{y \in \Lambda_N}\Big(H(\tfrac{x}{N})-H(\tfrac{y}{N})\Big)^2p(y-x)\Big\}^2.
	\label{bbound2}
	\end{equation}
By the mean value theorem, \eqref{bbound2} can be bounded from above by a constant times
	\begin{equation}
\frac{\Theta(N)^2}{N^6}\sum_{x\in \Lambda_N}\Big\{\sum_{y \ne x \in \Lambda_N}|y-x|^{1-\gamma}\Big\}^2\leq \frac{\Theta(N)^2}{N^6}\sum_{x\in \Lambda_N}\Big\{\sum_{z=1}^Nz^{1-\gamma}\Big\}^2 \lesssim \cfrac{\Theta (N)^2}{N^5},
\end{equation}
where the penultimate inequality comes from a change of variables and the fact that $x \in \{1,\dots, N-1\}$, and the last one because $\gamma>2$. Hence in any regime of $\theta$ this term vanishes as $N$ goes to infinity.\\
 }
 \textcolor{black}{
To estimate \eqref{second}, since $W_x (\eta)=(1-2\rho) (\eta (x) -\rho) $ it is enough to bound
\begin{equation}
\mathbb{E}_{\nu_{\rho}}\Bigg[\bigg(\int_0^t \kappa \frac{\Theta(N)}{N^{1+\theta}}\sum_{x \in \Lambda_N}(r^-_N+r^+_N)(\tfrac{x}{N})H(\tfrac{x}{N})^2 \left(\eta_s^N (x) -\rho\right) ds\bigg)^2\Bigg].
\label{second2}
\end{equation}
Now, since $(a+b)^2\lesssim a^2 + b^2$ and  from Lemma \ref{lemma1} with $c_N=\kappa \Theta(N) N^{-1-\theta}$, we  bound last display  from above by a constant times
\begin{equation}\label{boundKV}
\kappa \frac{\Theta(N)}{N^{2+\theta}}\sum_{x\in\Lambda_N}\big(r^-_N(\tfrac{x}{N})+r^+_N(\tfrac{x}{N})\big) H(\tfrac{x}{N})^4.
\end{equation}
We have now to distinguish two different cases, according to the value of $\theta$. 
If $\theta>0$, recalling that $\Theta(N)=N^2$, we can bound \eqref{boundKV} by
	$$\frac{\kappa}{N^\theta}\sum_{x\in\Lambda_N}\big(r^-_N(\tfrac{x}{N})+r^+_N(\tfrac{x}{N})\big) H(\tfrac{x}{N})^4\lesssim N^{-\theta},$$ since $H$ is bounded and  the remaining sum is bounded. Thus, we conclude that, in this regime, \eqref{second} vanishes, as $N$ goes to infinity.
If $\theta\leq 0$ we show in details how to estimate \eqref{boundKV} only for the part involving $r^-_N$, but the other one can be analyzed similarly. First  recall \eqref{def:rN} and observe that for $x\neq 1$, $$
r_N^-(\tfrac x N)=\sum_{y\geq x}p(y)\lesssim \sum_{y\geq x}\int_{y-1}^y\frac{1}{t^{1+\gamma}}dt=\int_{x-1}^{+\infty}\frac{1}{t^{1+\gamma}}dt=\frac{1}{\gamma {(x-1)^\gamma}}.$$
Therefore,
\begin{equation}
\label{coreect}
\begin{split}
\kappa \frac{\Theta(N)}{N^{2+\theta}} \sum_{x\in\Lambda_N}r^-_N (\tfrac{x}{N}) H(\tfrac{x}{N})^4\lesssim \frac{\Theta(N)}{N^{2+\theta}}\Big\{ \frac{1}{2}H(\tfrac{1}{N})^4+\sum_{x=1}^{N-2} H(\tfrac{x+1}{N})^4\tfrac{1}{\gamma x^\gamma}\Big\}.
\end{split}
\end{equation}
Now, from the fact that  $H \in {\mathcal S}_\theta={\mathcal S}$ (hence $H^{(k)}(0)=0$ for any $k\ge 0$) and from a Taylor expansion up to an arbitrary order $d\geq 1$ on $H$, we bound the right-hand side of last display by  a constant times
\begin{equation}
\label{boundKV4}
\frac{\Theta(N)}{N^{2+4d+\theta}}\Big(1+\sum_{x=1}^{N-2}x^{-\gamma+4d}\Big)\leq 2\frac{\Theta(N)}{N^{1+\theta+\gamma}} ,
\end{equation} 
where the inequality comes from choosing $d$ such that $4d>\gamma-1$. Now, if $\theta \in (2-\gamma,0]$ then $\Theta(N)=N^2$ and  the previous display equals to $ \frac{N}{N^{\theta+\gamma}}$ which is bounded by $N^{-1}$ and so it vanishes, as $N$ goes to infinity. On the other hand, if $\theta \leq 2-\gamma$, $\Theta(N)=N^{\gamma+\theta}$ and again we get the bound $N^{-1}$.  This ends the proof.}

\end{proof}
\color{black}

\section{Tightness}
\label{s6}
In this section we still consider $\mathcal{S}_{\theta}=\mathcal{S}$ also in the case $2-\gamma <\theta<1$, as we explained in the beginning of Section \ref{sec_car}. Note that all the spaces ${\mc S}_{\theta}$ are Fr\'echet spaces (see Remark \ref{Frec}). Moreover, these spaces, since they are subspaces of the nuclear space $C^\infty([0,1])$, they are also nuclear, see \cite{SchaeferWolff}. Therefore,  in order to prove tightness, we can use Mitoma's criterium. We recall it below. 

\begin{thm}[Mitoma's criterium, \cite{Mitoma}]
	A sequence of processes
	$\{x^N_t ; t \in[0, T ]\}_{N> 1}$ in ${D}([0, T ], {\mc S}_{\theta}')$ is tight with respect to the Skorohod topology if, and only if, the sequence $\{x^N_t (H )\; ;\;  t \in [0, T ]\}_{N> 1}$ of real-valued processes
	is tight with respect to the Skorohod topology of ${D}([0, T ], \RR)$, for any $H\in {\mc S}_{\theta}$.
\end{thm}
This criterium permits to reduce the proof of tightness of the fluctuation field in the Skorohod space ${D}([0,T], {\mc S}_{\theta}')$ to the proof of tightness of the fluctuation field applied to any test function in ${\mc S}_{\theta}$ in ${D}([0,T], \RR)$. From \eqref{mart1} we only have to prove tightness of $\{M_t^N(H)\}_{N> 1}$, $\{\mathcal{Y}_0^N(H)\}_{N> 1}$ and $\{\int_0^t \Theta(N)L_N\mathcal{Y}_s^N(H)ds\}_{N> 1}$ for any test function $H\in \mathcal{S}_{\theta}$. To prove tightness of each of the terms above, we will use Aldous' criterium that we recall below.
\begin{thm}[Aldous' criterium, Proposition 1.6 of \cite{KL}]
	A sequence of real valued processes $\{x^N_t\; ; \;  t \in [0,T]\}_{N \in \NN}$ \textcolor{black}{defined on some probability space $(\Omega, \mathcal F,\mathbb P)$} is tight with respect to the Skorohod topology of ${D}([0,T], \RR)$ if
	\begin{enumerate}
		\item[(i)] \textcolor{black}{$\lim_{A \rightarrow \infty} \limsup_{N\rightarrow \infty} \mathbb{P}(\sup_{t\in [0,T]}|x_t^N|>A)=0$;
		\item[(ii)] for any $\epsilon >0$,  $\lim_{\delta \rightarrow 0} \limsup_{N\rightarrow \infty} \sup_{\lambda<\delta} \sup_{\tau \in \mathcal{T}_T} \mathbb{P}(|x_{\tau+ \lambda}^N-x^N_{\tau}|>\epsilon)=0$;}
	\end{enumerate}
	where $\mathcal{T}_T$ is the space of stopping times bounded by $T$.
	\label{aldou}
\end{thm}

\subsection{Tightness at the initial time}

The tightness of $\{\mathcal{Y}_0^N\}_{N}$ follows from Proposition \ref{initialtime}. 

\subsection{Tightness of the martingale}
We already showed in Theorem \ref{Shir} (see Section \ref{convmart}) that the sequence of martingales $\{M^N_\cdot (H)\}_{N> 1}$ is convergent and so, in particular, it is tight.

\subsection{Tightness of the integral term}

In order to prove the hypothesis of Aldous' criterium for the integral term, we will use Chebychev's inequality in the following way:
\begin{equation}
\begin{split}
& \mathbb{P}_{\nu_{\rho}}\Big(\sup_{t\in [0,T]}\big|\int_0^t\Theta(N)L_N\mathcal{Y}_s^N(H)ds\big|>A\Big)\\
&\leq \frac{1}{{A^2}}\mathbb{E}_{\nu_{\rho}}\Big[\big(\sup_{t\in [0,T]}\int_0^t\Theta(N)L_N\mathcal{Y}_s^N(H)ds\big)^2\Big].
\end{split}
\end{equation}
Then, using Young and Cauchy-Schwarz inequalities, we can bound the expectation on the \textcolor{black}{right-hand side} of the previous display as follows
\begin{equation}
\begin{split}
& \frac{T}{A^2}\int_0^T\mathbb{E}_{\nu_{\rho}}\Big[\big(\Theta(N)L_N\mathcal{Y}_s^N(H)\big)^2\Big]ds\\
&\lesssim \frac{T}{A^2}\int_0^T\Bigg\{\mathbb{E}_{\nu_{\rho}}\Big[\big(\tfrac{\Theta(N)}{\sqrt{N-1}}\sum_{x \in \Lambda_N}\mathcal{L}_NH\left(\tfrac{x}{N}\right)\overline{\eta}_s^N(x)\big)^2\Big]\\&+\mathbb{E}_{\nu_{\rho}}\Big[\big(\tfrac{\kappa\Theta(N)}{N^{\theta}\sqrt{N-1}}\sum_{x \in \Lambda_N}H\left(\tfrac{x}{N}\right)\left(r^-_N\left(\tfrac{x}{N}\right)+r^+_N\left(\tfrac{x}{N}\right)\right)\overline{\eta}_s^N(x)\big)^2\Big]\Bigg\}ds.
\end{split}
\label{cheb}
\end{equation}
So, if we prove that these expectations are uniformly bounded in $N$, item (i) of Theorem \ref{aldou} is satisfied. The proof of item (ii) is similar, the only difference is that in the integral,  $0$ and $t$, will be replaced, respectively,  by $\tau$ and $\tau+\lambda$, so that the terms  are bounded uniformly in $N$, the proof is analogous.

In order to read easily the proof, we divide it according to the different values of $\theta$.

\subsubsection{Case $\theta<2-\gamma$}
We start proving item (i) of Theorem \ref{aldou}. In this case, since $\Theta(N)=N^{\gamma+\theta}$, the bound in \eqref{cheb} is given by
\begin{equation}
\begin{split}
\frac{T}{A^2}\int_0^T&\left\{ \mathbb{E}_{\nu_{\rho}}\Big[\big(\tfrac{N^{\gamma+\theta}}{\sqrt{N-1}}\sum_{x \in \Lambda_N}\mathcal{L}_NH\left(\tfrac{x}{N}\right)\overline{\eta}_s^N(x)\big)^2\Big] \right.\\
&\quad \left.+\mathbb{E}_{\nu_{\rho}}\Big[\big(\tfrac{\kappa N^{\gamma}}{\sqrt{N-1}}\sum_{x \in \Lambda_N}H\left(\tfrac{x}{N}\right)\left(r^-_N\left(\tfrac{x}{N}\right)+r^+_N\left(\tfrac{x}{N}\right)\right)\overline{\eta}_s^N(x)\big)^2\Big]\right\}ds.
\end{split}
\label{boundexp}
\end{equation}
\normalsize
In Section \ref{Secrev} we already proved that the first expectation in the previous expression is bounded (using Cauchy-Schwarz inequality) by a term that goes to $0$, as $N$ goes to infinity. To treat the remaining term in \eqref{boundexp}, we observe that: \textcolor{black}{
\begin{equation}
\begin{split}
&\mathbb{E}_{\nu_{\rho}}\Big[\big(\tfrac{\kappa N^{\gamma}}{\sqrt{N-1}}\sum_{x \in \Lambda_N}H\left(\tfrac{x}{N}\right)\left(r^-_N\left(\tfrac{x}{N}\right)+r^+_N\left(\tfrac{x}{N}\right)\right)\overline{\eta}_s^N(x)\big)^2\Big]\\
&\lesssim  \frac{\kappa^2 N^{2\gamma}}{N-1}\sum_{x \in \Lambda_N}\frac{x^{2d}}{N^{2d}}x^{-2\gamma}\\
&\lesssim  N^{2\gamma-2d-1}\sum_{x \in \Lambda_N}x^{2d-2\gamma}\lesssim 1
\label{bbb}
\end{split}
\end{equation}}
where we used Cauchy-Schwarz inequality and a Taylor expansion of $H$ around $0$ in the first inequality, up to the order $d\geq 1$, chosen in such a way that the sum diverges, so that we can bound it from above by $N^{2d-2\gamma+1}$. This proves that hypothesis (i) of Theorem \ref{aldou} is satisfied.

The proof of hypothesis (ii) is similar, the only difference is that we will have $\lambda/\epsilon^2$ in front of the bound \eqref{cheb} instead of $T/A^2$. So, since the terms  are uniformly bounded in $N$, we can conclude that (ii) holds.

\subsubsection{Case $\theta=2-\gamma$}
\label{2}
Since $\Theta(N)=N^2$ we treat the last expectation in \eqref{cheb} in the same way as we did in the case $\theta<2-\gamma$. Moreover, as we pointed out in Section \ref{secrev2}, using Lemma \ref{lemmasigma}, it is possible to prove that
\textcolor{black}{\begin{equation}
\frac{N^2}{\sqrt{N-1}}\sum_{x \in \Lambda_N}\mathcal{L}_NH\left(\frac{x}{N}\right)\overline{\eta}_s^N(x)=\frac{\sigma^2}{2}\mathcal{Y}^N_s(\Delta H)+o_N(1),
\end{equation}
where $o_N(1)$ is a term that vanishes in $L^2 (\nu_\rho)$ as $N$ goes to infinity.} Hence, since $\Delta H$ is bounded it is easy to see that the $L^2(\nu_\rho)$ norm of the term on the \textcolor{black}{right-hand side} of the previous display is of order $1$. So, we can conclude that the conditions (i) and (ii) of Theorem \ref{aldou} are satisfied.

\subsubsection{Case $\theta \in (2-\gamma,1)$}
Since we are taking functions in $\mc S$ and thanks to Lemma \ref{lemma1}, it is easy to see that the last expectation in \eqref{cheb} can be treated in an analogous way to what we did for the characterization of the limit points (see Section \ref{2-1}). In this way we show that the term goes to $0$, as $N$ goes to infinity and so, in particular, it is uniformly bounded in $N$. 

Also the analysis of the first term in \eqref{cheb} can be done in an analogous way to what we did in Section \ref{2-1}. So, we get that the term inside the expectation is bounded by a term that goes to $0$, as $N$ goes to infinity, plus the $L^2 (\nu_\rho)$ norm of $\frac{\sigma^2}{2}\mathcal{Y}^N_s(\Delta H)$ which is bounded uniformly in $N$ (see Section \ref{2}). This concludes the proof of items (i) and (ii) of Theorem \ref{aldou}.

\subsubsection{Case $\theta=1$}
In this case the test function are in $\mathcal{S}_{Rob}$. Then we can proceed in an analogous way to what we did in Section \ref{1} to analyze the term inside the expectation in the first bound found in \eqref{cheb}. In this way, it is possible to show that it is equal to a term which converges to $\frac{\sigma^2}{2}\mathcal{Y}^N_s(\Delta H)$ in $L^2 (\nu_\rho)$ plus terms vanishing in $L^2 (\nu_\rho)$, as $N$ goes to infinity, plus the following term
\begin{equation}
\sqrt{N}\Big\{\tfrac{\sigma^2}{2}[H'(0)+H'(1)]+m[H(0)-H(1)]\Big\},
\label{mmmj}
\end{equation}
which is equal to $0$ since $H \in \mc S_{Rob}$.
This permits us to conclude that, also in this case, both the hypothesis of Theorem \ref{aldou} are satisfied.

\subsubsection{Case $\theta>1$}
The analysis of the first expectation in \eqref{cheb} is done in a similar way to what we did in Section \ref{>11}, reasoning as in the previous case and using the properties of the test functions $\mathcal{S}_{Neu}$.

In order to prove the fact that also the other two expectations in \eqref{boundexp} are bounded, we have to use Lemma \ref{lemma1} and then, reasoning exactly as we did in Section \ref{>11}, we can show that, not only these terms are bounded uniformly in $N$, but that actually they go to $0$, as $N$ goes to infinity.

Concluding, we have shown that for all the regimes of $\theta$, the conditions of Aldous' criterium (Theorem \ref{aldou}) are satisfied. This, together with Mitoma's criterium, permits to conclude that the sequence 
\begin{equation}
\left\{\int_0^t \Theta(N) L_N \mathcal{Y}_s^N ds \right\}_{N> 1} 
\end{equation}
is tight with respect to the uniform topology of ${D}([0,T],{\mc S}_{\theta}')$.

\section{\textcolor{black}{Proof of Theorems \ref{prop:uniqueness} and \ref{maintheorem}}}
\label{s7}
 \textcolor{black}{In order to prove Theorem \ref{prop:uniqueness} we need to show existence and  uniqueness of the solution of the stochastic partial differential equation associated to the martingale problem stated in this theorem.  We have just proved that the sequence of  fluctuation fields is tight in the Skorohod topology of ${D}([0,T],{\mc S}_{\theta}')$ and to prove Theorem \ref{maintheorem}  we need to show that any limiting point $\mathcal Q$ of this sequence is concentrated on that unique solution. } The case  $\theta \geq1$ is treated in the same  way to what is done in \cite{PTA} and for this reason it is omitted.

We have now to consider  all the  regimes  $\theta <1$. 

\subsection{Case $\theta \in (2-\gamma,1)$}

\subsubsection{\textcolor{black}{$\mathcal Q$ is concentrated on solutions of $OU\Big( \mc S, \tfrac{\sigma^2}{2} \Delta, {\sqrt{2\chi(\rho)}} \sigma \Vert\cdot \Vert_1\Big)$ that satisfy $(i)$ and $(ii)$ of Theorem \ref{prop:uniqueness}}}

Recall that in this regime ${\mathcal S}_\theta ={\mc S}$ and that we proved tightness of the fluctuation field in ${D} ([0,T], {\mc S}^\prime)$ and we showed that the limit points satisfy the martingale problem $OU\Big( \mc S, \tfrac{\sigma^2}{2} \Delta, {\sqrt{2\chi(\rho)}} \sigma \Vert\cdot \Vert_1\Big)$, i.e. for any $H \in {\mc S}$, 
\begin{equation}
\begin{split}
M_t(H)&=\mathcal{Y}_t(H)-\mathcal{Y}_0(H)-\tfrac{\sigma^2}{2} \int_0^t\mathcal{Y}_s(\Delta H)ds, \\
N_t(H) &= M_t^2 (H) -  2\chi(\rho) \sigma^2 \Vert H\Vert_1^2\; t 
\end{split}
\label{M}
\end{equation}
are martingales with continuous trajectories. This follows directly from Lemma \ref{mean}, Lemma \ref{conv} and the computations of Section \ref{2-1}.

The proof of condition $(i)$ is trivial and comes directly from the fact that for any $H\in  \mathcal S$ we have that $\mathcal Y_t^N(H)$ converges in law, as $N \to+\infty$, to $\mathcal Y_t(H)$, $\mathbb E_{\nu_\rho}[(\mathcal Y_t^N(H)^2)]\lesssim \|H\|_{L^2}^2$ (in this inequality we used the fact that $\nu_\rho$ is stationary and product) and also the fact that upper bounds on second moments are preserved by weak convergence.

The proof of condition $(ii)$ is not trivial. Since the functions $\iota^0_{\epsilon}$ and $\iota^1_{\epsilon}$ are not in the set $\mc S$, we have first to define the quantities involved in condition $(ii)$. To this aim let ${\mf H}$ be the Hilbert space of real valued progressively measurable processes $\{x_t; t \in [0,T]\}$ in $L^2$ equipped with the norm $\Vert x \Vert_{\mf H} = \Big(\EE \big[\int_0^T | x_t|^2 dt \big]\Big)^{1/2}$.

\begin{lem}
	\label{lem:ex1-ss}
	Assume (i). For any $H\in L^2$ we can define in a unique way a stochastic process on $\mf H$ denoted, with some abuse of notation, by $\{ {\mc Y}_t (H); t\in [0,T]\}$ which coincides with $\{ {\mc Y}_t (H); t\in [0,T]\}$ for $H\in {\mc S}$. Moreover, the condition (i) holds also if $H\in L^2$.
\end{lem}

\begin{proof}
	Let $H\in L^2$ and let $\{H_\epsilon\}_{\epsilon>0}$ be a sequence of function in ${\mathcal S}$ converging in $L^2$ to $H$, as $\epsilon \rightarrow 0$. By Cauchy-Schwarz inequality and condition $(i)$, the sequence of real valued processes $\{ {\mc Y}_{t} (H_\epsilon) \; ; \; t \in [0,T]\}$ is a Cauchy sequence in $\mf H$ and thus converges to a process (in $\mf H$)  that we denote by $\{ {\mc Y}_{t} (H) \; ; \; t \in [0,T]\}$. It is easy to show that the limiting process depends only on $H$ and not on the approximating sequence $\{H_\epsilon\}_{\epsilon>0}$, justifying the notation. The condition $(i)$ trivially holds for the process defined. 
\end{proof}

\begin{lem}
	The condition (ii) holds.
\end{lem}

\begin{proof}
	Since the proof is the same for $\iota_{\epsilon}^0$ and $\iota_\epsilon^1$, let us prove it only for $\iota_\epsilon^0$. Fix $\epsilon >0$, let $\{H_\epsilon^k\}_{k\ge 0}$ be a sequence of functions in ${\mathcal S}$ converging in $L^2$ to $\iota^0_\epsilon$, as $k$ goes to infinity. By condition (i) and Cauchy-Schwarz inequality we have that
	\textcolor{black}{\begin{equation*}
	\mathbb{E}_Q \Big[\sup_{t\le T} \Big(\int_0^t\mathcal{Y}_s (\iota^0_{\epsilon})ds\Big)^2\Big] \le 2T^2 \Vert H_\epsilon^k -\iota_\epsilon^0\Vert_{L^2}^2 + 2  \mathbb{E}_Q\Big[\sup_{t\le T} \Big(\int_0^t\mathcal{Y}_s (H^k_{\epsilon})ds\Big)^2\Big].
	\end{equation*}}
	For any $H\in {\mc S}$ the application 
	\begin{equation*}
	{\mc Z}\in D([0,T]\; ; \; {\mc S}^\prime) \to \sup_{t\le T} \Big[\int_0^t {\mc Z}_s (H) \Big]^2 \in [0, +\infty)
	\end{equation*}
	is continuous when $D([0,T]; {\mc S}^\prime)$ is equipped with the uniform topology. In particular it is lower semi-continuous and bounded from bellow. Therefore, since $\{\mc Y^N_{\cdot}\}_{N> 1}$ converges in distribution to ${\mc Y}_{\cdot}$, we have that
	\begin{equation*}
	\mathbb{E}_Q \Big[\sup_{t\le T} \Big(\int_0^t\mathcal{Y}_s (H^k_{\epsilon})ds\Big)^2\Big]\le \liminf_{N\to \infty} \mathbb{E}_{\nu_\rho} \Big[\sup_{t\le T} \Big(\int_0^t\mathcal{Y}^N_s (H^k_{\epsilon})ds\Big)^2\Big].
	\end{equation*}
	By Cauchy-Schwarz inequality and stationarity we have 
	\begin{equation*}
	\mathbb{E}_{\nu_{\rho}} \Big[\sup_{t\le T} \Big(\int_0^t\mathcal{Y}^N_s (H^k_{\epsilon})ds\Big)^2\Big] \le 2T^2 \Vert H_\epsilon^k -\iota_\epsilon^0\Vert_{L^2}^2 + 2  \mathbb{E}_{\nu_{\rho}} \Big[\sup_{t\le T} \Big(\int_0^t\mathcal{Y}^N_s (\iota^0_{\epsilon})ds\Big)^2\Big].
	\end{equation*}
\textcolor{black}{By writing  $\iota_\epsilon^0=\iota_\epsilon^0- H_\epsilon^k+ H_\epsilon^k$, the linearity of $\mathcal Y_s$, the inequality $(x+y)^2\leq 2x^2+2y^2$ and Lemma  \ref{lem:ex1-ss}, if follows that
	\begin{equation*}
	\mathbb{E}_Q \Big[\sup_{t\le T} \Big(\int_0^t\mathcal{Y}_s (\iota^0_{\epsilon})ds\Big)^2\Big] \lesssim T^2 \Vert H_\epsilon^k -\iota_\epsilon^0\Vert_{L^2}^2 +\liminf_{N\to \infty} \mathbb{E}_{\nu_{\rho}} \Big[\sup_{t\le T} \Big(\int_0^t\mathcal{Y}^N_s (\iota^0_{\epsilon})ds\Big)^2\Big].
	\end{equation*}}
	Sending $k$ to infinity we get 
	\begin{equation*}
	\mathbb{E}_Q \Big[\sup_{t\le T} \Big(\int_0^t\mathcal{Y}_s (\iota^0_{\epsilon})ds\Big)^2\Big] \lesssim\liminf_{N\to \infty} \mathbb{E}_{\nu_{\rho}} \Big[\sup_{t\le T} \Big(\int_0^t\mathcal{Y}^N_s (\iota^0_{\epsilon})ds\Big)^2\Big]
	\end{equation*}
	and it is thus enough to show that
	$$\lim_{\epsilon \rightarrow 0}\lim_{N\rightarrow \infty}\mathbb{E}_{\nu_{\rho}}\Big[\sup_{t\le T} \Big(\int_0^t\mathcal{Y}_s^N(\iota^0_{\epsilon})ds\Big)^2\Big]=0.$$
	To see this, we can write:
	\begin{equation}
	\begin{split}
	\mathbb{E}_{\nu_{\rho}}\bigg[&\sup_{t\le T}\Big( \int_0^t\mathcal{Y}_s^N(\iota_{\epsilon})ds\Big)^2\bigg]=\mathbb{E}_{\nu_{\rho}}\bigg[\sup_{t\le T} \Big(\int_0^t\frac{1}{\epsilon\sqrt{N-1}}\sum_{x =1}^{\epsilon N}\overline{\eta}_s^N(x)ds\Big)^2\bigg]\\
	&\lesssim \mathbb{E}_{\nu_{\rho}}\bigg[ \sup_{t\le T} \Big(\int_0^t\frac{1}{\epsilon\sqrt{N-1}}\sum_{x \leq \epsilon N}({\eta}_s^N(x)-\eta_s^N(1))ds\Big)^2\bigg]\\
	& \quad +\mathbb{E}_{\nu_{\rho}}\bigg[ \sup_{t\le T} \Big(\int_0^t\sqrt{N-1}\overline \eta_s^N(1)ds\Big)^2\bigg].
	\end{split}
	\end{equation}
	Observe that summing up to $\epsilon N$ has to be read as summing up to $\lfloor \epsilon N \rfloor$, but in order to have a lighter notation, we write above and hereinafter just $\epsilon N$.
	Hence, by applying Lemma \ref{left1} and Lemma \ref{left2}, the last expression, for $\theta \in (2-\gamma,1)$, is of order $\epsilon$ and so it goes to $0$, as $\epsilon \rightarrow 0$, which is exactly condition (ii). \\
\end{proof}

{\color{black}{Up to here we have proved existence of the  solution of the stochastic partial differential equation associated to the martingale problem stated in item 2. of  Theorem \ref{prop:uniqueness}. Moreover, If we prove uniqueness of such solution then Theorem \ref{maintheorem}  is proved in this regime.}}

\subsubsection{{\color{black}{Proof of the uniqueness of the solution of the stochastic partial differential equation associated to the martingale problem in item 2. of  Theorem \ref{prop:uniqueness}.}}}

While we are able to prove uniqueness {\footnote{This follows from Proposition \ref{prop:unique} because the assumption \eqref{eq:exp-OU} is proved in \cite{PTA}.}} of the martingale problem $OU\Big( {\mc S}_{Dir}, \tfrac{\sigma^2}{2} \Delta, {\sqrt{2\chi(\rho)}} \sigma \Vert\cdot \Vert_1\Big)$, we do not know \footnote{And probably it does not hold.} how to prove uniqueness of the solution of this martingale problem for $OU\Big( {\mc S}, \tfrac{\sigma^2}{2} \Delta, {\sqrt{2\chi(\rho)}} \sigma \Vert\cdot \Vert_1\Big)$. Intuitively,  this is because test functions in the set $\mc S$ do not give enough information about the boundary behavior of $\mathcal{Y}_t$. 

So, the idea is the following: take any {\color{black}{solution}} $\{ {\mc Y}_t \in {\mc S}' \; ; \; t\in [0,T]\}$; show that for any time $t\in [0,T]$ there exists ${\tilde{\mc Y}}_t \in {\mc S}_{Dir}^\prime$ extending $\{ {\mc Y}_t \in {\mc S}' \; ; \; t\in [0,T]\}$; prove that the extension belongs to $C([0,T]; {\mc S}_{Dir}^{\prime} )$.  Since we know that there exists a unique (in law) solution of  the martingale problem $OU\Big( {\mc S}_{Dir}, \tfrac{\sigma^2}{2} \Delta, {\sqrt{2\chi(\rho)}} \sigma \Vert\cdot \Vert_1\Big)$, this would imply the uniqueness of the limit point  $\{ {\mc Y}_t \in {\mc S}' \; ; \; t\in [0,T]\}$.

Recall that $\{ {\mc Y}_t \in {\mc S}' \; ; \; t\in [0,T]\}$ satisfies  the two following conditions, stated in Theorem \ref{prop:uniqueness}:
\begin{itemize}
	\item[(i)] \textit{Regularity condition:} $\mathbb{E}[(\mathcal{Y}_t(H))^2]\lesssim ||H||_{L^2}^2$ for any $H \in \mc S$;
	\item[(ii)] \textit{Boundary condition:} for any $\theta \in (2-\gamma,1)$ the following limits hold \begin{equation}
	\lim_{\epsilon \rightarrow 0}\mathbb{E}\left[\sup_{t \in [0,T]}\left(\int_0^t \mathcal{Y}_s (\iota^0_{\epsilon}) ds \right)^2\right]=\lim_{\epsilon \rightarrow 0}\mathbb{E}\left[\sup_{t \in [0,T]}\left(\int_0^t \mathcal{Y}_s (\iota^1_{\epsilon}) ds\right)^2\right]=0.
	\end{equation} 
\end{itemize}

{\color{black}{Observe that in order to make sense of item $(ii)$ above we used the result of Lemma 	\ref{lem:ex1-ss}.}}

Now, let us define the space of functions
$$\tilde{\mc S}:=\{ H \in {\mc H}^2 ([0,1]) :H(0)=H(1)=H'(0)=H'(1)=0\} \subset {\mc H}^2([0, 1]).$$
Observe that if $H \in {\mc H}^2 ([0,1])$ then $H \in C^{1}([0,1])$ so that the previous definition makes sense. This space is a closed vector subspace of ${\mc H}^2([0,1])$ and contains ${\mc S}$. We denote  the extension of ${\mc A}_\theta$ to ${\mc H}^2 ([0,1])$ and the extension of the norm $\Vert\cdot \Vert_\theta$ (which, in this regime, coincides with $2\chi(\rho)\sigma^2||\cdot||_{1}$) to ${\mc H}^2([0,1])$ with the same notation. \\

Let $\mf B$ be the Banach space of real valued processes $\{x_t \; ; \, t \in [0,T]\}$ with continuous trajectories equipped with the norm $\Vert x\Vert_{\mf B}^2 =\EE \big[ \sup_{t \in [0,T]} | x_t|^2 \big]$. Observe that ${\mf B} \subset {\mf H}$ and that $\Vert x \Vert_{\mf H} \le {\sqrt T} \; \Vert x \Vert_{\mf B} $ so that, if a sequence of processes $\{x^\epsilon\}_{\epsilon>0}$ converges in ${\mf B}$ to $x\in {\mf B}$, as $\epsilon \to 0$, then the convergence also holds in ${\mf H}$.   \\

\begin{lem}	
	\label{lema1}
	If the condition (i) is satisfied then, for any $H\in {\tilde{\mc S}}$, we have that 
	\begin{equation}
	\label{eq:martext1}
	{\tilde M}_t(H)={{\mathcal{Y}}}_t(H)-{{\mathcal{Y}}}_0(H)-\int_0^t \mathcal{Y}_s ( {\mathcal{A}}_\theta H )ds
	\end{equation}
	is a martingale with continuous trajectories and with quadratic variation given by $t\to t\Vert H\Vert_{\theta}^2$.	Moreover, the process $\{ {\mc Y}_t (H) \; ; \; t\in [0,T]\}$ has continuous trajectories.
\end{lem}

\begin{rem}
	Since for $H\in \tilde{\mc S}$ the terms $H$ and  ${\mc A}_\theta H$ may not be in $\mc S$, the terms on the right-hand side of \eqref{eq:martext1} are understood in the sense of Lemma  \ref{lem:ex1-ss}.
\end{rem}

\begin{proof}
	Let $H \in \tilde {\mc S}$. By Lemma \ref{lem:appr}, we can approximate $H$ by a sequence $\{H_\epsilon\}_{\epsilon>0}$ in ${\mc S}$  such that $\lim_{\epsilon \to 0} H_\epsilon^{(k)} =H^{(k)}$ in $L^2$ for $k=0,1,2$, so that, in particular, $\Vert H_\epsilon -H\Vert_{\theta} \to 0$, as $\epsilon \to 0$. We have that for any $t \in [0,T]$
	\begin{equation}
	\label{eq:eqsrt}
	M_t (H_\epsilon) ={\mc Y}_t (H_\epsilon) -{\mc Y}_0 (H_\epsilon) -\int_0^t {\mc Y}_s ({\mc A}_\theta H_\epsilon) ds
	\end{equation}
	is a martingale with continuous trajectories and quadratic variation given by $t \to t\Vert H_\epsilon \Vert_\theta^2$.
	We claim first that the sequence of real valued martingales $\{M (H_\epsilon)\}_{\epsilon>0}$ converges in $\mf B$  to a martingale denoted by ${\tilde M} (H)$, and in particular ${\tilde M} (H)$ has continuous trajectories, whose quadratic variation is given by $t\to t\Vert H\Vert_{\theta}^2$. Indeed for any $\epsilon, \epsilon'>0$ observe that by Doob's inequality, we have the bound 
	\begin{equation*}
	\begin{split}
	\Vert M(H_\epsilon) -M(H_{\epsilon'})\Vert_{\mf B}^2 &\textcolor{black}{\lesssim} \sup_{t\in [0,T]} \EE \Big [\big(M_t (H_\epsilon) -M_t (H_{\epsilon'})\big)^2 \Big]\\
	&=\sup_{t\in [0,T]} \EE \Big [\big(M_t (H_\epsilon -H_{\epsilon'})\big)^2 \Big] = T \Vert H_\epsilon -H_{\epsilon'}\Vert_\theta^2
	\end{split} 
	\end{equation*}
	and the last term converges to $0$, as $\epsilon, \epsilon' \to 0$. Hence $\{M (H_\epsilon)\}_{\epsilon}$ is a Cauchy sequence in $\mf B$ and thus converges to a process in $\mf B$ denoted by ${\tilde M} (H)$. We can check easily that the limit depends only on $H$ and not on the approximating sequence $\{H_\epsilon\}_{\epsilon>0}$, so that the notation is justified. Moreover, we have that $\EE \big[ M_t (H)^2\big]=\lim_{\epsilon \to 0} \EE \big[M_t (H_\epsilon)^2\big]=t \lim_{\epsilon \to 0}  \Vert H_\epsilon \Vert_\theta^2 =t \Vert H \Vert_\theta^2$ so that the quadratic variation of ${\tilde M} (H)$ is given by $t\to t\Vert H\Vert_{\theta}^2$.
	
	The process ${\mc Y}(H_\epsilon)$ (resp. ${\mc Y}({\mc A}_\theta H_\epsilon)$) converges in ${\mf H}$ to ${\mc Y} (H)$ (resp. ${\mc Y} ({\mc A}_\theta H)$) by Lemma \ref{lem:ex1-ss}.

	By Cauchy-Schwarz's inequality and condition (i), the sequence of real valued processes \textcolor{black}{$\{ \int_0^{t} {\mc Y}_s ({\mc A}_\theta H_\epsilon)ds; t \in [0,T] \}_{\epsilon >0}$} is a Cauchy sequence in $\mf B$ because $\Vert {\mc A}_\theta H_\epsilon -{\mc A}_\theta H\Vert_{L^2} \to 0$, as $\epsilon \to 0$. Thus it is converging to a process $\{ {\mc Z}_t ({{\mc A}}_\theta H)\; ; \; t\in [0,T]\}$ belonging to $\mf B$ (here again the limiting process is independent of the approximating sequence). Of course the convergence holds also in $\mf H$.
	
	Hence we get from \eqref{eq:eqsrt} that in ${\mf H}$, and consequently  a.e. in time, and $\mathbb P$ a.s., the equality
	\begin{equation*}
	{\tilde M}_t (H) = {{\mc Y}}_t (H) -{{\mc Y}}_0 (H)  - {\mc Z}_t ({{\mc A}}_\theta H)
	\end{equation*}
	holds. Since ${\tilde M}_\cdot (H)$ and ${\mc Z}_\cdot ({\mc A}_\theta H)$ have continuous trajectories this implies that the same holds for the process ${{\mc Y}}_{\cdot} (H)$. It is easy to show that $\mathbb P$ a.s. we have the equality 
	\begin{equation*}
	{\mc Z}_t ({{\mc A}}_\theta H) =\int_0^t {{\mc Y}}_s ({{\mc A}}_\theta H) ds.
	\end{equation*}
	This concludes the proof.
\end{proof}

Now, we are going to construct a particular function which is not in $\mathcal{S}_{Dir}$ neither in $ \mathcal{\tilde S}$ for which we can define properly the martingale problem (see Lemma \ref{milt2}). This intermediate step will allow us to properly define the martingale problem for any test function $H \in \mathcal{S}_{Dir}$ (see Corollary \ref{milt4}). 

Let us consider the function $a:\RR \rightarrow \mathbb{R}$ defined as
\begin{equation}
\forall u\in \RR, \quad   a(u)=c e^{-\tfrac{1}{u(1-u)}}\mathbb{1}_{(0,1)}(u), \; \text{ where } \, \textcolor{black}{c=\left(\int_0^1 e^{-\tfrac{1}{u(1-u)}}du\right)^{-1},}
\end{equation}
and introduce the functions $\phi, {\tilde \phi}:\RR \to \RR$ defined as
\begin{equation}
\label{eq:phi-func}
\forall u\in \RR, \quad \tilde\phi(u):=\int_0^ua(t)dt, \quad  \phi(u):=1-\int_0^ua(t)dt.
\end{equation}

We also define the smooth functions $\psi_{\alpha, \beta}$
\begin{equation}
\label{psi}
u\in [0,1] \to \psi_{\alpha, \beta} (u):=u\phi(\alpha(u-\beta))
\end{equation} 
for some arbitrary chosen  $\beta \in (0,1)$ and  $\alpha>(1-\beta)^{-1}$. We have that $\psi_{\alpha, \beta} (u) = u$ for $u \in [0,\beta]$ and $\psi_{\alpha, \beta}(u)=0$ for $u\in [\beta+\tfrac{1}{\alpha},1]$,  see Figure \ref{psiplot}.

Let $\tilde{\psi}_{\alpha, \beta}:=\psi_{\alpha, \beta} \circ i$, where $i$ is defined for $u \in [0,1]$ as $i(u)=1-u$. 


\begin{lem}
	Assume (i) and (ii). Fix $\beta \in (0,1)$ and $\alpha>(1-\beta)^{-1}$. For $H\in \{\psi_{\alpha, \beta}, {\tilde \psi}_{\alpha, \beta} \}$ we have that 
	\begin{equation}
	\label{eq:martext2}
	{\tilde M}_t(H)={{\mathcal{Y}}}_t(H)-{{\mathcal{Y}}}_0(H)-\int_0^t \mathcal{Y}_s ( {\mathcal{A}}_\theta H )ds
	\end{equation}
	is a martingale with continuous trajectories and with quadratic variation given by $t\to t\Vert H\Vert_{\theta}^2$.	Moreover the process $\{ {\mc Y}_t (H); t\in [0,T]\}$ has continuous trajectories.
	\label{milt2}
\end{lem}

\begin{proof}
	To simplify notation we omit the indexes $\alpha, \beta$. Since the proof for $\tilde \psi$ is similar to the proof for $\psi$ we prove it only for $\psi$. 
	
	For $\epsilon >0$, define the Tanaka function
	\begin{equation}
	h_{\epsilon}(u):=\begin{cases}
	\tfrac{u^2}{2\epsilon}, & \text{ for } u \in [0,\epsilon],\\
	u-\tfrac{\epsilon}{2}, & \text{ for } u \in [\epsilon,1].
	\end{cases}
	\end{equation}
	The Tanaka function belongs to ${\mc H}^2 ([0,1])$. Define now \textcolor{black}{${\phi}_{\alpha,\beta}(u):= \phi(\alpha(u-\beta))$}. It is not difficult to show that, for any $\epsilon>0$ sufficiently small, $\psi_{\epsilon}:=h_{\epsilon}{\phi}_{\alpha,\beta} \in \tilde{\mc S}$ by computing explicitly its derivatives. Moreover, $\{\psi_{\epsilon}\}_{\epsilon>0}$ converges in $\mathcal{H}^1([0,1])$ to $\psi$, as $\epsilon \rightarrow 0$. Indeed, a simple computation shows that
	\begin{equation}
	\lim_{\epsilon \rightarrow 0}\left(\int_0^1 (\psi_{\epsilon}(u)-\psi(u))^2du + \int_0^1 (\psi'_{\epsilon}(u)-\psi'(u))^2du\right)=0.
	\end{equation}
	We know that, for any $\epsilon >0$ sufficiently small, and for any $t\in [0,T]$,
	\begin{equation}
	\label{eq:eqsrt2}
	{\tilde M}_t (\psi_\epsilon) ={\mc Y}_t (\psi_\epsilon) -{\mc Y}_0 (\psi_\epsilon) -\int_0^t {{\mc Y}}_s ({{\mc A}}_\theta \psi_\epsilon) ds
	\end{equation}
	is a real valued martingale with continuous trajectories and quadratic variation given by $t\to t \Vert \psi_\epsilon \Vert_{\theta}^2$. Moreover, the process ${\mc Y}_\cdot (\psi_\epsilon)$ has continuous trajectories.
	
	Since $\lim_{\epsilon \to 0} \Vert \psi_\epsilon -\psi \Vert_\theta = 0$ (which is true because in this regime $||\cdot||_{\theta}=2\chi(\rho)\sigma^2||\cdot||_1$), the sequence of martingales $\{{\tilde M} (\psi_\epsilon)\}_{\epsilon>0}$ converges in ${\mf B}$ to a martingale denoted by ${\tilde M} (\psi)$ with continuous trajectories and quadratic variation given by $t \to t \Vert \psi \Vert_\theta^2$. We have also the convergence in ${\mf H}$ of ${{\mathcal{Y}}}_\cdot (\psi_{\epsilon})$ to ${{\mathcal{Y}}}_\cdot (\psi) \in {\mf H}$, as $\epsilon \rightarrow 0$. This can be done as in the proof of Lemma \ref{lema1}.
	
	It remains to prove the convergence of the process $\{ x^\epsilon_t :=\int_0^t {{\mc Y}}_s ({{\mc A}}_\theta \psi_{\epsilon}) ds\; ; \; t\in [0,T]\}_{\epsilon>0} \in {\mf B}$ to the process $\{ x_t:=\int_0^t {{\mc Y}}_s ({{\mc A}}_\theta \psi) ds\; ; \; t\in [0,T]\}$ in ${\mf B}$. Observe that
	\begin{equation}
	\psi_{\epsilon}''={\phi}_{\alpha,\beta}h_{\epsilon}''+2{\phi}_{\alpha,\beta}'h_{\epsilon}'+h_{\epsilon}{\phi}_{\alpha,\beta}''.
	\end{equation} 
	The last two terms converge in $L^2$ to $\psi''$, as $\epsilon \rightarrow 0$, defined by $\psi''(u)= 2{\phi}_{\alpha,\beta}' (u)+u{\phi}_{\alpha,\beta}'' (u)$ for any $u \in[0,1]$. At the same time we have that ${\phi}_{\alpha,\beta}h_{\epsilon}''=\iota^0_{\epsilon}$ and so by (ii) we have that
	\begin{equation}
	\lim_{\epsilon \rightarrow 0}\mathbb{E} \bigg[ \sup_{t \in [0,T]} \Big(\int_0^t \mathcal{Y}_s({\phi}_{\alpha,\beta}h_{\epsilon}'')ds\Big)^2\bigg] =0.
	\end{equation}
	Hence we get that $\{x^\epsilon\}_{\epsilon>0}$ converges to $x$, as $\epsilon \to 0$, in ${\mf B}$.
	
\end{proof}

\begin{figure}[h]
	\begin{center}
		\includegraphics[scale=0.60]{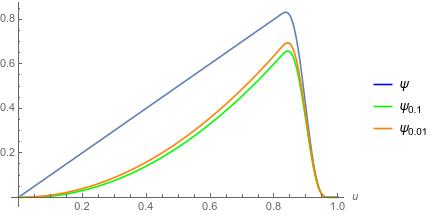}
		\caption{Plot of the functions $\psi_{\alpha,\beta}$ defined in \eqref{psi} and of two of the approximating functions $\psi_{\epsilon}$, for the particular case $\beta=0.5$ and $\alpha=2.3$ and for $\epsilon=0.1, 0.01$.}
		\label{psiplot}
	\end{center}
\end{figure}

We have the following corollary.
\begin{cor}\label{milt4}
	Assume (i) and (ii). Then the process $\{{\mc Y}_t; t\in [0,T]\}$ can be extended to a process $\{{\tilde{\mc Y}}_t \; ; \; t\in [0,T]\}$ belonging to $C([0,T]; {\mathcal S}_{Dir})$ and solution of the martingale problem $OU({\mathcal S}_{Dir}, {\mc A}_\theta, \Vert\cdot \Vert_{\theta})$.
\end{cor}
\begin{proof}
	For any $\beta\in (0,1)$ and any $\alpha>(1-\beta)^{-1}$ we decompose any function $H$ in $\mc S_{Dir}$ as
	\begin{equation}
	H=H'(0)\psi_{\alpha, \beta}-H'(1)\tilde{\psi}_{\alpha, \beta}  + (H-H'(0)\psi_{\alpha, \beta}+H'(1)\tilde{\psi}_{\alpha, \beta}).
	\label{formulation}
	\end{equation}
	Now, using the properties of $H$, $\psi_{\alpha, \beta}$ and $\tilde \psi_{\alpha, \beta}$ it is easy to check that $G_{\alpha, \beta}:=(H-H'(0)\psi_{\alpha, \beta}+H'(1)\tilde{\psi}_{\alpha, \beta} ) \in \tilde{\mathcal{S}}$ and therefore Lemma \ref{lema1} can be applied to this function. Also, note that the functions $H'(0)\psi_{\alpha, \beta}$ and $H'(1)\tilde{\psi}_{\alpha, \beta}$ are suitable for Lemma \ref{milt2}. So, using Lemma \ref{lema1} and Lemma \ref{milt2} we can conclude, by linearity of the martingale problem, that there exists a continuous martingale denoted by ${\tilde M}^{\alpha, \beta} (H) :=H'(0) {\tilde M} (\psi_{\alpha, \beta}) + H' (1) {\tilde M} (\tilde \psi_{\alpha, \beta}) +{\tilde M} (G_{\alpha, \beta})$ with continuous trajectories satisfying 
	\begin{equation}
	{\tilde M}^{\alpha, \beta}_t(H)={{\mathcal{Y}}}_t(H)-{{\mathcal{Y}}}_0(H)-\int_0^t \mathcal{Y}_s ( {\mathcal{A}}_\theta H )ds.
	\end{equation}
	Moreover, the process $\{{\mc Y}_t (H); t\in [0,T]\}$ has continuous trajectories (for any $H \in {\mc S}_{Dir}$)  which implies that the process $\{ {\mc Y}_t; t\in [0,T]\}$ is an element of $C([0,T]; {\mc S}_{Dir}^\prime )$. 
	
	We claim now that by choosing $\beta=1/\alpha$, the sequence of martingales $\{ {\tilde M}^{\alpha, \beta} (H) \}_{\alpha>0}$ converges in ${\mf B}$ as $\alpha \to \infty$ to a martingale ${\tilde M} (H)$ whose quadratic variation is given by $t \to t \Vert H\Vert_{\theta}^2$. This will conclude the proof of the corollary.
	
	The reader will easily check that (recall that $\beta=1/\alpha$) 
	\begin{equation*}
	\lim_{\alpha \to \infty} \Vert \psi_{\alpha, \beta}\Vert_\theta= \lim_{\alpha \to \infty} \Vert \tilde \psi_{\alpha, \beta}\Vert_\theta =0,
	\end{equation*}
	since the support of $\psi_{\alpha,\alpha^{-1}}$ and $\tilde \psi_{\alpha, \alpha^{-1}}$ goes to $\emptyset$, as $\alpha$ goes to infinity. This implies that
	\begin{equation*} 
	\lim_{\alpha \to \infty} \Vert G_{\alpha, \beta} - H \Vert_{\theta} =0.  
	\end{equation*}
	Moreover, recalling that $\langle \tilde M (\psi_{\alpha, \beta} )\rangle_t = t\Vert \psi_{\alpha, \beta}\Vert^2_\theta$, $\langle \tilde M (\tilde\psi_{\alpha, \beta} )\rangle_t = t\Vert \tilde\psi_{\alpha, \beta}\Vert^2_\theta$ and $\langle \tilde M (G_{\alpha, \beta} )\rangle_t = t\Vert G_{\alpha, \beta}\Vert^2_\theta$, we conclude by Doob's inequality  that $\{{\tilde M}^{\alpha, \beta}\}_{\alpha>0}$ is a Cauchy sequence in ${\mf B}$. Thus, it converges to a martingale ${\tilde M} (H)$ in ${\mf B}$ whose quadratic variation satisfies
	\begin{equation*}
	\langle {\tilde M} (H) \rangle_t = \lim_{\alpha \to \infty} \langle {\tilde M}^{\alpha, \beta} (H)\rangle_t = t \Vert H \Vert_\theta^2. 
	\end{equation*}

\end{proof}

This completes the proof of uniqueness of item $2.$ in Theorem \ref{prop:uniqueness}.

\subsection{Case $\theta<2-\gamma$}

In this regime 
$$ {\mc A}_{\theta} =- \kappa V_1, \quad {\mc S}_{\theta}={\mc S}.$$
{\color{black}{Recall that we have  proved tightness of the fluctuation field in ${D} ([0,T], {\mc S}^\prime)$ and we showed that the limit points satisfy the martingale problem $OU({\mc S}, {\mc A}_{\theta}, \Vert\cdot \Vert_\theta)$. This follows directly from Lemma \ref{mean}, Lemma \ref{conv} and the computations of Section \ref{Secrev}. This proves the existence part  in Theorem \ref{prop:uniqueness}. If we prove uniqueness of $OU({\mc S}, {\mc A}_{\theta}, \Vert\cdot \Vert_\theta)$ then Theorem \ref{maintheorem} follows in this regime. }}

To prove uniqueness of $OU({\mc S}, {\mc A}_{\theta}, \Vert\cdot \Vert_\theta)$, we have to show that the assumptions of Proposition \ref{prop:unique} hold. The semigroup associated to ${\mc A}_{\theta}$, for $\theta<2-\gamma$ is defined, for any $H \in \mc S$, $u \in (0,1)$ and any fixed $t>0$, by
\begin{equation}
P_t H(u):=H(u)e^{-t\kappa V_1(u)},
\label{exp_sem}
\end{equation}
with $P_tH(0)=P_tH(1)=0$.
\begin{lem}
	If $H \in \mc S$, then, for any $t>0$, the function $P_tH \in \mc S$.
	\label{prop}
\end{lem}
\begin{proof}
	Follows immediately by the explicit form of $P_tH$ given in \eqref{exp_sem}.
\end{proof}

\begin{lem}
	For any $H\in \mc S$, $P_{t+\epsilon}H-P_tH=-\epsilon \kappa V_1 P_tH+o(\epsilon,t)$ where with $o(\epsilon,t)$ we denote a function in $\mathcal{S}$ which depends on $\epsilon$ and $t$, such that $\lim_{\epsilon \rightarrow 0}o(\epsilon,t)/\epsilon=0$. These limits are uniform in time.
	\label{lem}
\end{lem}
\begin{proof}
	The statement follows directly from the explicit form of the semigroup \eqref{exp_sem}.
\end{proof}

\subsection{Case $\theta =2-\gamma$} 
\label{susubsec:slpbm}
{\color{black}{In this regime $\mathcal S_\theta=\mathcal S$ and  $\mathcal A_\theta=\mathcal{A}_{2-\gamma}$ with 
$\mathcal{A}_{2-\gamma}:=\tfrac{\sigma^2}{2}\Delta -  \kappa V_1.$
 As above, note that we have proved tightness of the fluctuation field in ${D} ([0,T], {\mc S}^\prime)$ and that all limit points satisfy the martingale problem $OU({\mathcal  S}, {\mc A}_{2-\gamma}, \Vert\cdot \Vert_{2-\gamma})$. As in the previous case, this follows from Lemma \ref{mean}, Lemma \ref{conv} and the computations of Section \ref{secrev2} and this  gives the existence   of Theorem \ref{prop:uniqueness}. It remains to prove uniqueness of $OU({\mc S}, {\mc A}_{2-\gamma}, \Vert\cdot \Vert_{2-\gamma})$ and then Theorem \ref{maintheorem} follows. }}

To that end, we will have to prove first the existence of an orthonormal basis of eigenvectors $\{\psi_n\}_{n> 1}$ in $L^2$  and that $\psi_n\in \mathcal S$ for all $n\geq1$. We will also have to get precise estimates on the supremum norm of the derivatives of $\psi_n$.

\begin{prop}
	There exists a self-adjoint unbounded operator $\hat {\mathcal A}_{2-\gamma}$ with domain $$\hat{\mathcal D}\subset \mathcal H^1_0\cap L^2_{V_{1}}$$ containing $\mathcal S$, and  which extends the operator $\mathcal A_{2-\gamma}$, that is,   for any $H\in\mathcal S$ $$\hat {\mathcal A}_{2-\gamma}H={\mathcal A}_{2-\gamma}H.$$ The spectrum of $\hat {\mathcal A}_{2-\gamma}$ is discrete and it is given by a sequence of negative real numbers $$\cdots<-\lambda_n\cdots<-\lambda_2<-\lambda_1$$
	associated to eigenvectors $\psi_n\in \hat{\mathcal D}\subset L^2$. Then $\{\psi_n\}_{n\geq 1}$ forms an orthonormal basis of $ L^2$ and there exists a constant $C>0$ such that for any $n\geq 1$ it holds that $\lambda_n\geq Cn^2.$
\end{prop} 
\begin{proof}
	Let $\mathcal  V:=\mathcal H^1_0\cap L^2_{V_{1}}$ be equipped with the norm $\|\cdot\|_{\mathcal V}$ defined by $$\|H\|_{\mathcal  V}^2=\|H'\|_{L^2}^2+\|H\|^2_{L^2_{{V_1}}}.$$ Consider the bilinear form $\mathcal Q:\mathcal V\times \mathcal  V\rightarrow{\bb R}$ defined on $H,G\in\mathcal  V$ by
	$$\mathcal Q (H,G)=- \frac{\sigma^2}{2}\int_0^1\nabla H(u)\nabla G(u)du-\kappa \int_0^1V_1(u)H(u)G(u)du.$$ Observe that $\mathcal S\subset \mathcal V$ and that $\mathcal  V$ is an Hilbert space dense in $L^2.$  A simple computation shows that:
	
	\begin{enumerate}
		\item $\forall H\in \mathcal V$ it holds   $-\mathcal Q(H,H)\gtrsim \|H\|_{L^2}^2$\quad  \textcolor{black}{(because $V_1\gtrsim 1$)};
		
		\item $\forall H\in \mathcal V  $ it holds $\|H\|_{L^2}\lesssim\|H\|_\mathcal V$; 
		
	\end{enumerate}
	
	Let $$\hat {\mathcal D}=\{H\in \mathcal V: \forall G\in \mathcal  V, |\mathcal Q(H,G)|\lesssim \|G\|_{L^2}\}.$$ It is easy to check that $\mathcal S\subset \hat{\mathcal D}$. We define,  by Riesz representation theorem, $\hat{\mathcal A}_{2-\gamma}H$ for $H\in \hat {\mathcal D}$, i.e. $\forall G\in \mc V $ we have
	\begin{equation}\label{eq:def_hat_A}
	\langle \hat{\mathcal  A}_{2-\gamma}H,G\rangle=\mathcal Q(H,G)
	\end{equation}
	and by density of $\mathcal  V$ in $L^2$, in fact, \eqref{eq:def_hat_A} holds for any $G\in L^2$ since $\mathcal Q(H,\cdot)$ can be extended to $L^2$.
	
	Observe now that $H\in\mathcal S\subset \hat {\mathcal D}$, then we can perform integration by parts in $\mathcal Q(H,\cdot)$ to get that for any $G\in L^2$:
	\begin{equation}
	\mathcal Q(H,G)=\int_0^1\Big(\frac {\sigma^2}{2}\Delta H(u)-\kappa V_1(u)H(u)\Big)G(u)du=\int_0^1(\hat{\mathcal A}_{2-\gamma}H)(u)G(u)du,
	\end{equation}
	so that $\hat{\mathcal  A}_{2-\gamma}H(u)={\mathcal  A}_{2-\gamma}H(u)$ a.e. with respect to the  Lebesgue  measure and, since ${\mathcal  A}_{2-\gamma}H$ is continuous, the equality holds everywhere. Since for any $H\in\hat{\mathcal D}$ it holds
	$$\langle \hat{\mathcal A}_{2-\gamma}H,H\rangle =\mathcal Q(H,H)\leq 0,$$
	with equality if and only if $H=0$, then it is well known (\cite{DF}) that $\hat{\mathcal A}_{2-\gamma}$ has a discrete spectrum composed of strictly negative eigenvalues 
	$$\cdots<-\lambda_n\cdots<-\lambda_2<-\lambda_1<0$$
	associated to an orthonormal basis $\{\psi_n\}_{n\geq 1}$ of eigenvectors (above $\lambda_n>0$ is equal to $-1$ times the $n$-th eigenvalue).
	Now, we finally prove the lower bound on the eigenvalues $\lambda_n$. From Theorem 4.5.3 of \cite{Davies} we know that
	\begin{equation}
	\begin{split}
	\lambda_n=&\inf_{\substack{L\subset \hat{ \mathcal D}\\\textrm{dim}(L)=n}}\sup_{\substack{\varphi \in L\\\|\varphi\|_{L^2}=1}}\Big(\frac{\sigma^2}{2}\int_0^1(\nabla\varphi)^2(u)du+\kappa \int_0^1 V_1(u)\varphi^2(u)du\Big)\\
	\geq& \inf_{\substack{L\subset \hat{ \mathcal D}\\\textrm{dim}(L)=n}}\sup_{\substack{\varphi \in L\\\|\varphi\|_{L^2}=1}}\Big(\frac{\sigma^2}{2}\int_0^1(\nabla\varphi)^2(u)du\Big)  \\
	\geq &\inf_{\substack{L\subset \mathcal {H}^1_0([0,1])\\\textrm{dim}(L)=n}}\sup_{\substack{\varphi \in L\\\|\varphi\|_{L^2}=1}}\Big(\frac{\sigma^2}{2}\int_0^1(\nabla\varphi)^2(u)du\Big) =\frac{\sigma^2}{2}(\pi n)^2.
	\end{split}
	\end{equation}    
	The last equality results from the exact computation of the eigenvalues of the Laplacian with Dirichlet boundary conditions.  This ends the proof. 
\end{proof}

\begin{lem}
	\label{lem:estimatepsi}
	Let $n\geq 1$ and define $u_n:=u_n^\kappa$ as the unique solution in $(0,1/2)$ of the equation
	\begin{equation*}
	\kappa V_1 (u_n) = \lambda_n.
	\end{equation*}
	\begin{enumerate}[1.]
		\item We have that  \textcolor{black}{$\lim_{n\to+\infty} u_n ({\lambda_n})^{1/ \gamma} =(\frac{\kappa c_\gamma}{\gamma})^{1/\gamma}$}.
		\item  For any $m\ge 0$, there exist constants $A_m$ depending on $m$ but independent of $n$ such that for any $u\in [0,u_n]$ (resp. $u\in [1-u_n, 1]$), 
		\begin{equation*}
		\vert \psi_n (u) \vert \le A_m \,  \lambda_n^{m} u^{\gamma m -1/2}, \quad ({\rm{resp.}}\quad  \vert \psi_n (u) \vert \le A_m \,  \lambda_n^{m} (1-u)^{\gamma m -1/2}).
		\end{equation*}
		\item We have that $\psi_n\in\mathcal S$ and for any integers $k, j\ge 0$, there exist positive \textcolor{black}{real} numbers \textcolor{black}{$A_{k,j}$} and $\alpha_{k,j}$ (independent of $n$) such that
		\begin{equation}
		\label{eq:bound-deriv}
		\int_0^1 \big(u^{-k} +(1-u)^{-k}\big)(\psi^{(j)}_n (u))^2 \, du \le A_{k,j}\, \lambda_n^{\alpha_{k,j}}.
		\end{equation}
		In particular for each $j\ge 0$, $\sup_{u\in [0,1]} | \psi_n^{(j)} (u) |$ is bounded by a positive power of $\lambda_n$. 
	\end{enumerate}
\end{lem}

\begin{proof}
	The first item is a direct consequence of the form of $V_1$ (see \eqref{V}) and, in particular, of its behavior around $0$ and $1$. 
	
	First, let us show that $\psi_n$ is smooth on $(0,1)$.
	We use the following characterization of the space $\mathcal H^1({[a,b]})$ where $0<a<b<1$:
	$$f\in {\mc H}^1({[a,b]}) \quad  \textrm {if and only if, \textcolor{black}{for any $G \in C^\infty_c([a,b])$,}}	\quad \Big |\int_a^bf(u)G'(u)du \Big|\lesssim \|G\|_{{ L}^2}.$$
	We know that for any $H,G\in\hat{\mathcal D}$ it holds that
	$$\langle\hat {\mathcal{A}}_{2-\gamma}  H,G\rangle=Q(H,G)=-\frac{\sigma^2}{2}\int_{0}^1\nabla H(u)\nabla G(u)du-\kappa \int_0^1 V_1(u)H(u)G(u)du.$$	
	Recall that $\psi_n\in\mathcal H_0^1\cap L^2_{V_1}\subset L^2$. 
	Now, fix $G\in C^\infty_c([a,b])$, apply the previous identity with $H=\psi_n$ and recall that ${\mc A}_{2-\gamma} \psi_n =-\lambda_n \psi_n$, to get that
	$$-\lambda_n \int_a^b\psi_n(u)G(u)du=\frac{\sigma^2}{2}\int_a^b\psi_n'' (u)G(u)du-\kappa \int_a^b V_1(u)\psi_n(u)G(u)du.$$
	By integration by parts plus the Cauchy-Schwarz inequality we get that\\
	$$\Big|\int_a^b \nabla\psi_n(u)\nabla G(u)du\Big|\lesssim C_{n,a,b}\sqrt {\int_a^b G^2(u)du}$$
	and this shows  that $\psi_n^{'} \in \mathcal H^1({[a,b]})$. This implies, in particular, that $\psi_n^{'}$ is  continuous. Now, repeating the same argument replacing $G$ with $G'$, we get the same result for $\psi^{''}_n$, that is $\psi_n^{''} \in \mathcal H^1({[a,b]})$ which implies that $\psi_n^{''}$ is  continuous. Iterating this argument we prove that $\psi_n$ is smooth on $(0,1)$, since $0<a<b<1$ are arbitrary.\\

	Now, we use a semi-classical analysis argument to show the second item. Since everything is symmetric with respect to $1/2$ ($V_1 (\cdot) =V_1 (1-\cdot)$) it is sufficient to prove only the first bound on $[0,u_n]$.  Let $h\in (0,1)$ and define $g_h^n:=g_h:(0,1/2h) \to \RR$ by $g_h (\cdot) =\psi_n (h \cdot)$. Since ${\mc A}_{2-\gamma} \psi_n =-\lambda_n \psi_n$, the function $g_h$ satisfies, $\forall u \in (0,1/2h),$
	\begin{equation*}
	\begin{split}
	& \dfrac{\sigma^2}{2}g_h^{\prime\prime} (u) = \dfrac{\sigma^2}{2} h^{2}\psi_n''(hu)=h^2(\kappa V_1(hu)-\lambda_n)g_h(u)=\tfrac{c_\gamma \kappa}{\gamma}  h^{2-\gamma}u^{-\gamma}g_h(u)+\epsilon_h(u)
	\end{split}
	\end{equation*}
	where on $(0,1/2h)$, $|\varepsilon_h| \le  [\tfrac{2^\gamma c_\gamma \kappa}{\gamma} +\lambda_n] \, h^{2} | g_h |$. Taking now $\tilde h = h^{\gamma/2 -1}$ we obtain
	\begin{equation*}
	\forall u \in (0, 1/2h), \quad \tilde h^2 g_h^{\prime \prime} (u) = 2\kappa\sigma^{-2}  u^{-\gamma} g_h (u) +R_h (u)
	\end{equation*}
	where 
	\begin{equation*}
	\forall u \in (0, 1/2h), \quad |R_h (u) | \le 2 \sigma^{-2} \Big[\tfrac{2^\gamma c_\gamma \kappa}{\gamma} +\lambda_n\Big]\,  h^\gamma  |g_h (u) |.
	\end{equation*}
	Let us consider some open intervals \textcolor{black}{${\mc U}\subset \bar {\mc U}\subset  {\mc W} \subset (0,1/2)$} and let $\xi$ be a smooth cutoff function vanishing outside $(0,1)$ and equal to $1$ on $\mc W$. Let $g_{\tilde h} =\xi g_h$ which is a smooth function equal to $g_h$ on $\mc W$. We can use now Theorem 7.3 in \cite{lib2} for $g_{\tilde h} \in C_c^\infty (\RR)$ with $V(u)=2\kappa\sigma^{-2}  u^{-\gamma}$ and $E=\lambda=0$. We get that there exists a constant $C:=C({\mc U},{\mc W})>0$ such that if $\tilde h \le 1/C$ then 
	\begin{equation*}
	\sqrt{\int_{\mc U} g_{\tilde h}^2 (u) du} \le Ce^{-1/(C\tilde h)} \; \sqrt{\int_{\mc W} g_{\tilde h }^2 (u) du} +C \sqrt{\int_{\mc W} [\tilde h^2 g_{\tilde h}^{\prime\prime} -V  g_{\tilde h} ]^2 (u) du}.
	\end{equation*}
	Observe that on ${\mc W}$ (and thus on ${\mc U}$) we have the following equalities 
	\begin{equation*}
	g_{\tilde h }=g_h, \quad \tilde h^2 g_{\tilde h}^{\prime\prime} - V g_{\tilde h} = \tilde h^2 g_{h}^{\prime\prime} - V g_{h} = R_h.
	\end{equation*}
	Hence we get 
	\begin{equation*}
	\| g_h\|_{L^2 ({\mc U})} \le C \left\{e^{-1/(C\tilde h)} + 2 \sigma^{-2} \Big[\tfrac{2^\gamma c_\gamma \kappa}{\gamma} +\lambda_n\Big]  \; h^{\gamma} \right\} \| g_h\|_{L^2({\mc W})}.
	\end{equation*}
	It follows that since $\lambda_n \to +\infty$, for some constant {\footnote{$C^{'}= C \sup_{n\ge 1} \sup_{0<h<1} \left[ \lambda_{n}^{-1} h^{-\gamma} e^{-1/(C\tilde h)} + 2 \sigma^{-2} \Big[\tfrac{2^\gamma c_\gamma \kappa}{\gamma}\lambda_{n}^{-1}  + 1\Big] \right] $.}} $C':=C' ({\mc U},{\mc  W})$, we have 
	\begin{equation*}
	\| g_h\|_{L^2 ({\mc U})} \le C' \; \lambda_n \; h^{\gamma} \| g_h\|_{L^2({\mc W})}.
	\end{equation*}
	By iterating this argument with $\mc U$ playing the role of $\mc W$ and so on, we get for any $m\ge 0$ and any open interval $\mc U=(a,b) \subset (0,1/2)$ that there exists a constant $C_m:=C_m (\mc U)$  such that
	\begin{equation*}
	\| g_h\|_{L^2 ({\mc U})} \le C_m \,  (\lambda_n h^{\gamma})^m \, h^{-{1/2}},
	\end{equation*}
	because for any ${\mc W} \subset (0,1/2)$ we have that
	\begin{equation*}
	\| g_h\|_{L^2({\mc W})}\le h^{-{1/2}} \| \psi_n\|_{L^2(0,1)} =  h^{-1/2}.
	\end{equation*}
	Equivalently we obtained that
	\begin{equation}
	\label{eq:convmax}
	\int_{ha}^{hb} \psi_n^2  (u)\;  du \le C_{m}^2\,  \lambda_n^{2m} h^{2\gamma m}.
	\end{equation}
	Observe now that since $\psi_n$ satisfies 
	\begin{equation*}
	\frac{\sigma^2}{2}\psi_n '' (u) =\kappa V_1(u)\psi_n (u)-\lambda_n  \psi_n(u), \quad u \in (0,1),
	\end{equation*}
	we have that for any $u\in (0, u_n]$, $[\psi_n^2]^{\prime\prime} (u) =2[\psi_n^{\prime}]^2 (u)  +\tfrac{4}{\sigma^2}\psi_n^2 (u) \, (\kappa V_1 (u) -\lambda_n)  \ge 0$, i.e. $\psi_n^2$ is a convex function on $(0, u_n]$. Hence (\ref{eq:convmax}) and Jensen's inequality imply that for any $m\ge 0$, $h \in(0,1)$ such that $hb\le u_n$, 
	\begin{equation*}
	\psi_n^2 \left( \cfrac{h (b+a)}{2}\right) \le \tfrac{1}{b-a}\, C_m^2\,  \lambda_n^{2m} h^{2\gamma m -1}.
	\end{equation*}
	Take for example $a=1/4, b=1/3$ to conclude that for any $m\ge 0$, there exists a constant $A_m$ such that
	\begin{equation*}
	\forall u\in (0, u_n], \quad \psi_n^2 (u) \le A^2_m \,  \lambda_n^{2m} u^{2\gamma m -1}.
	\end{equation*}
	
	This proves the second item and consequently that in a neighborhood of $0$, $|\psi_n (u)|$ is smaller than any power of $u$, hence that all its derivatives tend to $0$.\\
	
	For the third item, it remains to prove \eqref{eq:bound-deriv} since the last part of the item follows from the previous one for $k=0$, the equality $\psi_n^{(j)} (u) =\psi_n^{(j)} (u) -\psi_n^{(j)} (0) =\int_0^u \psi_n^{(j+1)} (v) dv$ and Cauchy-Schwarz's inequality. The proof of \eqref{eq:bound-deriv} is done by induction on $j$. The case $j=0$ results from the first and second items, by splitting the integration interval according to $[0,u_n]\cup [u_n, 1-u_n]\cup[1-u_n, 1]$, and recalling that $\int_0^1 \psi_n^2 (u) du =1$. To prove the induction step, let us call $W_k (u)=u^{-k}+ (1-u)^{-k}$. From integration by parts, since $\psi_n \in {\mc S}$ we have
	\begin{equation*}
	\begin{split}
	& \int_0^{1} W_k (u) (\psi_n^{(j+1)} (u))^2 du = -\int_0^{1} \tfrac{d}{du} (W_k (u) \psi_n^{(j+1)} (u))   \psi_n^{(j)} (u) \,  du\\
	&=  -\int_0^{1} W_k^{'} (u) \psi_n^{(j+1)} (u)\,   \psi_n^{(j)} (u) \,  du -\int_0^{1} W_k (u) \psi_n^{(j+2)} (u)\,   \psi_n^{(j)} (u) \,  du\\
	&=\cfrac{1}{2} \int_0^{1} W_k^{''} (u) [\psi_n^{(j)} (u)]^2 \,  du -\int_0^{1} W_k (u) \psi_n^{(j+2)} (u)\,   \psi_n^{(j)} (u) \,  du.
	\end{split}
	\end{equation*}
	Since $\tfrac{\sigma^{2}}{2} \psi_n^{(2)} =(\kappa V_1 -\lambda_n) \psi_n$, we have
	\begin{equation*}
	\tfrac{\sigma^{2}}{2} \psi_n^{(j+2)} = \sum_{m=0}^{j}\,  \begin{pmatrix} j\\m\end{pmatrix} \, \Big[ \kappa V_1 -\lambda_n \Big]^{(j-m)}\,  \psi_n^{(m)}.
	\end{equation*}
	It holds that $(\kappa V_1 -\lambda_n)^{(j-m)}$ is a linear combination of $\{W_p\, \; \, p\ge 0\}$, the coefficients being polynomial functions of $\lambda_n$. Hence by the induction hypothesis and Cauchy-Schwarz's inequality we get the result.
\end{proof}

\begin{definition}
	For any $t\geq 0$ and any $ H \in L^2$  we define the semigroup $\{P_t\}_{t\geq 0 }$ associated to the operator $\hat{\mathcal A}_{2-\gamma}$ by  $$P_tH=\sum_{n\geq 1}e^{-\lambda_n t }H_n\psi_n,$$ where $H_n=\langle H,\psi_n \rangle$. 
\end{definition}

\begin{prop}
	If $H\in\mathcal S$ and $t\geq 0$, then $P_t H\in\mathcal S$ and we have the expansion 
	$$P_{t+\epsilon}H-P_tH= \epsilon\mathcal A_{2-\gamma} P_tH+o(\epsilon,t)$$ where $o(\epsilon,t)$ denotes a function in $\mc S$ depending on $\epsilon$ and $t$ such that 
	$$\lim_{\epsilon \rightarrow 0}\epsilon^{-1}o(\epsilon,t)=0$$
	with respect to the topology on $\mathcal S$ induced by the seminorms defined in \eqref{seminorms}, uniformly \textcolor{black}{in time.}
\end{prop}

\begin{proof}
	%
	Observe that if $H\in\mathcal S$ we have that:
	$$H_n=\langle H, \psi_n\rangle=\frac{1}{\lambda_n}\left[\kappa \int_0^1V_1(u)H(u)\psi_n(u)du-\frac{\sigma^2}{2}\int_0^1\Delta \psi_n(u)H(u)du\right].$$
	Moreover, since $H\in\mathcal S$, then $V_1H\in L^2$
	and $\Delta H\in L^2$, so that, by using an integration by parts argument and Cauchy-Schwarz inequality, we get that $|H_n|\lesssim \frac{1}{\lambda_n}.$ By iterating the same argument several times we get, for any $j\geq 1$,  that $|H_n|\lesssim \frac{1}{\lambda^j_n}$. From this, the third item in Lemma \ref{lem:estimatepsi} and the fact that $n^2 \lesssim \lambda_n$ it follows that 
	$$P_t H =\sum_{n\geq 1}e^{-\lambda_n t}H_n\psi_n$$ 
	is a smooth function on $[0,1]$ such that for any $j\ge 0$, we can exchange the derivative with the series, i.e.
	\begin{equation*}
	(P_t H)^{(j)} = \sum_{n\geq 1}e^{-\lambda_n t}H_n\psi_n^{(j)}.
	\end{equation*}
	So, in particular, $P_t H \in {\mc S}$. Observe also that for any $t\ge0$ and $\epsilon\ge 0$, we have that
	\begin{equation*}
	P_{t+\epsilon}H-P_tH- \epsilon\mathcal A_{2-\gamma} P_tH = \sum_{n\ge 1} e^{-\lambda_n t} H_n \left[e^{-\lambda_n \epsilon} -1 + \lambda_n \epsilon \right] \psi_n.
	\end{equation*}
	We have now to remember that the coefficients $\{H_n\}_{n\ge 1}$ decay to $0$ faster than any negative power of the $\lambda_n$'s, while the supremum norm of the derivatives of $\psi_n$ can be bounded by a positive power of the $\lambda_n$'s and that
	$$\epsilon^{-2} \left\vert e^{-\lambda_n \epsilon} -1 + \lambda_n \epsilon \right\vert \le C\, \lambda_n^2  $$
	where 
	\begin{equation*}
	C=\sup_{x\ge 0} \left\vert \cfrac{e^{-x} -1+x }{x^2}\right\vert <\infty.
	\end{equation*}
	It results that the term   
	\begin{equation*}
	\epsilon^{-1} \sum_{n\ge 1} e^{-\lambda_n t} H_n \left[e^{-\lambda_n \epsilon} -1 -\lambda_n \epsilon \right] \psi_n
	\end{equation*}
	converges in ${\mc S}$ to $0$, uniformly in time.
\end{proof}

\color{black}
\section{Replacement lemmas}
\label{s8}
Let us introduce some important quantities that we will use in the next results. For a probability measure $\mu$ on $\Omega_N$ and a function $f \in L^2(\mu)$ we introduce  
\begin{eqnarray} 
\label{left_rig_form}
D_{N}^{0}({f},\mu):=\tfrac{1}{2}\sum_{x,y\in\Lambda_N}p(y-x)\, I_{x,y}({f},\mu), 
\end{eqnarray}
\begin{eqnarray}
\label{left_dir_form}
D_{N}^{l}({f},\mu):=\sum_{x\in\Lambda_N}r_N^-(\tfrac{x}{N})I^\alpha_{x}\, ({f},\mu),
\end{eqnarray} 
\begin{eqnarray}
\label{right_dir_form}
D_{N}^{r}({f},\mu):=\sum_{x\in\Lambda_N}r_N^+(\tfrac{x}{N})I^\beta_{x}\, ({f},\mu).
\end{eqnarray} 
Above, we used the following notation
\begin{eqnarray*}
	I_{x,y}(f,\mu)&:=& \int \left({f(\sigma^{x,y}\eta)}- {f(\eta)}\right)^{2} d\mu(\eta),\\
	I_{x}^{\delta}( f,\mu)&:=& \int  c_{x}(\eta;\delta)\left( {f(\sigma^{x}\eta)}-{f(\eta)}\right)^{2} d\mu(\eta), \quad \delta\in\{\alpha, \beta\}.
\end{eqnarray*}

We will use the relation
\begin{equation}\label{diiirest}
\langle L_N {f},{f} \rangle_{\nu_{\rho}}= -D_{N}({f},\nu_{\rho} )
\end{equation}
between the Dirichlet form $\langle L_N {f}, {f} \rangle_{\nu_{\rho}}:=\int {f}(\eta) L_N {f}(\eta)d\nu_{\rho}(\eta)$ and the quantity
$$
D_{N}({f},\nu_{\rho} ):= (D_{N}^{0}+\kappa N^{-\theta} D_{N}^{l}+\kappa N^{-\theta}D_{N}^{r})({f},\nu_{\rho}).$$
See \cite{BGJO} \textcolor{black}{(inequality (5.4))} for further details about \textcolor{black}{\eqref{diiirest}}, observing that, in this case, the equality holds because we are considering the Bernoulli product measure $\nu_{\rho}$ and $\rho=\alpha=\beta$.

Now, we are ready to present some technical results, with the respective proofs, which we used in this work.

\begin{lem}
	Fix  $t>0$. Recall the definition of $r_N^\pm(\cdot)$ from \eqref{def:rN}. Then
	\begin{equation}
	\begin{split}
	& \mathbb{E}_{\nu_{\rho}} \bigg[\sup_{t\in [0,T]}\Big(\int_0^t c_N \sum_{x \in \Lambda_N}H(\tfrac{x}{N})r^\pm_N(\tfrac{x}{N})(\rho-{\eta}_s^N(x))ds\Big)^2\bigg] \\
	&\lesssim \frac{c_N^2N^\theta}{\kappa\Theta(N)}\sum_{x\in\Lambda_N}r^{\pm}_N(\tfrac{x}{N})H(\tfrac{x}{N})^2,
	\end{split}
	\label{aa}
	\end{equation}
	where $c_N$ is a positive constant eventually depending on $N$.
	\label{lemma1}
\end{lem}
\begin{proof}
	We prove the result just for $r^-_N$ since for $r^+_N$ it is analogous. 
	Using Lemma 4.3 of \cite{changlandimolla}, we can bound from above the expectation in the statement of the lemma by \textcolor{black}{a constant times
	\begin{equation}	\label{KV}
	\begin{split}
	T\sup_{f \in L^2({\nu_{\rho})}} &\bigg\{ 2\int c_N \sum_{x \in \Lambda_N}H(\tfrac{x}{N})r^-_N(\tfrac{x}{N})(\rho-{\eta}^N(x))f(\eta)d\nu_{\rho}(\eta) + \Theta(N)\langle L_Nf,f\rangle_{\nu_{\rho}} \bigg\}.
	\end{split}
	\end{equation}}
	We can rewrite the first term inside the supremum  above as twice its half then, in one of the terms we make the change of variables $\eta \mapsto\sigma^x \eta$. Since 
	\begin{equation*}
	\begin{split}
	\left(\frac{d\nu_{\rho}(\sigma^x\eta)}{d\nu_{\rho}(\eta)}\right)(\rho-(1-\eta(x)))=\overline{\eta}(x),
	\end{split}
	\end{equation*}
	we write that term as 
	\begin{equation*}
	\int c_N \sum_{x \in \Lambda_N}H(\tfrac{x}{N})r^-_N(\tfrac{x}{N})\overline{\eta}(x)(f(\sigma^x\eta)-f(\eta))d\nu_{\rho}(\eta).
	\end{equation*}
	Now we can use Young's inequality to bound from above last expression by
	\begin{equation}\label{proof2}
	\frac 12 \int c_N\sum_{x\in\Lambda_N}H(\tfrac{x}{N})r^-_N(\tfrac{x}{N})\left[B_x\overline{\eta}(x)^2+\dfrac{1}{B_x}(f(\eta)-f(\sigma^x\eta))^2\right]d\nu_{\rho}(\eta)
	\end{equation}
	for any $B_x>0$. Taking 
	\begin{equation*}
	B_x=\frac{M_\rho c_N |H(\tfrac{x}{N})| N^\theta}{2 \Theta(N)\kappa},
	\end{equation*}
	where $M_{\rho}=(\min \{\rho, 1-\rho\})^{-1}$,
	the last term in (\ref{proof2}) can be bounded from above by
	\begin{equation*}
	\frac{\kappa \Theta(N)}{M_\rho N^\theta}\sum_{x\in\Lambda_N}c_x(\eta;\rho)^{-1}r^-_N(\tfrac{x}{N})I_x^{\rho}(f,\nu_{\rho}).
	\end{equation*}
	Since $|\eta(x)|\leq 1$, the previous expression can be bounded from above by 
	$
	- \Theta(N)\langle L_Nf,f\rangle_{\nu_\rho}.
	$
	So, \eqref{proof2} is bounded from above by 
	\begin{equation*}
	\frac{M_\rho c_N^2N^\theta}{4\kappa\Theta(N)}\sum_{x\in\Lambda_N}\int r^-_N(\tfrac{x}{N})H(\tfrac{x}{N})^2\overline{\eta}(x)^2d\nu_{\rho}(\eta)-\Theta(N)\langle L_Nf,f\rangle_{\nu_\rho}.
	\end{equation*}
	This means that \textcolor{black}{ \eqref{KV}} is bounded from above by
	\begin{equation*}
	T\frac{ M_\rho c_N^2N^\theta}{4\kappa\Theta(N)}\sum_{x\in\Lambda_N}\int r^-_N(\tfrac{x}{N})H(\tfrac{x}{N})^2\overline{\eta}(x)^2 d\nu_{\rho}(\eta) \lesssim \frac{c_N^2N^\theta}{\kappa\Theta(N)}\sum_{x\in\Lambda_N} r^-_N(\tfrac{x}{N})H(\tfrac{x}{N})^2.
	\end{equation*}
\end{proof}

%

\begin{lem}
	For $\theta \geq 2-\gamma$ and $x \in \{1,N-1\}$:
	\begin{equation}
	\mathbb{E}_{\nu_{\rho}}\bigg[\sup_{t \in [0,T]}\Big(\int_0^t\sqrt{N-1}\overline{\eta}_s^N(x)ds\Big)^2\bigg] \lesssim {N^{\theta-1}}.
	\end{equation}
	\label{left1}
\end{lem}
\begin{proof}
	The proof is analogous to the proof of Lemma \ref{lemma1} taking $c_N=(\sqrt{N-1})^{-1}$, $H$ constantly equal $1$ and bounding the functions $r^{\pm}_N$ from above by $1$.
\end{proof}

\begin{lem}
	For $\theta \geq 2-\gamma$:
	\begin{equation}
	\mathbb{E}_{\nu_{\rho}}\bigg[\sup_{t \in [0,T]}\Big(\int_0^t\frac{1}{\epsilon\sqrt{N-1}}\sum_{\substack{x \in \Lambda_N \\x \leq \epsilon N}}({\eta}_s^N(x)-\eta_s^N(1))ds\Big)^2\bigg]\lesssim \epsilon.
	\end{equation}
	\label{left2}
\end{lem}
\begin{proof}
	First of all observe that we can write, by a telescopic argument,
	\begin{equation*}
	\begin{split}
	&\mathbb{E}_{\nu_{\rho}}\bigg[\sup_{t \in [0,T]}\Big(\int_0^t\frac{1}{\epsilon\sqrt{N-1}}\sum_{\substack{x \in \Lambda_N \\x \leq \epsilon N}}({\eta}_s^N(x)-\eta_s^N(1))ds\Big)^2\bigg]\\
	&=\mathbb{E}_{\nu_{\rho}}\bigg[\sup_{t \in [0,T]}\Big(\int_0^t\frac{1}{\epsilon\sqrt{N-1}}\sum_{\substack{x \in \Lambda_N \\x \leq \epsilon N}}\sum_{y =1}^{x-1}({\eta}_s^N(y+1)-\eta_s^N(y))ds\Big)^2\bigg].
	\end{split}
	\end{equation*}
	Then using Lemma 4.3 of \cite{changlandimolla}, we can bound the previous expectation by a constant times
	\begin{equation}
	\begin{split}
	T\sup_{f \in L^2 (\nu_{\rho}) }&\bigg\{\int \tfrac{1}{\epsilon\sqrt{N-1}}\sum_{\substack{x \in \Lambda_N \\x \leq \epsilon N}}\sum_{y =1}^{x-1}({\eta}^N(y+1)-\eta^N(y))f(\eta)d\nu_{\rho}(\eta)-\Theta(N)\langle L_Nf,f\rangle_{\nu_{\rho}} \bigg\}.
	\end{split}
	\end{equation}
	Now, we can perform a change of variables on the first term inside the supremum above to rewrite it as
	\begin{equation}\label{new_eq}
	\int\frac{1}{\epsilon\sqrt{N-1}}\sum_{\substack{x \in \Lambda_N \\x \leq \epsilon N}}\sum_{y =1}^{x-1}({\eta}^N(y+1)-\eta^N(y))(f(\eta)-f(\sigma^{y,y+1}\eta))d\nu_{\rho}(\eta).
	\end{equation}
	We can bound this term, using Young's inequality, by 
	\begin{equation}	\label{1lemmaleft}
	\int \frac{1}{\epsilon \sqrt{N-1}}\frac{B}{2}\sum_{\substack{x \in \Lambda_N \\x \leq \epsilon N}}\sum_{y =1}^{x-1}({\eta}^N(y+1)-\eta^N(y))^2d\nu_{\rho}(\eta),
	\end{equation}
	plus
	\begin{equation}\label{2lemmaleft}
	\int \frac{1}{\epsilon \sqrt{N-1}}\frac{1}{2B}\sum_{\substack{x \in \Lambda_N \\x \leq \epsilon N}}\sum_{y =1}^{x-1}(f(\eta)-f(\sigma^{y,y+1}\eta))^2d\nu_{\rho}(\eta).
	\end{equation}
	for some arbitrary $B>0$. Notice now that for \eqref{2lemmaleft} the following bounds hold
	\begin{equation}	\label{way}
	\begin{split}
	\int \textcolor{black}{\frac{1}{\epsilon \sqrt{N-1}}\frac{1}{2B}}\sum_{\substack{x \in \Lambda_N \\x \leq \epsilon N}}\sum_{y =1}^{x-1}(f(\eta)-f(\sigma^{y,y+1}\eta))^2d\nu_{\rho}(\eta)&\leq \frac{\sqrt{N}}{2B}D^{NN}(f,\nu_{\rho})\\&\leq \textcolor{black}{c_0}\frac{\sqrt{N}}{2B}D^N_0(f,\nu_{\rho}),
	\end{split}
	\end{equation}
	where $D^{NN}$ is $D_0^N$ for $p(1)=p(-1)=1/2$ and \textcolor{black}{$c_0$ is a positive constant which depends only on $p(\cdot)$. } Note that the notation $D^{NN}$ comes from the fact that it is the Dirichlet form associated to the simple symmetric exclusion process. Now, we just have to choose $B$ is such a way $\textcolor{black}{c_0}\frac{\sqrt{N}}{2B}D^N_0(f,\nu_{\rho})-\Theta(N)\langle L_Nf,f\rangle_{\nu_{\rho}} \leq0$. This is possible  by using \textcolor{black}{\eqref{diiirest}} and  taking 
	\textcolor{black}{$B=\frac{c_0\sqrt{N}}{2\Theta(N)}$} . Then, the only term that remains to bound is \eqref{1lemmaleft} which, with this choice of $B$, is bounded by a term of order
	$\frac{\epsilon N^2}{\Theta(N)}$.
	This means that, for $\theta \geq 2-\gamma$, the global order of the term in the statement is $\epsilon$.
\end{proof}

\begin{lem}
	For any $t \in [0,T]$ and $\theta \geq 2-\gamma$:
	\begin{equation}
	\limsup_{N \rightarrow \infty}\mathbb{E}_{\nu_{\rho}}\Bigg[ \sup_{t \in [0,T]} \Big(\sqrt{N}\int_0^t\sum_{x \in \Lambda_N}r^{\pm}_N\left(\tfrac{x}{N}\right)(\overline{\eta}_s^N(x)-\overline{\eta}_s^N(1))ds\Big)^2\Bigg]=0;
	\end{equation}
	\begin{equation}
	\limsup_{N \rightarrow \infty}\mathbb{E}_{\nu_{\rho}}\Bigg[ \sup_{t \in [0,T]}\Big(\sqrt{N}\int_0^t\sum_{x \in \Lambda_N}\Theta^{\pm}_x(\overline{\eta}_s^N(x)-\overline{\eta}_s^N(1))ds\Big)^2\Bigg]=0.
	\end{equation}
	\label{lemmaPat1}
\end{lem}
\begin{proof}
	The proof is similar to the proof of Lemma \ref{left2}. We show how to proceed to prove the first limit in the statement. By Lemma 4.3 of \cite{changlandimolla}, we can bound the expectation by a constant times
	\begin{equation}
	\begin{split}
	T\sup_{f \in L^2(\nu_{\rho})}& \bigg\{\int \sqrt{N}\sum_{x \in \Lambda_N}r^-_N\left(\tfrac{x}{N}\right)\sum_{y=1}^{x-1}\Big[\frac{1}{B_y}(f(\eta)-f(\sigma^{y,y+1}\eta))^2+B_{y}\Big]d\nu_{\rho}(\eta)\\
	& -N^2\langle L_Nf,f\rangle_{\nu_{\rho}} \textcolor{black}{\bigg\},}
	\label{rem}
	\end{split}
	\end{equation}
\textcolor{black}{	where $B_y$ are arbitrary positive constants.}
	Then observe that using Fubini's theorem we can write
	\begin{equation*}\begin{split}
	&\sum_{x \in \Lambda_N}r^-_N\left(\tfrac{x}{N}\right)\sum_{y=1}^{x-1}\frac{1}{B_y}(f(\eta)-f(\sigma^{y,y+1}\eta))^2\\
	&=\sum_{y \in \Lambda_N}\frac{1}{B_y}(f(\eta)-f(\sigma^{y,y+1}\eta))^2\sum_{x\geq y}r^-_N\left(\tfrac{x}{N}\right)
	\end{split}
	\end{equation*}
	and, since $\sum_{x \geq y} r^-_N(x/N)$ is of order $y^{-\gamma+1}$, we can take $B_y=y^{-\gamma+1}N^{-3/2}$ and, reasoning in a similar way to what we did in  \eqref{way}, this term vanishes with the Dirichlet form $N^2\langle L_Nf,f \rangle_{\nu_{\rho}} $. So, we still have to analyze the remaining term in \eqref{rem}, which, by using Fubini's theorem and with this choice of $B_y$, is of order
	\begin{equation}
	\frac{1}{N}\sum_{y \in \Lambda_N}y^{-\gamma+1}\sum_{x\geq y} x^{-\gamma}\lesssim \frac{1}{N}\sum_{y \in \Lambda_N}y^{-2\gamma+2}
	\label{end}
	\end{equation}
	and, since the sum is convergent, the limit for $N \rightarrow \infty$ is equal to $0$. The same procedure works for $r^+_N$.  In order to prove the second limit of the statement we can proceed exactly in the same way, just noticing that the order of the term $\sum_{x\geq y}\Theta_x^{\pm}$ is $y^{-\gamma+2}$ and so we have to choose $B_y=y^{-\gamma+2}N^{-3/2}$. In this way,  at the end (corresponding to the bound \eqref{end} in the case with $r^-_N$), we get the following bound for the expectation
	\begin{equation*}
	\frac{1}{N}\sum_{y \in \Lambda_N}y^{-2\gamma+4}.
	\end{equation*}
	Now, if $\gamma >3$ the sum is converging and so the term goes to $0$ as $N$ goes to infinity; if $\gamma=3$ the sum is of order $\log (N)$ and so again in the limit, as $N$ goes to infinity, the previous bound goes to $0$; if $\gamma \in (2,3)$ the sum is of order $N^{5-2\gamma}$ and so the order of the whole term is $N^{4-2\gamma}$ which again goes to $0$, as $N$ goes to infinity.
\end{proof}

\section{Uniqueness of $OU({\mc C}, {\mc A}, c)$}
\label{Subsec:UOU}

In this section we prove Proposition \ref{prop:unique}. 

Fix $H \in \mc C$ and $s>0$. Then, using It\^o's formula (see \cite{Martingales}, Theorem 3.3 and Proposition 3.4), we know that the process $\{X_t^s(H)\; ; \;  s\le t \le T\}$ defined by
\begin{equation}
\begin{split}
X_t^s(H)& =\exp\left\{ i (M_t(H) -M_s (H)) + \frac{1}{2} \Big[\langle M(H)\rangle_t -\langle M (H)\rangle_s \Big]\right\}\\
&=\exp\left\{i\left(\mathcal{Y}_t(H)-\mathcal{Y}_s(H) -  \int_s^t \mathcal{Y}_r ({\mc A} H) \, dr\right)  + \frac{t-s}{2} c^2 (H) \right\}
\end{split}
\end{equation}
is a complex martingale with continuous trajectories.

\begin{lem}\label{lem1ced}
	\textcolor{black}{For any $S\leq T$, the process $\{Z_t^S\; ; \; 0 \leq t \leq S\}$ defined by
	\begin{equation}
	Z^S_t=\exp\left\{ \tfrac{1}{2}\,\textcolor{black}{\int_0^t  c^2 ({P}_{S-s}H )ds}+ i \mathcal{Y}_t (P_{S-t}H) \right\}
	\end{equation}}
	is a complex martingale with continuous trajectories. 
\end{lem}

\begin{proof}
	To prove it consider two times $0\le s < s+\delta <S$ and for each $n\ge 1$ a partition of the interval $[s, s+\delta]$ with mesh $\delta/n$:
	$$s=s_0<s_1<\ldots<s_n=s+\delta, \quad s_{j+1}-s_j =\delta/n.$$ 
	We have
	\begin{equation*}
	\begin{split}
	&\prod_{j=0}^{n-1} X_{s_{j+1}}^{s_j} (P_{{S-s_j}} H) \\
	&=\exp \left\{ i \sum_{j=0}^{n-1} \left({\mc Y}_{s_{j+1}} (P_{S-s_j} H) -{\mc Y}_{s_j} (P_{S-s_j} H) - \int_{s_j}^{s_{j+1}} {\mc Y}_r ({\mc A} P_{S-s_j} H) dr \right) \right. \\
	& \hspace{2cm}+\left.\cfrac{\delta}{2n} \sum_{j=0}^{n-1} c^2 (P_{S-s_j} H) \right\}.
	\end{split}
	\end{equation*}
	Observe that the second sum appearing in the exponential is a Riemann sum involving the function $r\in [S-(s+\delta), S-s] \to c^2 (P_{r} H) \in \RR$, which is a continuous function because it is the composition of the two continuous functions $c^2: {\mc C} \to [0,\infty)$ and $r\in [0, \infty) \to P_r H \in {\mc C}$ (see the comment after Proposition \ref{prop:unique}).  Thus, this sum converges to $\tfrac12 \int_{s}^{s+\delta} c^2 (P_{S-r} H) \, dr$.
	
	By making a telescopic sum appear, the first sum can be rewritten as 
	\begin{equation*}
	\begin{split}
	&{\mc Y}_{s+\delta} (P_{S-(s+\delta)} H) -{\mc Y}_{s} (P_{S-s} H)\\
	& + \sum_{j=0}^{n-1} \Big[ {\mc Y}_{s_{j+1}} (P_{S-{s_{j}}} H) -{\mc Y}_{s_{j+1}} (P_{S-s_{j+1}} H) -\int_{s_{j}}^{s_{j+1}} {\mc Y}_r ({\mc A} P_{S-s_{j}} H) dr \Big]\\
	&={\mc Y}_{s+\delta} (P_{S-(s+\delta)} H) -{\mc Y}_{s} (P_{S-s} H)\\
	&+\sum_{j=0}^{n-1} \Big[ {\mc Y}_{s_{j+1}} (P_{\delta/n} H_j) -{\mc Y}_{s_{j+1}} (H_j) -\int_{s_{j}}^{s_{j+1}} {\mc Y}_r ({\mc A} P_{\delta/n} H_j) dr \Big]
	\end{split}
	\end{equation*}
	where $H_j=P_{S-s_{j+1}}H \in {\mc C}_H:=\{P_r H\; ; \; r\in [0,S]\} \subset {\mc C}$. By Lemma \ref{lem:topolo}  proved below,  the previous sum goes to $0$ $\PP$-a.s., as $n$ goes to infinity. It follows that
	$$\lim_{n \to \infty} \prod_{j=0}^{n-1} X_{s_{j+1}}^{s_j} (P_{{S-s_j}} H)  = \cfrac{Z^S_{s+\delta}}{Z^S_s}$$
	$\PP$-a.s. and, since $\alpha \in \RR \to e^{i\alpha}$ is a bounded function, the dominated convergence theorem implies that the previous convergence holds, in fact, in $L^1$. We conclude that if $U$ is a bounded ${\mc F}_s$-measurable random variable then   
	$$\EE \left[U\cfrac{Z^S_{s+\delta}}{Z^S_s}  \right] = \lim_{n \to \infty} \EE \left[ U \prod_{j=0}^{n-1} X_{s_{j+1}}^{s_j} (P_{{S-s_j}} H) \right].$$
	By using iteratively the martingale property of the processes $X_\cdot^{u} (G)$, $u\in [0,S]$, for $G\in {\mc C}$, we establish easily that
	$$\EE \left[ U \prod_{j=0}^{n-1} X_{s_{j+1}}^{s_j} (P_{{S-s_j}} H) \right]=\EE [U]$$
	so that  $\EE \left[U\tfrac{Z^S_{s+\delta}}{Z^S_s}  \right] =\EE \left[ U \right]$ and $\{Z^S_t \; ; \; t\in [0,S]\}$ is a martingale.
\end{proof}

If \textcolor{black}{$0\le s \le t \le S\le T$}, we can conclude from the previous lemma that $\mathbb{E}[Z^S_t|\mathcal{F}_s]=Z^S_s$, or equivalently
\begin{equation*}
\begin{split}
&\mathbb{E}\bigg[\exp\left\{ \tfrac{1}{2}\,\textcolor{black}{\int_0^t  c^2 ({P}_{S-r}H )dr}+ i \mathcal{Y}_t(P_{S-t}H)\right\}\bigg|\mathcal{F}_s\bigg]\\
=&\exp\left\{ \tfrac{1}{2}\,\textcolor{black}{\int_0^s  c^2 ({P}_{S-r}H )dr}+ i \mathcal{Y}_s(P_{S-s}H) \right\}.
\end{split}
\end{equation*}
This implies that
\begin{equation*}
\begin{split}
&\mathbb{E}\big[\exp\left\{ i \mathcal{Y}_t(P_{S-t}H)\right\}|\mathcal{F}_s\big]\\
=&\exp\left\{ -\tfrac{1}{2}\,\textcolor{black}{\int_s^t  c^2 ({P}_{S-r}H )dr} + i \mathcal{Y}_s(P_{S-s}H) \right\}.
\end{split}
\end{equation*}

Now, \textcolor{black}{applying the previous equality to $t=S$}, we get
\textcolor{black}{\begin{equation}
\label{eq:cons-fin}
\mathbb{E}\big[\exp\left\{ i \mathcal{Y}_S(H)\right\}|\mathcal{F}_s\big]=\exp\left\{ -\tfrac{1}{2}\int_s^S c^2 (P_{S-r}H)dr  + i \mathcal{Y}_s(P_{S-s}H) \right\}
\end{equation} and replacing $H$ by $\lambda H$ for $\lambda \in \mathbb{R}$, we get, by using \eqref{hypothesis_c}
\begin{equation}
\label{conditioning}
\mathbb{E}\big[\exp\left\{ i \lambda \mathcal{Y}_S(H)\right\}|\mathcal{F}_s\big]=\exp\left\{ -\tfrac{\lambda^2}{2}\int_s^S c^2 (P_{S-r}H)dr  + i \lambda\mathcal{Y}_s(P_{S-s}H) \right\}
\end{equation} 
which means that, conditionally to $\mc F_s$ , the random variable $\mathcal{Y}_S(H)$ is Gaussian distributed with mean $\mathcal{Y}_s(P_{S-s}H)$ and variance  $\int_s^S c^2 (P_{S-r}H )dr$. } Since the distribution at initial time is given, by iterating the conditioning on \eqref{conditioning} we get the uniqueness of the finite dimensional distribution of $\{\mathcal{Y}_t(H),t \in [0,T]\}$, which grants the uniqueness in law of $\mathcal{Y}_{\cdot}$.

\begin{lem}
	\label{lem:Guss}
	Assume that the conditions of Proposition \ref{prop:unique} are satisfied and assume that the initial condition of $OU(\mc C, \mc A, c)$ is given by a centered Gaussian field with covariance function ${\mf C}$. Then, the solution is a Gaussian process and the covariance function of the process is given by \eqref{eq:covpat}.
\end{lem}

\begin{proof}
	The proof of this lemma is a direct consequence of \eqref{eq:cons-fin}.
\end{proof}

\begin{lem}
	\label{lem:topolo}
	With the notations introduced above (see proof of Lemma \ref{lem1ced}), we have that $\PP$-a.s., the sum
	$$\sum_{j=0}^{n-1} \Big[ {\mc Y}_{s_{j+1}} (P_{\delta/n} H_j) -{\mc Y}_{s_{j+1}} (H_j) -\int_{s_{j}}^{s_{j+1}} {\mc Y}_r ({\mc A} P_{\delta/n} H_j) dr \Big]$$
	converges to $0$, as $n$ goes to infinity.
\end{lem}

\begin{proof}
	For any $G\in {\mc C}$ denote 
	\begin{equation*}
	{r}_G^{\epsilon} =\epsilon^{-1} \Big[ P_\epsilon G -G -\epsilon {\mc A} G\Big]  \in{\mc C}.
	\end{equation*}
	The uniformity in time in \eqref{eq:exp-OU} implies that the sequence $\{r_G^{\epsilon}\}_{\epsilon>0}$ converges to $0$ in $\mc C$ uniformly over $G\in {\mc C}_H$, i.e. for any open neighborhood $V_0^{\mc C}$ of $0$ in ${\mc C}$ there exists $\epsilon_0>0$ such that for any $0<\varepsilon <\varepsilon_0$, for any $G\in {\mc C}_H$, we have that ${r}_G^{\epsilon} \in V_{0}^{\mc C}$.    
	
	We rewrite the sum to estimate as (recall that the semigroup $P_r$ commutes with ${\mc A}$ on ${\mc C}$)
	\begin{equation*}
	\begin{split}
	& \sum_{j=0}^{n-1}\int_{s_{j}}^{s_{j+1}} ({\mc Y}_{s_{j+1}}- {\mc Y}_r )({\mc A} H_j) dr-\cfrac{\delta}{n} \, \sum_{j=0}^{n-1}\int_{s_{j}}^{s_{j+1}} {\mc Y}_r  ({\mc A}^2 H_j) dr \\&-\cfrac{\delta}{n} \, \sum_{j=0}^{n-1}\int_{s_{j}}^{s_{j+1}} {\mc Y}_r  (r_{\mc A H_j}^{\delta/n} ) dr
	+\cfrac{\delta}{n}\, \sum_{j=0}^{n-1} {\mc Y}_{s_{j+1}} (r^{\delta/n}_{H_j}).
	\end{split}
	\end{equation*}
	Trivial inequalities show that it is sufficient to prove
	\begin{equation*}
	\begin{split}
	&\PP \text{-a.s.}\; \lim_{n \to \infty} \sup_{| r-r'| \le {\delta}/{n}} \sup_{G \in {\mc C}_{{\mc A} H}} \vert ({\mc Y}_r -{\mc Y}_{r'}) (G)\vert=0, \\
	&\PP \text{-a.s.}\; \sup_{0\le r\le T} \sup_{G\in {\mc C}_{\mc A^2 H}} \vert {\mc Y}_r (G)\vert < \infty,\\
	& \PP \text{-a.s.}\; \limsup_{n\to \infty} \sup_{0\le r\le T} \sup_{G\in {\mc C}_{\mc A H}} \vert {\mc Y}_r (r_G^{\delta/n})\vert <\infty,\\
	& \PP \text{-a.s.}\; \lim_{n \to \infty}\, \sup_{0\le r\le T} \sup_{G\in {\mc C}_{H}} \vert {\mc Y}_r (r_G^{\delta/n})\vert=0. 
	\end{split}
	\end{equation*}
	Let us observe that, since for any $G\in {\mc C}$, the map $t\in [0,T] \to P_t G \in {\mc C}$ is continuous and $[0,T]$ is compact, the set ${\mc C}_G$ is compact in $\mc C$. The third one is a direct consequence of the fourth one by replacing $H$ by ${\mc A} H$. The second one follows from the fact that $\Phi:(r, G) \in [0,T]\times {\mc C} \to {\mc Y}_r (G)$ is continuous and that $[0,T]\times {\mc C}_{{\mc A}^2 H}$ is a compact set. For the first one, observe that since ${\mc C}_{\mc A H}$ is compact, the restriction of $\Phi$ to $[0,T]\times {\mc C}_{\mc A H}$ is uniformly continuous, i.e. for any $\epsilon>0$ there exists $\alpha>0$ and an open neighborhood $V_0^{\mc C}$ of $0$ in ${\mc C}$ such that for any $r, r' \in [0,T]$ such that $| r-r'| < \alpha$ and any $F,G \in {\mc C}_{\mc A H} $ such that $F-G \in V_0^{\mc C}$, 
	$$\vert {\mc Y}_r (F) -{\mc Y}_{r'} (G) \vert < \epsilon.$$
	For the last assertion we observe that since $\{r_G^{\delta/n}\}_{n\ge1}$ converges to $0$ uniformly on ${\mathcal C}_H$, we have that the set 
	$${\mc K}:=\{  r_G^{\delta/n} \; ; \;  G \in {\mc C}_H, \quad n\ge 1\} \cup \{ 0\}$$ 
	is a compact subset of ${\mc C}$. The restriction of $\Phi$ to $[0,T]\times {\mc K}$ is therefore uniformly continuous, i.e. for any $\epsilon>0$ there exists $\alpha>0$ and an open neighborhood $V_0^{\mc C}$ of $0$ in ${\mc C}$ such that for any $r, r' \in [0,T]$ such that $| r-r'| < \alpha$ and any $p,q \in {\mc K}$ such that $p-q \in V_0^{\mc C}$, 
	$$\vert {\mc Y}_r (p) -{\mc Y}_{r'} (q) \vert < \epsilon.$$
	Fix $\epsilon>0$ and $\alpha>0$, $V_0^{\mc C}$ as above. By the uniformity in time of the first order expansion, there exists $n_0\ge 1$ such that for any $n\ge n_0$, any $G\in {\mc C}_H$, $r_G^{\delta/n} \in V_{0}^{\mc C}$. Hence for $n\ge n_0$ and any $r,r' \in [0,T]$ such that  $| r-r'| <\alpha$, 
	$$|{\mc Y}_{r} (r_{G}^{\delta/n})| =| {\mc Y}_r (r_G^{\delta/n}) -{\mc Y}_{r'}(0)| \le \epsilon.$$ 
	Given $r \in [0,T]$ we can always find $r'\in [0,T]$ such that $|r-r'| <\alpha$. Therefore, the last inequality proves the fourth assertion. 
	
\end{proof}

\appendix
\section{}

\subsection{Discrete versions of continuous functions}
Here we prove some results which permits us to pass from discrete functions or operators to continuous ones.

\begin{lem}
	Recall the definition of the operator $K_N$ given in \eqref{KN}. Let $H:\RR \rightarrow \RR$ be a two times continuously differentiable bounded function and $\gamma>2$ then
	\begin{equation*}
\textcolor{black}{	\lim_{N\rightarrow \infty}} \sup_{x\in\Lambda_N}\left|N^2(K_NH)(\tfrac{x}{N})-\tfrac{\sigma^2}{2}\Delta H(\tfrac{x}{N})\right|=0
	\label{lemmasigma}
	\end{equation*}
\end{lem}

\begin{proof} 
	In \textcolor{black}{Lemma 3.2 of} \cite{BGJO} the authors proved this lemma for test functions two times continuously differentiable and with compact support, but actually we just need the functions to be $C^2$ and uniformly bounded therefore, we leave the details to the reader. 
\end{proof}

\begin{lem}
	For any function $H \in \mathcal{S}$ we have that 
	\begin{equation}
	\limsup_{N \rightarrow \infty}\frac{1}{{N}}\sum_{x \in \Lambda_N}H(\tfrac{x}{N})^2(N^{\gamma}r^-_N(\tfrac{x}{N})-r^-(\tfrac{x}{N}))^2=0.
	\label{statement}
	\end{equation}
	\label{reaction}
	The statement above is also true replacing $r_N^-$ by $r_N^+$ and $r^-$ by $r^+$.
\end{lem}
\begin{proof}
	We can split the sum over $x$ of the statement in the cases $aN\leq x\leq bN$ and $0<x<aN$ and $bN<x<N $ for $0<a<b<1$. In the first case we can perform exactly the same proof as in Lemma 3.3 of \cite{BJ} because, in that case, the uniform convergence of $N^\gamma r_N^-$ to $r^-$ holds. So, it remains to  treat the  remaining cases. We present the proof in the case $0<x<aN$ but the other case ($bN<x<N$) is analogous. 	
	At this point we have to bound
	\begin{equation}
	\frac{1}{{N}}\sum_{\substack{x \in \Lambda_N\\x< aN}}H(\tfrac{x}{N})^2(N^{\gamma}r^-_N(\tfrac{x}{N})-r^-(\tfrac{x}{N}))^2.
	\end{equation}
	From (B.1) in \cite{BJ} we see that
	$
	\Big|N^{\gamma}r^-_N(\tfrac{x}{N})-r^-(\tfrac{x}{N})\Big|\leq \tfrac {c_\gamma}{N}(\tfrac{x}{N})^{-1-\gamma}.
	$
	Using the fact that the  test functions  are in $\mathcal S$ and applying Taylor expansion of $H$ around the point $0$ up to order $d\geq 1$ such that $2d-2\gamma>1$, we  get that
	\begin{equation}
	\begin{split}
	& \frac{1}{{N}}\sum_{\substack{x \in \Lambda_N\\x< aN}}H(\tfrac{x}{N})^2(N^{\gamma}r^-_N(\tfrac{x}{N})-r^-(\tfrac{x}{N}))^2\\
	&\leq 
	N^{2(\gamma-d)-1}\sum_{\substack{x \in \Lambda_N\\x< aN}}x^{-2\gamma + 2d-2}\\
	&\leq N^{2(\gamma-d)-1} (aN)^{2d-1-2\gamma}=a^{2d-1-2\gamma}N^{-2},
	\end{split}
	\end{equation}
	which  goes to $0$, as $N$ goes to infinity, choosing $d$ such that $2\gamma-2d+1<0 $ and the sum diverges (otherwise the last inequality does not hold).  The proof for the case involving $r^+_N$ is similar.
\end{proof}

\textcolor{black}{Now we recall Lemma 3.3 of \cite{BJ}.}

\begin{lem}
	Let $\gamma >0$ and $a \in (0,1)$. Then we have the following uniform convergence on $[a,1-a]$:
	$$\lim_{N\rightarrow \infty}N^\gamma r^\pm_N([Nu])=r^\pm(u)$$
	where  $r^+(u)=c_{\gamma} \gamma^{-1}(1-u)^{-\gamma}$ and $r^-(u)=c_{\gamma} \gamma^{-1}u^{-\gamma}$.
	\label{lemmalim}
\end{lem}

\subsection{An approximation lemma}
\label{subsec:proflemmachiant}

In this subsection we prove an approximation lemma that we used in the proof of Lemma \ref{lema1}.

\begin{lem}
	\label{lem:appr}
	Let $H \in \tilde {\mc S}$.  There exists a sequence of functions $\{H_\epsilon\}_{\epsilon>0}$ in ${\mathcal S}$ such that $\lim_{\epsilon \to 0} H^{(k)}_{\epsilon} = H^{(k)}$ for $k=0,1,2$ in $L^2$, i.e. $\{H_\epsilon\}_{\epsilon>0}$ converges to $H$ in ${\mc H}^2 ([0,1])$.
\end{lem}

\begin{proof}

	Recall \eqref{eq:phi-func}. For any constant $\epsilon \in (0,1)$ let us define $\tilde \phi_{\epsilon}(u):= \tilde\phi(u/\epsilon)$ and $\hat \phi_{\epsilon}(u):=\phi_{\epsilon}(u-(1-\epsilon))$ where $\phi_{\epsilon}(u):=\phi(u/\epsilon)$.

	Let us define for any $u \in [0,1]$
	\begin{equation}\label{PHI}
	\Phi_{\epsilon}(u)=\begin{cases}
	\tilde{\phi}_{\epsilon}(u) & \text{ if } u\in [0,\epsilon],\\
	\hat{\phi}_{\epsilon}(u) & \text{ if } u \in [1-\epsilon,1],\\
	1 & \text{ if } u\in (\epsilon,1-\epsilon).\\
	\end{cases}
	\end{equation}
	Thanks to the particular form of the function $a$, it is possible to see that $\Phi_{\epsilon} \in \mc S$,  see Figure \ref{approxfig}.
	
	\begin{figure}[h]
		\begin{center}
			\includegraphics[scale=0.45]{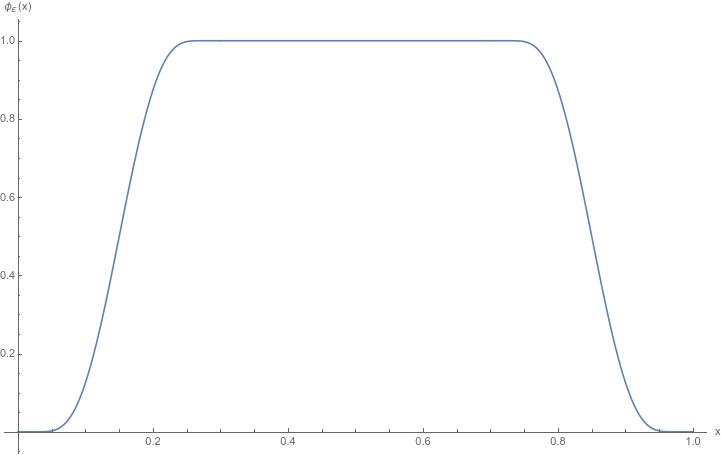}
			\caption{Plot of the function $\Phi_{\epsilon}$ defined in \eqref{PHI} for the particular case $\epsilon=0.3$.}
			\label{approxfig}
		\end{center}
	\end{figure}
	
	Therefore, for any $H \in \tilde{\mc S}$ the product $H_{\epsilon}:=H\Phi_{\epsilon}$ is in $\mc S$, which is again easy to see thanks to the form of $\Phi_{\epsilon}$. Moreover, $H_{\epsilon} \rightarrow H$ uniformly and in particular in $L^2$, as $\epsilon \rightarrow 0$, since $\Phi_{\epsilon}\rightarrow 1$ in $[0,1]$, as $\epsilon \rightarrow 0$. 
	
	Let us now check that the following convergences hold: $H'_{\epsilon} \rightarrow H'$ and $H''_{\epsilon}\rightarrow H''$ in $L^2$, as $\epsilon \rightarrow 0$. Since $H-H_{\epsilon}=H(1- \Phi_{\epsilon})$, in order to prove this statement we need to prove that, for $\epsilon \rightarrow 0$, the quantities $(H(1-\Phi_{\epsilon}))'$ and $(H(1-\Phi_{\epsilon}))''$ go to $0$ in $L^2$.
	\begin{itemize}
		\item[(i)] First notice that $ (H(1-\Phi_{\epsilon}))'(u)= H'(u)(1-\Phi_{\epsilon}(u)) - H(u) \Phi_{\epsilon}'(u)$, where we denoted by $\Phi_{\epsilon}'(u)$ the derivative of $\Phi_{\epsilon}$ with respect to $u$ which is 
		\begin{equation}\label{der}
		\Phi'_{\epsilon}(u)=\begin{cases} \epsilon^{-1}a(\tfrac{u}{\epsilon}) &\text{if } u \in [0,\epsilon];\\
		\epsilon^{-1}a(\tfrac{u-(1-\epsilon)}{\epsilon})&\text{if } u \in [1-\epsilon,1];\\
		0 & \text{otherwise}.
		\end{cases}
		\end{equation}
		Then, we can show that
		\begin{equation}
		\begin{split}
		||H'(1-\Phi_{\epsilon})||_{L^2}^2 &=\int_{[0,\epsilon]\cup[1-\epsilon,1]}H'(u)^2(1-\Phi_{\epsilon})^2dx \\&\lesssim ||(1-\Phi_{\epsilon})||^2_{\infty} ||H'||^2_{\infty} \epsilon \xrightarrow[]{\epsilon \rightarrow 0} 0.
		\end{split}
		\end{equation}
		Moreover, recalling \eqref{der}, we have
		\begin{equation}\label{firstorder}
		\begin{split}
		||H \Phi_{\epsilon}'||_{L^2}^2=&\epsilon^{-2}||a(\tfrac{\cdot}{\epsilon})||_{\infty}^2\bigg[\int_0^\epsilon H(u)^2 du+ \int_{1-\epsilon}^1H^2(u)du\bigg]\\&\lesssim \epsilon^{-2}\bigg[\int_0^\epsilon u^2 du+ \int_{1-\epsilon}^1(1-u)^2du\bigg] \lesssim \epsilon \xrightarrow[]{\epsilon \rightarrow 0} 0;
		\end{split}
		\end{equation}
		where the first inequality follows from a Taylor expansion of $H$ of first order around $0$ in the first integral and around $1$ in the second.
		
		\item[(ii)] the second derivative is $(H(1-\Phi_{\epsilon}))''(u)=H''(u)(1-\Phi_{\epsilon}(u)) - 2H'(u) \Phi'(u)-H(u) \Phi_{\epsilon}''(u)$; the first two terms vanish exactly as in the previous computations. Let us check that also the third one goes to $0$ in $L^2$. First, observe that 
		\begin{equation}\label{der2}
		\Phi''_{\epsilon}(u)=\begin{cases} \epsilon^{-2}a'(\tfrac{u}{\epsilon}) &\text{if } u \in [0,\epsilon];\\
		\epsilon^{-2}a'(\tfrac{u-(1-\epsilon)}{\epsilon})&\text{if } u \in [1-\epsilon,1];\\
		0 & \text{otherwise}.
		\end{cases}
		\end{equation}

		Therefore, reasoning in a similar way to what we did in \eqref{firstorder}, but with the Taylor expansion up to the second order, we get
		\begin{equation}
		\begin{split}
		||H \Phi''||_{L^2}^2
		\lesssim \epsilon^{-4}||a'(\tfrac{\cdot}{\epsilon})||_{\infty}^2||H''||^2_{\infty}\bigg[\int_0^\epsilon u^4 du+ \int_{1-\epsilon}^1(1-u)^4du\bigg]\lesssim \epsilon \xrightarrow[]{\epsilon \rightarrow 0} 0.
		\end{split}
		\end{equation}
	\end{itemize}
	
\end{proof}

\section*{Acknowledgements}
The authors are very grateful to Maxime Ingremeau for many discussions concerning Subsection \ref{susubsec:slpbm}.  P. G. and S. S. thank  FCT/Portugal for support through the project UID/MAT/04459/2013.  This project has received funding from the European Research Council (ERC) under  the European Union's Horizon 2020 research and innovative programme (grant agreement   No 715734). This work has also been supported by the project LSD ANR-15-CE40-0020- 01 of the French National Research Agency (ANR).

\nocite{*}
\bibliography{biblio.bib}

\end{document}